\documentclass{statsoc}

\usepackage[a4paper]{geometry}
\usepackage[textwidth=8em,textsize=small]{todonotes} 
\usepackage{enumerate,amsmath,amssymb,bbm,graphicx,mathrsfs,breqn,mathtools,enumitem,rotating,bm,txfonts,units,setspace,algorithmicx,natbib,caption,subcaption,subcaption,booktabs}

\DeclareMathOperator{\Tr}{Tr}
\DeclareMathOperator{\diag}{diag}
\DeclareMathOperator*{\argmin}{arg\,min}
\DeclareMathOperator*{\argmax}{arg\,max}

\DeclareMathOperator{\E}{\mathbb{E}}
\DeclareMathOperator{\V}{\text{Var}}
\DeclareMathOperator{\C}{\text{Cov}}
\DeclareMathOperator{\vecM}{\text{vec}}

\DeclareMathOperator{\edist}{\stackrel{\text{\scriptsize \textit{d}}}{=}}
\DeclareMathOperator{\im}{\text{Im}}
\DeclarePairedDelimiter\abs{\lvert}{\rvert}%
\DeclarePairedDelimiter\norm{\lVert}{\rVert}%

\makeatletter
\let\oldabs\abs
\def\abs{\@ifstar{\oldabs}{\oldabs*}}
\let\oldnorm\norm
\def\norm{\@ifstar{\oldnorm}{\oldnorm*}}
\makeatother

\newtheorem{assumption}{\textit{Assumption}}

\newtheorem{algorithm}{\textit{Algorithm}}
\newtheorem{proposition}{\textit{Proposition}}
\newtheorem{corollary}{\textit{Corollary}}
\newtheorem{theorem}{\textit{Theorem}}
\newtheorem{lemma}{\textit{Lemma}}

\newcommand{\Ltab}{
 \begin{minipage}{\linewidth}
   \centering
   \vspace{0.25cm}
   \captionof{table}{\textsl{The $\pi_k$ and $\eta_k$ values used to simulated $\bm{L}$ and the resulting average $\gamma_k$ ($k=1,\ldots,10$).}}\label{Table:L}
	\begin{tabular}{c | c | c | c | c | c | c | c | c | c | c}
 Factor number ($k$) & $1$ & $2$ & $3$ & $4$ & $5$ & $6$ & $7$ & $8$ & $9$ & $10$\\\hline
  $\pi_k$ & 0 & 0.45 & 0.60 & 0.71 & 0.79 & 0.85 & 0.90 & 0.92 & 0.94 & 0.95\\\hline
  $\eta_k$ & 1 & 0.4 & 0.4 & 0.4 & 0.4 & 0.4 & 0.4 & 0.4 & 0.4 & 0.5\\\hline
  $\gamma_k$ & 143 & 12.7 & 9.5 & 6.6 & 4.8 & 3.3 & 2.5 & 1.8 & 1.4 & 1.4
  \end{tabular}
  \vspace{0.4cm}
  \end{minipage}
}

\newcommand{\TissueCorr}{
	\begin{minipage}{\linewidth}
    \centering
    \vspace{0.25cm}
	\captionof{table}{\textsl{The correlation matrix corresponding to the covariance matrix $p^{-1}\sum\limits_{g=1}^p \bm{M}_g$.}}\label{Table:Corr}
    \begin{tabular}{c | c | c | c}
    & Tissue 1 & Tissue 2 & Tissue 3 \\\hline
    Tissue 1 & 1 & 0.72 & 0.58 \\\hline
    Tissue 2 & 0.72 & 1 & 0.80 \\\hline
    Tissue 3 & 0.58 & 0.80 & 1
    \end{tabular}
    \vspace{0.5cm}
	\end{minipage}
}

\newcommand{\RealData}{
	\begin{minipage}{\linewidth}
    \centering
    \vspace{0.5cm}
	\captionof{table}{\textit{The fraction of sex-associated CpGs identified using data from \cite{Martino} that are also one of the 3,532 sex-associated CpGs confidently identified in \cite{BMCValidation} or \cite{NatureValidation}. }}\label{Table:RealData}
    \small
    \begin{tabular}{c | c | c | c | c}
     \shortstack{CBCV-CorrConf \\ $(K=2)$} & \shortstack{BCconf \\ $(K=4)$} & \shortstack{Cate-RR \\ $(K=4)$} & \shortstack{dSVA \\ $(K=15)$} & \shortstack{IRW-SVA \\ $(K=15)$} \\\hline
     39\% (283/735) & 23\% (424/1836) & 20\% (474/2404) & 19\% (487/2513) & 28\% (341/1236) 
    \end{tabular}
    \vspace{0.5cm}
	\end{minipage}
}

\title[Latent factors in correlated data]{Estimating and accounting for unobserved covariates in high dimensional correlated data}
\author[Chris McKennan {\it et al.}]{Chris McKennan}
\address{Department of Statistics, University of Chicago,
Chicago, IL,
USA.}
\email{cgm29@uchicago.edu}
\author[Chris McKennan]{Dan Nicolae}
\address{Department of Statistics, University of Chicago,
Chicago, IL,
USA.}

\begin{document}

\begingroup
\allowdisplaybreaks  
\begin{abstract}
Many high dimensional and high-throughput biological datasets have complex sample correlation structures, which include longitudinal and multiple tissue data, as well as data with multiple treatment conditions or related individuals. These data, as well as nearly all high-throughput `omic' data, are influenced by technical and biological factors unknown to the researcher, which, if unaccounted for, can severely obfuscate estimation and inference on effects due to the known covariate of interest. We therefore developed CBCV and CorrConf: provably accurate and computationally efficient methods to choose the number of and estimate latent confounding factors present in high dimensional data with correlated or nonexchangeable residuals. We demonstrate each method's superior performance compared to other state of the art methods by analyzing simulated multi-tissue gene expression data and identifying sex-associated DNA methylation sites in a real, longitudinal twin study. As far as we are aware, these are the first methods to estimate the number of and correct for latent confounding factors in data with correlated or nonexchangeable residuals. An R-package is available for download at https://github.com/chrismckennan/CorrConf.
\end{abstract}
\keywords{Latent factors, batch effects, unwanted variation, correlation, multi-tissue, cell-type heterogeneity}

\section{Introduction}
The development of high-throughput biotechnologies has provided biologists with a cornucopia of genetic, proteomic and metabolomic data that can been used to elucidate the genetic components of many phenotypes and the mediation of environmental exposures. Many of these `omic' data have complex sample-correlation structure, which includes longitudinal data \citep{LongRNAseq,LongDNAmeth,URECA}, multi-tissue data \citep{GTEX,MultTissue}, data with multiple treatment conditions \citep{Knowles} and data with related individuals \citep{Martino,KinshipBaboons}. Longitudinal studies have even been cited as a critical area of future DNA methylation research to assess the stability of methylation marks over time \citep{SmallEffectsLong,EnvirLong}, so it is crucial that suitable methods exist to analyze these data. An important feature of these high-throughput experimental data is the presence of unmeasured factors, which include technical factors like batch variables and biological factors like cell composition, that influence the measured response \citep{LeekBatch,Cell}. When unaccounted for, these factors can bias test statistics, reduce power and lead to irreproducible results \citep{BatchReproducible2,BatchReproducible1}. There have been a number of methods developed by the statistical community to estimate and correct for latent factors \citep{BujaFA,SVA2008,RUV,LEAPP,RUV4,Houseman,bcv,Fan1,BiometrikaConfounding,CATE,ChrisANDDan}. However, all of these methods make the critical assumption that conditional on both the observed and unobserved covariates, samples are independent with homogeneous residual variances, and tend to perform poorly when these assumptions are violated. The goal of this paper is therefore to provide a provably accurate method to both choose the number of latent factors and estimate them from the measured correlated data precisely enough so that downstream inference on the covariate of interest is as accurate as when the latent variables are observed. To the best of our knowledge, this is the first method to estimate and correct for latent covariates in correlated data.\par
\indent We use DNA methylation quantified in $p\approx 8 \times 10^5$ methylation sites (i.e. CpGs) in 183 unrelated individuals at birth and age 7 ($n=2\times 183$) from \cite{URECA} as a motivating data example, although we analyze other methylation data with a more complex covariance structure in Section \ref{section:RealData}. The aim is to jointly model methylation at birth and age 7 to determine if the effects due to ancestry on methylation levels changed or remained constant over time. If we ignore observed nuisance variables like the intercept, a reasonable model for the methylation at birth and age 7 at CpG $g$ ($\bm{y}_{g}^{(0)} \in \mathbb{R}^{n/2}$ and $\bm{y}_{g}^{(7)} \in \mathbb{R}^{n/2}$) in the presence of unobserved covariates $\bm{C} \in \mathbb{R}^{n \times K}$ ($K << n$, an unknown constant) is
\begin{subequations}
\label{equation:IntroModel}
\begin{align}
\bm{y}_g &= \left( \bm{y}_{g}^{(0)} \, \bm{y}_{g}^{(7)}\right)^T = \bm{X}\left( \beta_g^{(0)}\, \beta_g^{(7)}\right)^T + \bm{C}\bm{\ell}_g + \bm{e}_g, \quad \bm{e}_g \sim N\left(0,\bm{V}_g\right) \quad (g=1,\ldots,p)\\
\bm{V}_g &= \phi^2_g \bm{B}_1 + \sigma_{g,0}^2 \bm{B}_2 + \sigma_{g,7}^2 \bm{B}_3 \quad (g=1,\ldots,p)
\end{align}
\end{subequations}
where $\bm{X} = \bm{A}\oplus \bm{A} \in \mathbb{R}^{n \times 2}$ and $\bm{A} \in \mathbb{R}^{n/2}$ gives each individual's ancestry, $\bm{B}_1 \in \mathbb{R}^{n \times n}$ is a partition matrix that groups samples by individuals and $\bm{B}_2 \in \mathbb{R}^{n \times n}$ and $\bm{B}_3 \in \mathbb{R}^{n \times n}$ are diagonal matrices with ones in the first $n/2$ and second $n/2$ diagonal entries, respectively, and zeros everywhere else that capture the potential difference in residual variance at birth and age 7. In order to avoid overestimating the residual variance and biasing our estimates for the effects of interest $\beta_g^{(0)}$ and $\beta_g^{(7)}$, we must first estimate $\bm{C}$. Instead of effectively reducing the sample size by 50\% and estimating $\bm{C}$ separately at birth and age 7, which risks underestimating $K$ and subsequently biasing estimates for $\beta_g^{(0)}$ and $\beta_g^{(7)}$, we would ideally estimate $K$ and subsequently $\bm{C}$ using all $n$ samples, since we expect technical and many biological factors (including cell composition) to affect methylation at both ages. However, naively applying standard methods to choose $K$ designed for data with independent and identically distributed residuals, like parallel analysis \citep{BujaFA} and bi-cross validation \citep{bcv}, will typically overestimate the latent factor dimension to be on the order of the sample size, as they will be unable to distinguish the low dimensional factors $\bm{C}$ from the high dimensional random effect. Even if $K$ were known, the correlation among the residuals obfuscates current state of the art method's estimates for $\bm{C}$.\par
\indent Our proposed computationally efficient method uses all of the available data to estimate $K$ and $\bm{C}$ for data whose gene-, methylation-, protein- or metabolite-specific covariance can be written as a linear combination of known positive semi-definite matrices. As we will eventually show, this covariance structure is quite general and includes longitudinal, multi-tissue, multi-treatment and twin studies, as well as studies with individuals related through a kinship matrix. While our ultimate goal is to be able to do inference on the effects of interest that is as accurate as when $\bm{C}$ is observed, our method also provides a way to efficiently perform factor analysis in data with correlated or nonexchangeable residuals $\bm{e}_{g1}, \ldots, \bm{e}_{gn}$ ($g=1,\ldots,p$) and has application in quantitative train loci (QTL) studies \citep{Knowles,URECA}.\par
\indent We demonstrate the efficacy of our method by showing it selects the correct latent factor dimension with high probability, its estimate for the column space of $\bm{C}$ is nearly as accurate as when samples are independent (i.e. $\bm{V}_g$ in \eqref{equation:IntroModel} is a multiple of the identity) and that downstream inference on the effect of interest is asymptotically equivalent to inference with the generalized least squares estimator when $\bm{C}$ is known. We also simulate multi-tissue gene expression data with a complex, gene-dependent correlation structure to illustrate our method's superior performance in both choosing $K$ and estimating $\bm{C}$.\par
\indent The remainder of the paper is organized as follows. Section \ref{section:Setup} provides a description of and generative model for the data and an overview of our estimation procedure. Section \ref{section:Estimation} gives a detailed description of our estimation procedure, as well as all relevant theory. We then analyze simulated data in Section \ref{section:SimulatedData} and apply our method to a longitudinal DNA methylation dataset with measurements made on pairs of twins to identify CpGs with sex-dependent methylation levels in Section \ref{section:RealData}. An R package called CorrConf that implements our method to estimate $K$ and $\bm{C}$ is freely available with download instructions provided in Section \ref{section:Software}. The proofs of all propositions, lemmas and theorems can be found in the Supplement.

\section{Notation and problem set-up}
\label{section:Setup}
\subsection{Basic notation}
\label{subsection:Notation}
Define $\bm{1}_n \in \mathbb{R}^n$ be the vector of all ones and $I_n \in \mathbb{R}^{n \times n}$ to be the identity matrix. For any matrix $\bm{M} \in \mathbb{R}^{n \times m}$, we define $P_{M}$, $P_{M}^{\perp}$ to be the orthogonal projections matrices onto the image and orthogonal complement of $\bm{M}$. If $d > 0$ is the dimension of the null space of $\bm{M}^T$, we define $\bm{Q}_{M} \in \mathbb{R}^{n \times d}$ to be a matrix whose columns form an orthonormal basis for the null space of $\bm{M}^T$, i.e. $\bm{Q}_M^T \bm{M} = \bm{0}$, $\bm{Q}_M^T \bm{Q}_M = I_d$ and $\bm{Q}_M \bm{Q}_M^T = P_{M}^{\perp}$. Lastly, $\bm{X} \edist \bm{Y}$ if the two random variables, vectors or matrices $\bm{X}$ and $\bm{Y}$ have the same distribution.

\subsection{A description of and generative probability model for the data}
\label{subsection:Model}
Suppose we measure the expression or methylation $\bm{y}_g \in \mathbb{R}^n$ ($g=1,\ldots,p$) of $p$ units across $n$ potentially correlated samples and, when applicable, observe the covariates of interest $\bm{X} \in \mathbb{R}^{n \times d}$. Throughout this work, we will assume $\bm{X}$ is ``well-behaved'', i.e. the smallest and largest eigenvalues of $n^{-1}\bm{X}^T \bm{X}$ remains bounded above 0 and below a fixed constant. When unobserved factors $\bm{C} \in \mathbb{R}^{n \times K}$ influence the response, we assume the data $\bm{y}_g$ are generated as
\begin{subequations}
\label{equation:Model_yANDvar}
\begin{align}
\label{equation:Model_y}
\bm{y}_g &= \bm{X}\bm{\beta}_g + \bm{C}\bm{\ell}_g + \bm{e}_g, \quad \bm{e}_g \sim N\left( \bm{0}, \bm{V}_g\right)\,\, (g=1,\ldots,p)\\
\label{equation:Model_Y}
\bm{Y} &= \left[
\bm{y}_1 \, \cdots \, \bm{y}_p
\right]^T = \bm{\beta}\bm{X}^T + \bm{L}\bm{C}^T + \bm{E}\\
\label{equation:Model_var}
\bm{V}_g &= v_{g,1} \bm{B}_1 + \cdots + v_{g,b}\bm{B}_b \,\, (g=1,\ldots,p)
\end{align}
\end{subequations}
where the $g^{\text{th}}$ row of $\bm{\beta}$, $\bm{L}$ and $\bm{E}$ is $\bm{\beta}_g$, $\bm{\ell}_g$ and $\bm{e}_g$, respectively, $\bm{e}_g$ and $\bm{e}_{g'}$ are independent for $g \neq g'$, the variance multipliers $v_{g,j}$ ($g=1,\ldots,p$; $j=1,\ldots,b$) are unknown and $\bm{B}_j \in \mathbb{R}^{n \times n}$ ($j=1,\ldots,b$) are known and describe how the $n$ samples are related. Lastly, we assume that the variance multipliers $\bm{v}_g = \left(
v_{g,1} \, \cdots \, v_{g,b}
\right)^T$ lie in the convex polytope $\Theta$, defined as
\begin{align}
\label{equation:Theta}
\Theta = \left\lbrace \bm{x} \in \mathbb{R}^b : \bm{A}_{\mathcal{E}}\bm{x} = \bm{0}, \bm{A}_{\mathcal{I}}\bm{x} \geq \bm{0} \right\rbrace
\end{align}
where $\bm{A}_{\mathcal{E}} \in \mathbb{R}^{N_{\mathcal{E}} \times b}$ and $\bm{A}_{\mathcal{I}} \in \mathbb{R}^{N_{\mathcal{I}} \times b}$ are the equality and inequality constraints on $\bm{v}_g$. For example, we may know that certain multipliers must be larger than others, or that the sum of two sets of multipliers must be equal to ensure the marginal variances are the same. Depending on how we parametrize the covariance matrices $\bm{V}_g$, the multipliers can take negative values as long as $\bm{V}_g$ is interpretable as a covariance matrix. We lastly note that the form for parameter space $\Theta$ is sufficient for most problems because it assumes one only has prior information about the relative sizes of the variance multipliers (as well as their sign). It is unlikely for practitioner to know the absolute sizes of the multipliers.\par 
\indent In many applications there are other observed covariates $\bm{Z} \in \mathbb{R}^{n \times r}$ that may influence the response $\bm{y}_g$ but whose effects are not of interest. In that case, the model for $\bm{y}_g$ would be
\begin{align*}
\bm{y}_g = \bm{X}\bm{\beta}_g + \bm{Z}\bm{w}_g + \bm{C}\bm{\ell}_g + \bm{e}_g,
\end{align*}
and we can get back to model \eqref{equation:Model_yANDvar} by multiplying $\bm{y}_g$ by $\bm{Q}_{Z}^T$, provided $r$ does not grow with $n$:
\begin{align*}
\bm{Q}_{Z}^T\bm{y}_g &= \left(\bm{Q}_{Z}^T \bm{X}\right)\bm{\beta}_g + \left(\bm{Q}_{Z}^T \bm{C}\right)\bm{\ell}_g + \bar{\bm{e}}_g, \quad \bar{\bm{e}}_g \sim N\left( \bm{0}, \bar{\bm{V}}_g\right)\\
\bar{\bm{V}}_g &= v_{g,1} \left(\bm{Q}_{Z}^T\bm{B}_1 \bm{Q}_{Z}\right) + \cdots + v_{g,b} \left(\bm{Q}_{Z}^T \bm{B}_b \bm{Q}_{Z}\right).
\end{align*}
Therefore, we assume that the only observed covariates are contained in $\bm{X}$.\par 
\indent Our primary goal is to estimate the column space of $\bm{C}$ accurately enough so that inference on $\bm{\beta}_g$ is just as accurate as the generalized least squares estimate for $\bm{\beta}_g$ when $\bm{C}$ is known, and secondary goal is to simply estimate the column space of $\bm{C}$ when there are no covariates of interest, i.e. $d=0$. Since accomplishing the former implies we can achieve the latter, we only consider the case when $d \geq 1$.\par
\indent To estimate $\bm{C}$, we partition the variability of $\bm{y}_1,\ldots,\bm{y}_p$ into two pieces: one part explainable by $\bm{X}$ and another orthogonal to $\bm{X}$. Specifically, for some positive definite matrix $\bm{G} \in \mathbb{R}^{n \times n}$ we define
\begin{subequations}
\label{equation:DataSplitting}
\begin{align}
\label{equation:yg1}
&\bm{y}_{g_1} = \left( \bm{X}^T \bm{G}^{-1} \bm{X}\right)^{-1}\bm{X}^T \bm{G}^{-1}\bm{y}_{g} = \bm{\beta}_g + \bm{\Omega}\bm{\ell}_g + \bm{e}_{g_1}\,\, (g=1,\ldots,p)\\
\label{equation:OmegaWLS_0}
&\bm{\Omega} = \left( \bm{X}^T \bm{G}^{-1} \bm{X}\right)^{-1}\bm{X}^T \bm{G}^{-1}\bm{C}\\
\label{equation:yg2}
&\bm{y}_{g_2} = \bm{Q}_X^T \bm{y}_g = \bm{C}_{\perp}\bm{\ell}_g + \bm{e}_{g_2}, \quad \bm{e}_{g_2} \sim N\left( \bm{0}, \bm{W}_g\right)\,\, (g=1,\ldots,p)\\
\label{equation:Cperp}
&\bm{C}_{\perp} = \bm{Q}_X^T \bm{C}\\
\label{equation:Vgtilde}
&\bm{W}_g = \sum\limits_{j=1}^b v_{g,j} \bm{Q}_X^T \bm{B}_j \bm{Q}_X \,\, (g=1,\ldots,p)\\
\label{equation:Yi}
&\bm{Y}_1 = \left[
\bm{y}_{1_1} \, \cdots \, \bm{y}_{p_1}
\right]^T = \bm{\beta} + \bm{L}\bm{\Omega}^{T} + \bm{E}_1, \quad \bm{Y}_2 = \left[
\bm{y}_{1_2} \, \cdots \, \bm{y}_{p_2}
\right]^T = \bm{L}\bm{C}_{\perp}^T + \bm{E}_2
\end{align}
\end{subequations}
where the $g^{\text{th}}$ rows of $\bm{E}_1$ and $\bm{E}_2$ are $\bm{e}_{g_1}$ and $\bm{e}_{g_2}$, $\bm{\Omega}$ is the weighted least squares regression coefficient for the regression $\bm{C}$ onto $\bm{X}$ and $\bm{y}_{g_1}$ is the naive weighted least squares estimator for $\bm{\beta}_g$ ($g=1,\ldots,p$) that ignores $\bm{C}$. We show how to choose the appropriate $\bm{G}$ when we provide our estimate for $\bm{\beta}_g$ in Section \ref{subsection:EstOmega}. $\bm{C}_{\perp}$ and $\bm{y}_{g_2}$ lie in the space orthogonal to $\bm{X}$, where the latter no longer depends on $\bm{\beta}_g$. Algorithm \ref{algorithm:EsimationProcedure} below provides a cursory description of how we use the objects defined in \eqref{equation:DataSplitting} to estimate $\bm{C}$.
\begin{algorithm}
\label{algorithm:EsimationProcedure}
\begin{enumerate}[label=\textit{(\alph*)}]
\item Use $\bm{Y}_2$ and a cross-validation procedure to estimate $K$, the dimension of the column space of $\bm{C}$. This procedure is called CBCV (\textbf{C}orrelated \textbf{B}i-\textbf{C}ross \textbf{V}alidation) and is enumerated in Section \ref{subsection:ChooseK}.\label{item:EstimationProcedure:CBCV}
\item For $K$ known, use $\bm{Y}_2$ to estimate $\bm{L}$, $\bm{C}_{\perp}$ and the appropriate $\bm{G}$ using an iterative factor analysis procedure we call ICaSE (\textbf{I}terative \textbf{C}orrelation \textbf{a}nd \textbf{S}ubspace \textbf{E}stimation), which we discuss in Section \ref{subsection:ICaSE}.\label{item:EstimationProcedure:ICaSE}
\item Estimate $\bm{\Omega}$ using $\bm{Y}_1$, $\hat{\bm{L}}$, $\hat{\bm{C}}_{\perp}$ and $\bm{G}$. Since $\bm{C} = \bm{X}\bm{\Omega} + \bm{G}\bm{Q}_X \left(\bm{Q}_X^T \bm{G}\bm{Q}_X\right)^{-1} \bm{C}_{\perp}$, our estimate for $\bm{C}$ is $\hat{\bm{C}} = \bm{X}\hat{\bm{\Omega}} + \bm{G}\bm{Q}_X \left(\bm{Q}_X^T \bm{G}\bm{Q}_X\right)^{-1} \hat{\bm{C}}_{\perp}$. For $K$ known, we call this method of estimating $\bm{C}$ CorrConf and is discussed in Section \ref{subsection:EstOmega}.\label{item:EstimationProcedure:Omega}
\end{enumerate}
\end{algorithm}
If the samples were independent and $K$ were known, steps \ref{item:EstimationProcedure:ICaSE} and \ref{item:EstimationProcedure:Omega} with $\bm{G} = I_n$ are similar to methods used by \citet{LEAPP,Houseman,CATE,BiometrikaConfounding,ChrisANDDan} to estimate $\bm{C}$. However, it is not hard to show that their estimates are inaccurate in the presence of even moderate correlation. In what follows, we first show that for $K$ known, ICaSE is able to estimate $\bm{C}_{\perp}$ as accurately as when samples are independent. We then use the theory developed in step \ref{item:EstimationProcedure:ICaSE} to show that CBCV consistently chooses the correct factor dimension $K$ in step \ref{item:EstimationProcedure:CBCV}, and lastly show that our estimate for $\bm{C}$ in part \ref{item:EstimationProcedure:Omega} ensures that the generalized least squares estimate for $\bm{\beta}_g$ using $\hat{\bm{C}}$ as a plug-in estimate for $\bm{C}$ is asymptotically equivalent to the generalized least squares estimator when $\bm{C}$ is known. 

\section{Estimating the factor dimension $K$ and latent covariates $\bm{C}$}
\label{section:Estimation}
\subsection{A computationally tractable model for the data}
\label{subsection:EstimationModel}
The generative model assumed in \eqref{equation:Model_yANDvar} is not conducive to estimating $\bm{C}$, since this would require jointly estimating $\bm{C}$ and all $p$ covariance matrices $\bm{V}_1,\ldots,\bm{V}_p$. Instead, we use the following simpler, but incorrect, model to estimate $K$, $\bm{L}$ and $\bm{C}$:
\begin{subequations}
\label{equation:WrongModel}
\begin{align}
\label{equation:WrongModel_Y}
\bm{Y} &= \bm{\beta}\bm{X}^T + \bm{L}\bm{C}^T + \bm{E}, \quad \bm{E} \sim MN_{p \times n}\left( \bm{0}, \delta^2 I_p, \bm{V}\right)\\
\label{equation:WrongModel_V}
\bm{V} &= \tau_1 \bm{B}_1 + \cdots + \tau_b \bm{B}_b, \quad \log\abs{\bm{Q}_X^T \bm{V} \bm{Q}_X} = 0\\
\label{equation:WrongModel_tau}
\bm{\tau} &= \left(
\tau_1 \, \cdots \, \tau_b
\right)^T
\end{align}
\end{subequations}
where we introduce $\delta^2$ so that $\bm{V}$ is scale-invariant (and equal to $I_n$ when samples are independent), and is defined in terms of the determinant for reasons expounded upon in Section \ref{subsection:ChooseK}. Since this is not the correct model, we use the KL-divergence to express the ``closest'' model parameters $\delta^2$, $\bm{V}$ in terms of the data-generating model parameters from \eqref{equation:Model_yANDvar} in Proposition \ref{proposition:KLVdelta} below.
\begin{proposition}
\label{proposition:KLVdelta}
Let $\bar{\bm{V}} = p^{-1}\sum\limits_{g=1}^p \bm{V}_g$ and $c_1$ be a constant not dependent on $n$ or $p$ and suppose the matrix $\left[\bm{H}\right]_{rs} = n^{-1}\Tr\left( \bm{B}_r \bar{\bm{V}}^{-1}\bm{B}_s \bar{\bm{V}}^{-1} \right) - c_1^{-1} I(r=s)$ ($r=1,\ldots,b; s=1,\ldots,b$) is positive definite for all $n$ large enough. Let $F_{\eqref{equation:Model_yANDvar}}$ and $F_{\eqref{equation:WrongModel},\left(\bm{\tau}, \delta^2 \right)}$ be the distribution functions for $\bm{Y}$ under the data-generating and incorrect, but tractable model assumed in \eqref{equation:Model_yANDvar} and \eqref{equation:WrongModel} with parameters $\bm{\tau}, \delta^2$, respectively, and define
\begin{align*}
\left(\bm{\tau}_*, \delta_*^2 \right) &= \argmin_{\substack{ \bm{V}\left( \bm{\tau} \right) \succ \bm{0} \\ \log \abs{\bm{Q}_X^T \bm{V}\left( \bm{\tau} \right) \bm{Q}_X} = 0 }} KL\left\lbrace F_{\eqref{equation:Model_yANDvar}} \mid \mid F_{\eqref{equation:WrongModel},\left(\bm{\tau}, \delta^2 \right)} \right\rbrace
\end{align*}
where $\bm{V} = \bm{V}\left( \bm{\tau} \right)$. Then for $\bm{V}_* = \bm{V}\left( \bm{\tau}_*\right)$, $\bm{W}_* = \bm{Q}_X^T \bm{V}_* \bm{Q}_X$ and $\bm{W}_g$ defined in \eqref{equation:Vgtilde},
\begin{align}
\label{equation:KLVResult}
\delta_*^2\bm{\tau}_* = p^{-1}\sum\limits_{g=1}^p \bm{v}_g, \qquad \delta_*^2\bm{V}_* = p^{-1}\sum\limits_{g=1}^p \bm{V}_g, \qquad \delta_*^2\bm{W}_* = p^{-1}\sum\limits_{g=1}^p \bm{W}_g
\end{align}
and $\bm{\tau}_* \in \Theta$ for all $n$ large enough..
\end{proposition}
This is a simple, but auspicious result that helps to guarantee that using the incorrect model \eqref{equation:WrongModel} to perform computationally tractable factor analysis to estimate $\bm{C}$ does not sacrifice statistical accuracy. For example, consider estimating $\bm{C}_{\perp}$ in step \ref{item:EstimationProcedure:ICaSE} of Algorithm \ref{algorithm:EsimationProcedure}. The empirical between-sample covariance matrix can be written as
\begin{align*}
\bm{S}_2 = p^{-1}\bm{Y}_2^T \bm{Y}_2 \approx \bm{C}_{\perp}\left( p^{-1} \bm{L}^T \bm{L}\right)\bm{C}_{\perp}^T + p^{-1}\bm{E}_2^T \bm{E}_2 \approx \bm{C}_{\perp}\left( p^{-1} \bm{L}^T \bm{L}\right)\bm{C}_{\perp}^T + \delta_*^2 \bm{W}_*.
\end{align*}
Since $\bm{W}_*$ is not a multiple of the identity, the first $K$ eigenvectors of $\bm{S}_2$ will not be an accurate estimate for $\bm{C}_{\perp}$. If $\bm{W}_*$ were known, we could instead first estimate $\bm{W}_*^{-1/2}\bm{C}_{\perp}$ as the first $K$ eigenvectors of $\bm{W}_*^{-1/2}\bm{S}_2 \bm{W}_*^{-1/2}$, and then rescale the estimate by $\bm{W}_*^{1/2}$ to estimate $\bm{C}_{\perp}$. It turns out that this estimation paradigm performs just as well as when samples are known to be independent.

\subsection{ICaSE: an iterative factor analysis to estimate $\bm{C}_{\perp}$}
\label{subsection:ICaSE}
Here we present the algorithm ICaSE, a method to estimate $\delta_*^2$, $\bm{V}_*$ and $\bm{C}_{\perp}$ when $K$ is known from the portion of the data in the orthogonal complement of $\bm{X}$, as well as theory to illustrate its accuracy. We will also use ICaSE and its theoretical guarantees to estimate $K$ in Section \ref{subsection:ChooseK}. Recall from \eqref{equation:DataSplitting} that the mean of $\bm{Y}_2$ is $\bm{L}\bm{C}_{\perp}^T$ and is not dependent on $\bm{X}$ and the covariance of the $g^{\text{th}}$ row is $\bm{W}_g$, were $\delta^2_*\bm{W}_* = p^{-1}\sum\limits_{g=1}^p \bm{W}_g$. The discussion at the end of Section \ref{subsection:EstimationModel} suggests that if $\bm{W}_*$ were known, one could estimate $\bm{C}_{\perp}$ by first computing the first $K$ right singular vectors of $\bm{Y}_2 \bm{W}_*^{-1/2}$, and then re-scaling them by $\bm{W}_*^{1/2}$. On the other hand, if $\bm{C}_{\perp}$ were known, one could easily estimate $\delta_*^2$ and $\bm{\tau}_*$ using restricted maximum likelihood. ICaSE iterates between these two steps to jointly estimate $\delta_*^2$, $\bm{\tau}_*$ and $\bm{C}_{\perp}$.\par 
\indent It remains to show how to determine an appropriate starting point for $\delta_*^2$ and $\bm{\tau}_*$ that avoids incorporating signal from the random effect into the estimate for $\bm{C}_{\perp}$, as imprudent starting points often beget biased estimates for $\bm{\tau}_*$, and therefore $\bm{C}_{\perp}$. For example, consider a simple scenario where $\bm{W}_g = \sigma_1^2 \bm{R}\bm{R}^T + \sigma_0^2 I_{n-d}$ for all $g=1,\ldots,p$, where $\bm{R} \in \mathbb{R}^{(n-d) \times s}$ and $\bm{R}\bm{R}^T \neq I_{n-d}$. Then we can re-write $\bm{Y}_2$ as
\begin{align*}
\bm{Y}_2 = \bm{L}\bm{C}_{\perp}^T + \sigma_1\bm{\Gamma}\bm{R}^T + \sigma_0\bm{\mathcal{E}}, \quad \bm{\Gamma} \sim MN_{p \times s}\left( \bm{0}, I_p, I_s\right), \,\, \bm{\mathcal{E}}\sim MN_{p \times (n-d)}\left( \bm{0}, I_p, I_{n-d}\right).
\end{align*}
If we use $\bm{W}_* = I_{n-d}$ as a starting point, the initial estimate for $\bm{C}_{\perp}$ will be a mixture of $\bm{C}_{\perp}$ and $\bm{R}$. If $K$ is moderate to large or $n$ is small, we will have a difficult time re-estimating $\bm{W}_*$ because too much of the random effect would be treated as a fixed effect. Further, if we attribute signal from the random effect as arising from $\bm{L}\bm{C}_{\perp}^T$, we will tend to overestimate $K$ and have fewer degrees of freedom to estimate $\bm{\beta}$.\par 
\indent In order to better separate the variability in $\bm{Y}_2$ due to the unobserved covariates $\bm{C}_{\perp}$ and the random effect, we employ a ``warm start" technique often used when solving penalized regression problems \citep{WarmStart}. There, the solution from an optimization with a larger regularization constant, which typically yields a model with few parameters, is used as a starting point for the same problem with a smaller regularization constant. In our case, we initialize our estimates for $\delta_*^2$ and $\bm{\tau}_*$ assuming $K=0$, since there is no ambiguity where the variability in $\bm{Y}_2$ is coming from if $\bm{L}\bm{C}_{\perp}^T = 0$. We then use the estimate for $\bm{W}_*$ when the dimension of the column space of $\bm{C}_{\perp}$ is $k-1$ as the starting point when the dimension is $k$. This ensures that we attribute as much variability as possible to the random effect and only attribute signal to the latent covariates if the observed signal is not amenable to the model for the variance. Algorithm \ref{algorithm:ICaSE} below enumerates the steps in ICaSE.

\begin{algorithm}[ICaSE]
\label{algorithm:ICaSE}
\begin{enumerate}[label=\textit{(\arabic*)}]
\setcounter{enumi}{-1}
\item For $K = 0$, estimate $\delta_*^2$ and $\bm{W}_*$ by restricted maximum likelihood using the incorrect model $\bm{Y}_2 \sim MN_{p \times (n-d)}\left( \bm{0}, \delta^2 I_p, \bm{W}\right)$ under the restriction that $\log\abs{\bm{W}} = 0$ and $\delta^2 \bm{\tau} \in \Theta$.\label{ICaSE:startV}
\item For $K=k$ and given $\hat{\delta}^2$ and $\hat{\bm{W}}$, estimates for $\bm{W}_*$ and $\delta_*^2$, estimate $\hat{\bm{W}}^{-1/2}\bm{C}_{\perp}$ using the truncated singular value decomposition of $\bm{Y}_2 \hat{\bm{W}}^{-1/2}$ and rescale the estimate by $(n-d)^{1/2}\hat{\bm{W}}^{1/2}$ to get an estimate for $\bm{C}_{\perp}$, $\hat{\bm{C}}_{\perp}  \in \mathbb{R}^{(n-d) \times k}$.\label{ICaSE:Subspace}
\item Given $\hat{\bm{C}}_{\perp}$, obtain $\hat{\delta}^2$ and $\hat{\bm{\tau}}$ (estimates for $\delta_*^2$ and $\bm{\tau}_*$) by restricted maximum likelihood, assuming $\bm{Y}_2 \sim MN_{p \times (n-d)}\left\lbrace \bm{L}\hat{\bm{C}}_{\perp}^T, \delta^2 I_p, \bm{W}\left( \bm{\tau}\right)\right\rbrace$ under the restriction $\log\abs{\bm{W}\left( \bm{\tau}\right)} = 0$ and $\delta^2 \bm{\tau} \in \Theta$.\label{ICaSE:Variance}
\item Iterate between steps \ref{ICaSE:Subspace} and \ref{ICaSE:Variance}, stop on step \ref{ICaSE:Subspace} of the second iteration and use $\hat{\bm{W}}$ as the starting point for $K = k+1$ in step \ref{ICaSE:Subspace}.
\end{enumerate}
\end{algorithm}

\indent Evidently, iterating between steps \ref{ICaSE:Subspace} and \ref{ICaSE:Variance} is not explicitly maximizing an objective function. However, we prove that one needs only to iterate two times to obtain estimates for $\bm{C}_{\perp}$ that are nearly as accurate as those obtained from methods that assume samples are independent \citep{ChrisANDDan}. Before we explicitly state our theory, we place technical assumptions on the components of the data-generating model \eqref{equation:Model_yANDvar} in Assumption \ref{assumption:Basics}, as well as assumptions on the parameter estimates from step \ref{ICaSE:Variance} of Algorithm \ref{algorithm:ICaSE} in Assumption \ref{assumption:ICaSE}.


\begin{assumption}
\label{assumption:Basics}
Let $c_2 > 0$ be a constant not dependent on $n$ or $p$. Assume:
\begin{enumerate}[label=\textit{(\alph*)}]
\item $\bm{V}_g - c_2^{-1}I_n \succ \bm{0}$, $\abs{v_{g,j}} \leq c_2$ and $\norm{\bm{B}_j}_2 \leq c_2$, where $g \in [p]$ and $j \in [b]$.\label{item:assumptionBasics:V}
\item $\bm{C} \in \mathbb{R}^{n \times K}$ is a non-random, full-rank matrix such that
\begin{align*}
(n-d)^{-1}\bm{C}^T \bm{V}_*^{-1/2} P_{V_*^{-1/2}X}^{\perp} \bm{V}_*^{-1/2}\bm{C} = I_K
\end{align*}
and $K \geq 0$ is an unknown constant not dependent on $n$ or $p$.\label{item:assumptionBasics:CtC}
\item $(n-d)p^{-1} \bm{L}^T\bm{L} = \diag\left( \gamma_1, \ldots, \gamma_K\right)$ where $c_2^{-1} \leq \gamma_K < \cdots < \gamma_1 \leq c_2n$, $\gamma_1/\gamma_K \leq c_2$ and $\left( \gamma_k - \gamma_{k+1}\right)/\gamma_{k+1} \geq c_2^{-1}$ ($k = 1,\ldots,K$ and $\gamma_{K+1} = 0$). Further, the magnitude of the entries of $\bm{L}$ are bounded above by $c_2$.\label{item:assumptionBasics:LtL}
\item $p$ is a non-decreasing function of $n$ such that $n/p, \, n^{3/2}/\left( p\gamma_K\right) \to 0$.\label{item:assumptionBasics:np}
\end{enumerate}
\end{assumption}
Item \ref{item:assumptionBasics:V} ensures that $\bm{V}_g$ is interpretable as a covariance matrix and that no particular direction explains an overwhelming majority of the variability in the residuals $\bm{e}_g$. The conditions that $P_{V_*^{-1/2}X}^{\perp} \bm{V}_*^{-1/2}\bm{C}$ have orthogonal columns and $\bm{L}^T \bm{L}$ be diagonal with decreasing elements in items \ref{item:assumptionBasics:CtC} and \ref{item:assumptionBasics:LtL} are without loss of generality by the identifiability of $\bm{C}\bm{\ell}_g$ in model \eqref{equation:Model_yANDvar} and because our eventual estimator for $\bm{\beta}_g$ only depends on the column space of $\bm{C}$. These two conditions imply that $\gamma_1,\ldots,\gamma_K$ are exactly the non-zero eigenvalues of the identifiable matrix
\begin{align*}
P_{V_*^{-1/2}X}^{\perp} \bm{V}_*^{-1/2}\bm{C}\left( p^{-1}\bm{L}^T\bm{L}\right) \bm{C}^T \bm{V}_*^{-1/2} P_{V_*^{-1/2}X}^{\perp}
\end{align*}
and quantify the average variability in $\bm{y}_1,\ldots,\bm{y}_p$ due to the latent factors that can be unequivocally distinguished from the variation due to $\bm{X}$. This intuition will be important when we describe our estimation procedure in Section \ref{subsection:EstOmega} and when we analyze simulated data in Section \ref{section:SimulatedData}. We choose to treat the unobserved factors $\bm{C}$ as non-random to illustrate that naive estimates for $\bm{\beta}_g$ that do not account for $\bm{C}$ are biased, although our main results can easily be extended to the case when $\bm{C}$ is assumed random. Lastly, item \ref{item:assumptionBasics:np} assumes we are in the high dimensional setting ($p > n$) and there are enough units to suitably estimate $\bm{C}$ ($n^{3/2}/\left(p\gamma_K\right) \to 0$). The latter assumption is standard in the latent factor correction literature when $\bm{V}_g$ is assumed to be a multiple of the identity, where \cite{CATE} considers the case when the data are acutely informative for $\bm{C}$ ($\gamma_K \asymp n$) and \cite{ChrisANDDan} allow the data to be only moderately informative for some of the factors ($\gamma_K \geq c_2^{-1}$).\par


\begin{assumption}
\label{assumption:ICaSE}
\begin{enumerate}[label=\textit{(\alph*)}]
\item The estimates $\hat{\delta}^2$ and $\hat{\bm{\tau}}$ from step \ref{ICaSE:Variance} in Algorithm \ref{algorithm:ICaSE} are such that $\hat{\delta}^2 \hat{\bm{\tau}}$ lies in the convex set $\Theta_*$, where
\begin{align*}
\Theta_* = \Theta \cap \left\lbrace \bm{x} \in \bm{R}^b : \norm{\bm{x}}_2 \leq 2 b c_2, \bm{V}\left( \bm{x} \right) - \left( 2c_2\right)^{-1} I_n \succ \bm{0} \right\rbrace
\end{align*}
where $c_2$ was defined in Assumption \ref{assumption:Basics}.\label{item:ICaSE:ThetaStar}
\item Define the function
\begin{align*}
h_n\left( \delta^2 \bm{\tau}\right) = -n^{-1}\log\abs{\delta^2 \bm{V}\left( \bm{\tau}\right)} - n^{-1}\Tr\left[\delta_*^2 \bm{V}_* \left\lbrace \delta^2 \bm{V}\left( \bm{\tau}\right) \right\rbrace^{-1} \right].
\end{align*}
Then $h_n \to h$ uniformly in $\Theta_*$ and $-\nabla^2 h_n \left( \delta_*^2 \bm{\tau}_*\right) - c_1^{-1}I_b \succ \bm{0}$, where $c_1$ was defined in the statement of Proposition \ref{proposition:KLVdelta}.\label{item:ICaSE:Technical}

\item Define the function 
\begin{align*}
f_{n,g}\left( \bm{\theta}\right) = -n^{-1}\log\abs{\bm{V}\left(\bm{\theta} \right)} - n^{-1}\Tr\left\lbrace \bm{V}_g \bm{V}\left(\bm{\theta} \right)^{-1} \right\rbrace \,\,\, (g = 1,\ldots,p).
\end{align*}
Then $f_{n,g} \to f_g$ uniformly in $\Theta_*$ and $-\nabla^2 f_{n,g}\left( \bm{v}_g\right) - c_3^{-1} I_b \succ \bm{0}$ for all $g = 1,\ldots,p$, where $c_3$ is a constant not dependent on $n$ or $p$.\label{item:Vg:Technical}
\end{enumerate}
\end{assumption}
Item \ref{item:ICaSE:ThetaStar} simply makes the residual variance parameter space compact and is analogous to Assumption D in \cite{BaiLi} and Assumption 2 in \cite{CATE}. The functions $h_n$ and $f_{n,g}$ are the minus KL-divergence defined in Proposition \ref{proposition:KLVdelta} (up to terms not dependent on $\delta^2 \bm{\tau}$) and twice the expected log-likelihood under the model $\bm{e}_g \sim N\left( \bm{0}, \bm{V}_g\right)$. Their convergence properties and conditions on their second derivatives in items \ref{item:ICaSE:Technical} and \ref{item:Vg:Technical} are standard likelihood theory assumptions and ensure consistency and identifiability. We can now state Theorem \ref{theorem:ICaSE}:
\begin{theorem}
\label{theorem:ICaSE}
Suppose the data $\bm{Y}$ are generated according to model \eqref{equation:Model_yANDvar} and Assumptions \ref{assumption:Basics} and \ref{assumption:ICaSE} hold. Suppose further that we stop on step \ref{ICaSE:Subspace} of the second iteration of Algorithm \ref{algorithm:ICaSE} for each $k = 1,\ldots,K_{\max}$, where $K_{\max} \geq K$. Then the estimates for $\delta_*^2$, $\bm{\tau}_*$ and $\bm{C}_{\perp}$ from Algorithm \ref{algorithm:ICaSE} are such that
\begin{subequations}
\label{equation:ConfSub}
\begin{align}
\label{equation:ConfSub:V}
\abs{\hat{\delta}^2 - \delta_*^2}, \norm{\hat{\bm{\tau}} - \bm{\tau}_*}_2 &= O_P\left( n^{-1}\right),\,\,\, k \geq K\\
\label{equation:ConfSub:C}
\norm{P_{\bm{C}_{\perp}} - P_{\hat{\bm{C}}_{\perp}}}_F^2 &= O_P\left\lbrace n/\left( \gamma_K p\right) + \left( p \gamma_K\right)^{-1/2} + \left( n\gamma_K\right)^{-1}\right\rbrace,\,\,\, k = K\\
\label{equation:ConfSub:CKlarge}
\norm{P_{\bm{C}_{\perp}} - P_{\hat{\bm{C}}_{\perp}}P_{\bm{C}_{\perp}} }_F^2 &= O_P\left\lbrace n/\left( \gamma_K p\right) + \left( p \gamma_K\right)^{-1/2} + \left( n\gamma_K\right)^{-1}\right\rbrace,\,\,\, k \geq K
\end{align}
\end{subequations}
where $\hat{\delta}^2$ and $\hat{\bm{\tau}}$ depend on $k$.
\end{theorem}
\noindent Theorem \ref{theorem:ICaSE} implies that the column space of $\bm{C}_{\perp}$ is estimated well when $k=K$ and its column space is approximately a subspace of the column space of $\hat{\bm{C}}_{\perp}$ when we overestimate $K$. This result is quite remarkable because besides the additional factor $(n\gamma_K)^{-1}+\left( p \gamma_K\right)^{-1/2}$ (which is $<< n^{-1/2}$), this is the same rate obtained when the samples are uncorrelated (i.e. when $\bm{V}_g$ is a multiple of the identity for every unit $g = 1,\ldots,p$) \citep{ChrisANDDan}.

\subsection{Correlated Bi-Cross Validation to estimate $K$}
\label{subsection:ChooseK}
We propose estimating $K$, the dimension of the column space of $\bm{C}$, with a cross-validation paradigm developed in \citet{bcv} that partitions $\bm{Y}_2$ into a training and test set. Besides having substantially shorter computation time and providing a less variable estimate for $K$, what differentiates our method from the one developed in \citet{bcv} is we provide an approach amenable to dependent data, while their method is only valid for independent data.\par
\indent Our primary point of concern is ensuring our estimates for $K$ are not biased by the correlations in our data \citep{EltsStatLearning,GeneralizedCV}. An alternative, but equivalent, concern is we want to avoid including factors arising from the random effect in our estimate for $K$. To describe our procedure, we assume for notational convenience that $\bm{Y}_{p \times n} = \bm{L}_{p \times K}\bm{C}_{n \times K}^T + \bm{E}_{p \times n}$ where the rows $\bm{e}_g$, $\bm{e}_{g'}$ of $\bm{E}$ are independent for $g \neq g'$ and $\bm{e}_g \sim \left( \bm{0}, \bm{V}_g \right)$ (i.e. $\bm{Y}$ is distributed according to \eqref{equation:Model_yANDvar} with $d = 0$). Algorithm \ref{algorithm:CBCV} provides an outline for the algorithm we use to estimate $K$, which we call CBCV (Correlated Bi-Cross Validation).

\begin{algorithm}[CBCV]
\label{algorithm:CBCV}
Randomly partition the rows of $\bm{Y}$ (i.e. genes) into $F = O(1)$ folds (our software default is $F=5$). Fix some fold $f\in [F]$ as the test set (with $p_{f}$ elements) and let the training set be the rest of the data ($p_{(-f)} = p - p_{f}$).
\begin{enumerate}[label=\textit{(\alph*)}]
\item Fix some $k \in \left\lbrace 0,1,\ldots,K_{\max} \right\rbrace$ and without loss of generality assume the partition of the rows is such that
\begin{align*}
\bm{Y} = \begin{bmatrix}
\bm{Y}_{(-f)}\\
\bm{Y}_{f}
\end{bmatrix} = \begin{bmatrix}
\bm{L}_{(-f)}\bm{C}^T\\
\bm{L}_{f}\bm{C}^T
\end{bmatrix} + \begin{bmatrix}
\bm{E}_{(-f)}\\
\bm{E}_{f}
\end{bmatrix}
\end{align*}
where $\bm{Y}_{(-f)} \in \mathbb{R}^{p_{(-f)} \times n}$, $\bm{Y}_{f} \in \mathbb{R}^{p_{f} \times n}$ are the training set and test sets, respectively, and are independent.\label{item:CBCV:partition}
\item Obtain $\hat{\bm{C}} \in \mathbb{R}^{n \times k}$ and $\hat{\bm{V}}_{(-f)}$ from $\bm{Y}_{(-f)}$ using Algorithm \ref{algorithm:ICaSE}.\label{item:CBCV:LogDet}
\item Define $\bar{\bm{Y}}_{f} = \bm{Y}_{f} \hat{\bm{V}}_{(-f)}^{-1/2}$ and $\hat{\bar{\bm{C}}} = \hat{\bm{V}}_{(-f)}^{-1/2}\hat{\bm{C}}$. Determine the loss for this fold, dimension combination using leave-one-out cross validation:
\begin{align}
\label{equation:LOOLoss}
LOO_{f}\left( k \right) = \sum\limits_{i=1}^n \norm{\bar{\bm{y}}_{f,i} - \hat{\bm{L}}_{f,(-i)} \hat{\bar{\bm{c}}}_i }_{2}^2,
\end{align}
where $\bar{\bm{y}}_{f,i} \in \mathbb{R}^{p_f}$ and $\hat{\bar{\bm{c}}}_i \in \mathbb{R}^k$ are the $i^{\text{th}}$ columns of $\bar{\bm{Y}}_{f}$ and $\hat{\bar{\bm{C}}}^T$, respectively, and $\hat{\bm{L}}_{f,(-i)}$ is the ordinary least squares regression coefficient from the regression $\bar{\bm{Y}}_{f}$ onto $\hat{\bar{\bm{C}}}$ whose $i^{\text{th}}$ column and row have been removed, respectively.\label{item:CBCV:Loss}
\item Repeat this for folds $f = 1,\ldots, F$ and $k=0, 1, \ldots, K_{\max}$ and set $\hat{K}$ to be
\begin{align*}
\hat{K} = \argmin_{k \in \left\lbrace 0,1,\ldots,K_{\max} \right\rbrace} \left\lbrace \sum\limits_{f=1}^F LOO_{f}\left( k \right)\right\rbrace.
\end{align*}
\end{enumerate}\label{item:CBCV:ChooseK}
\end{algorithm}
\indent The form that the loss function takes in equation \eqref{equation:LOOLoss} is not the standard squared loss, but is instead scaled by the estimate $\hat{\bm{V}}_{(-f)}^{-1/2}$. This is sensible because it places more importance on correctly estimating the portion of $\bm{C}$ not captured by the model for the residual variance. However, unless proper care is taken, this scaled loss function would underestimate $K$ simply because the estimated residual variance is larger for underspecified mean models. The restriction that $\log \abs{\hat{\bm{V}}_{(-f)}} = 0$ alleviates this issue by making the loss function scale-invariant, a feature that is critical to the performance of CBCV and something that the minus log-likelihood cannot guarantee.\par
\indent To understand why CBCV gives accurate estimates for $K$, we study the expected leave-one-out squared error for a particular fold, latent dimension pair, conditioned on the training data $\bm{Y}_{(-f)}$:
\begin{align*}
\E \left\lbrace LOO_{f}(k) \mid \bm{Y}_{(-f)}\right\rbrace =& \E\left\lbrace \sum\limits_{i=1}^n\norm{\bm{L}_{f} \bar{\bm{c}}_i - \hat{\bm{L}}_{f,(-i)} \hat{\bar{\bm{c}}}_i}_2^2 \mid \bm{Y}_{(-f)}\right\rbrace + \E\left[\Tr\left\lbrace \hat{\bm{V}}_{(-f)}^{-1}\bm{E}_{f}^T \bm{E}_{f}\right\rbrace   \mid \bm{Y}_{(-f)}\right]\\
&- 2\sum\limits_{i=1}^n\E \left\lbrace \left( \bm{L}_{f} \bar{\bm{c}}_i - \hat{\bm{L}}_{f,(-i)} \hat{\bar{\bm{c}}}_i \right)^T \bm{E}_{f} \hat{\bm{V}}_{(-f)}^{-1/2} \bm{a}_i \mid \bm{Y}_{(-f)} \right\rbrace = I + II + III
\end{align*}
where $\bm{a}_i \in \mathbb{R}^n$ is the $i^{\text{th}}$ standard basis vector. In standard cross validation, $\hat{\bm{V}}_{(-f)} = I_n$, meaning the residual variance term ($II$) would be constant for all $k=0,1,\ldots,K_{\max}$ and the correlation term ($III$) would be 0, since $\hat{\bm{L}}_{f,(-i)}$ would not be correlated with the $i^{\text{th}}$ column of $\bm{E}_{f}$, $\bm{E}_{f} \bm{a}_i$. Therefore, the minimizer of the squared bias term ($I$) would also minimize the cross validated error, which is exactly what one would hope because $I$ would be minimized when $k=K$ \citep{bcv}. However, since we must account for the correlation between samples, we now need to ensure that $I + II$ is minimized when $k=K$ and that the correlation term does not contribute to the expected loss. The restriction that $\log \abs{\hat{\bm{V}}_{(-f)}} = 0$ helps ensure this is the case, since
\begin{align*}
(np_f)^{-1} II = \delta_{f*}^2 n^{-1}\Tr\left\lbrace \hat{\bm{V}}_{(-f)}^{-1} \bm{V}_{f*} \right\rbrace \geq \delta_{f*}^2,
\end{align*}
where the inequality holds with equality if and only if $\hat{\bm{V}}_{(-f)} = \bm{V}_{f*}$ by Jensen's Inequality (the average log-eigenvalue of $\hat{\bm{V}}_{(-f)}^{-1} \bm{V}_{f*}$ is 0, which is no greater than the log of the average eigenvalue of $\hat{\bm{V}}_{(-f)}^{-1} \bm{V}_{f*}$). Since an accurate estimate of $\bm{C}$ begets an accurate estimate of $\bm{V}_{(-f)*} \approx \bm{V}_{f*}$, the minimizer of $I+II$, and also $I+II+III$ if $III\approx 0$, should be very close to $K$. We make this rigorous with the following theorem:

\begin{theorem}
\label{theorem:CBCV}
Suppose Assumptions \ref{assumption:Basics} and \ref{assumption:ICaSE} hold for $d = 0$ with $\gamma_K \to \infty$ and define $\hat{\bar{h}}_i$ to be the $i^{\text{th}}$ leverage score of $\hat{\bar{\bm{C}}}$ (i.e. the $i^{\text{th}}$ diagonal element of $P_{\hat{\bar{C}}}$) and for each fold $f \in [F]$, suppose we modify the loss function in \eqref{equation:LOOLoss} to be
\begin{align}
\label{equation:NewLOO}
LOO_f(k) = \begin{cases}
\sum\limits_{i=n}^n \norm{\bar{\bm{y}}_{f,i} - \hat{\bm{L}}_{f,(-i)} \hat{\bar{\bm{c}}}_i}_2^2, & \mathop{\max}\limits_{i \in [n]} \hat{\bar{h}}_i \leq 1 - \eta \frac{\log m_n}{m_n}\\
\infty, & \text{otherwise}
\end{cases}
\end{align}
where $\eta > 0$ is a constant and $m_n = \min\left( n, p/n \right)$. Then if $K_{\max} \geq K$ and the maximum leverage score of $\bm{V}^{-1/2}_*\bm{C}$ is $o(1)$ as $n \to \infty$,
\begin{align*}
\argmin_{k \in \left\lbrace 0,1,\ldots,K_{\max}\right\rbrace}\E\left\lbrace LOO_f(k) \mid \bm{Y}_{(-f)}, \bm{\pi}\right\rbrace = K + o_P(1)
\end{align*}
where $\bm{\pi}$ is sampled uniformly from the set of all permutations on $[p]$ and partitions the $p$ units into $F$ folds.
\end{theorem}
\indent First, the condition on the leverage scores of $\bm{V}^{-1/2}_*\bm{C}$ help ensure there are no influential points that drive estimates of $\bm{L}_f$ substantially more than others and is a weak assumption, given that the average leverage score of $\bm{V}^{-1/2}_*\bm{C}$ is $K/n \to 0$. Second, the re-definition of the loss function in \eqref{equation:NewLOO} is purely for theoretical reasons to ensure we do not divide by zero and only increases the loss when $k \neq K$. We have not encountered any simulations or real data examples where the maximum leverage score was large enough to justify setting the loss to infinity.\par
\indent An important conclusion of Theorem \ref{theorem:CBCV} is that we can only guarantee our current version of CBCV to correctly estimate $K$ when $p >> n$, since plugging in $\hat{\bm{V}}_{(-f)}$ as an estimate for $\bm{V}_f^*$ in step \ref{item:CBCV:Loss} of Algorithm \ref{algorithm:CBCV} requires $\bm{V}_{(-f)*} \approx \bm{V}_{f*}$, which is only true when the number of units in the training and test sets is large. This is typically the case is gene expression and methylation data when $n$ is on the order of hundreds and $p$ ranges from $10^4 - 10^6$. However, when $n$ is on the same order as $p$ (e.g. in some metabolomic data), it may be wise to estimate $\bm{V}_{f*}$ from $\bm{Y}_f$, although caution will need to be taken to ensure estimates for $K$ are not biased by training and testing with the same data.\par
\indent Theorem \ref{theorem:CBCV} applies when $b=1$, $\bm{B}_1 = I_n$ and $\tau_{*} = 1$, i.e. when samples are assumed independent. When this assumptions holds, we can use results from \cite{ChrisANDDan} to extend Theorem \ref{theorem:CBCV} by removing the assumption that $\gamma_1/\gamma_K \leq c_2$ in item \ref{item:assumptionBasics:LtL} of Assumption \ref{assumption:Basics} (i.e. the confounding effects are all on the same order of magnitude). This is contrary to parallel analysis, a permutation method generally attributed to \cite{BujaFA} and used as the default method to choose $K$ in many software packages \citep{SVAinR,dSVAinR,vicarinR}, which generally fails to capture factors with moderate to small eigenvalues in the presence of factors with larger eigenvalues \citep{PermChooseK}. We show this empirically and discuss it in greater detail in Section \ref{section:SimulatedData}.

\subsection{De-biasing estimates for the main effect in correlated data}
\label{subsection:EstOmega}
Once we have estimated the dimension of the subspace generated by the hidden covariates and the portion of the the subspace orthogonal to the design matrix, $\bm{C}_{\perp}$, we need only estimate the portion of $\bm{C}$ explainable by $\bm{X}$. While $\bm{C}_{\perp}$ allows one to estimate the confounding effects $\bm{\ell}_g$ and the variance $\bm{V}_g$, estimating $P_{X}\bm{C}$ allows one to distinguish the direct effect of $\bm{X}$ on $\bm{y}_g$ from the indirect effect mediated through the latent factors $\bm{C}$, and helps make results reproducible.\par
\indent In order to complete Algorithm \ref{algorithm:EsimationProcedure} and estimate $\bm{C}$, we must first specify a suitable $\bm{G}$ from \eqref{equation:DataSplitting}. Proposition \ref{proposition:KLVdelta} suggests that a reasonable choice would be $\bm{G} = \hat{\bm{V}}$, where $\hat{\bm{V}}$ is computed in Algorithm \ref{algorithm:ICaSE}. From here on out, we assume $K$ is known and define $\bm{y}_{g_1}, \bm{\Omega}$ and $\bm{Y}_1$ from \eqref{equation:yg1}, \eqref{equation:OmegaWLS_0} and \eqref{equation:Yi} by setting $\bm{G} = \hat{\bm{V}}$.\par 
\indent Our strategy for estimating $\bm{\Omega}$ is to regress $\bm{Y}_1$ onto the estimate for $\bm{L}$ obtainable after completing Algorithm \ref{algorithm:ICaSE}, where
\begin{align}
\label{equation:Lhat}
&\hat{\bm{\ell}}_g = \left(\hat{\bm{C}}_{\perp}^T \hat{\bm{W}}^{-1} \hat{\bm{C}}_{\perp}\right)^{-1}  \hat{\bm{C}}_{\perp}^T \hat{\bm{W}}^{-1} \bm{y}_{g_2}.
\end{align}
In order to eventually perform inference on $\bm{\beta}_g$ ($g=1,\ldots,p$), we will require the estimator $\hat{\bm{\Omega}}$ to be such that $n^{1/2}\norm{\hat{\bm{\Omega}} - \bm{\Omega}}_2 = o_P(1)$. If the main effect $\bm{\beta} = \bm{0}$, one possible estimator for $\bm{\Omega}$ is the ordinary least squares estimate
\begin{align}
\label{equation:OmegaEstNaive}
\hat{\bm{\Omega}} = \bm{Y}_1^T \hat{\bm{L}} \left( \hat{\bm{L}}^T \hat{\bm{L}} \right)^{-1}.
\end{align}
When samples are independent, this strategy is similar to those used in \cite{LEAPP,Houseman,BiometrikaConfounding,CATE,Fan1}, although the exact estimators differ slightly. However, \cite{ChrisANDDan} show that depending on the size of the eigenvalues $\gamma_1,\ldots,\gamma_K$, \eqref{equation:OmegaEstNaive} underestimates $\bm{\Omega}$ when samples are independent because $\hat{\bm{L}}$ is a noisy estimate for the design matrix $\bm{L}$, and suggest a bias correction that amounts to replacing $\hat{\bm{L}}^T \hat{\bm{L}}$ with a better estimate for $\bm{L}^T \bm{L}$. Therefore, we use a de-biased estimate of $\bm{\Omega}$, which is analogous to the estimator used in \cite{ChrisANDDan} when samples are assumed independent. Before we provide the estimator, we make the following assumption about the sparsity of the $d$ columns of $\bm{\beta}$:
\begin{assumption}
\label{assumption:B}
Let $s_r$ be the fraction of non-zero terms in the $r^{\text{th}}$ main effect, i.e.
\begin{align*}
s_r = p^{-1}\sum\limits_{g=1}^p I\left( \bm{\beta}_{gr} \neq 0\right) \quad (r=1,\ldots,d).
\end{align*}
Then $s_r = o\left( n^{-3/2}\gamma_K \right)$ for all $r = 1,\ldots, d$ and $\mathop{\max}\limits_{g \in [p]} \norm{\bm{\beta}_g}_{2} \leq c_4$ for some $c_4 > 0$ not dependent on $n$ or $p$.
\end{assumption}
This provides an explicit relationship between the maximum allowable sparsity on the main effect and the informativeness of the data $\bm{Y}$ for estimating $\bm{C}$: the more signal in $\bm{L}\bm{C}_{\perp}^T$ (i.e. the larger $\gamma_K$), the part of $\bm{L}\bm{C}^T$ unequivocally distinguishable from $\bm{B}\bm{X}^T$, the less stringent Assumption \ref{assumption:B} becomes. This is the same maximum allowable sparsity assumed in \citet{CATE} and \citet{ChrisANDDan}, both of which assumed $\bm{V}_g = I_n$ for all $g = 1,\ldots, p$. While this may seem restrictive, we show through simulation in Section \ref{section:SimulatedData} that we can accurately estimate $\bm{C}$ even when this assumption is violated. We can now provide our estimator for $\bm{\Omega}$, as well as a lemma about its asymptotic behavior:
\begin{lemma}
\label{lemma:OmegaWLS}
Suppose Assumptions \ref{assumption:Basics}, \ref{assumption:ICaSE} and \ref{assumption:B} hold, $K$ is known and we estimate $\bm{C}_{\perp}$, $\bm{\tau}_*$ and $\delta_*^2$ according to Algorithm \ref{algorithm:ICaSE}. If we define the estimate for $\bm{\Omega}$ to be
\begin{align}
\label{equation:OmegaWLS_bc}
\hat{\bm{\Omega}}_{bc} = \bm{Y}_1^T \hat{\bm{L}}\left\lbrace \hat{\bm{L}}^T \hat{\bm{L}} - p \hat{\delta}^2\left( \hat{\bm{C}}_{\perp}^T \hat{\bm{W}}^{-1}\hat{\bm{C}}_{\perp} \right)^{-1} \right\rbrace^{-1},
\end{align}
then
\begin{align*}
n^{1/2}\norm{\bm{\Omega} - \hat{\bm{\Omega}}_{bc}}_2 = o_P(1).
\end{align*}
\end{lemma}
\noindent The $p \hat{\delta}^2\left( \hat{\bm{C}}_{\perp}^T \hat{\bm{W}}^{-1}\hat{\bm{C}}_{\perp} \right)^{-1}$ term in \eqref{equation:OmegaWLS_bc} removes the bias in $\hat{\bm{L}}^T \hat{\bm{L}}$ and reduces to the bias correction used in \cite{ChrisANDDan} when $\bm{V}_g = I_n$ for all $g = 1,\ldots, p$. Amazingly, the rate of convergence of $\hat{\bm{\Omega}}_{bc}$ is the same as that achieved in \cite{ChrisANDDan}.\par
\indent Upon having estimated $\bm{\Omega}$, our estimate for $\bm{C}$ is
\begin{align}
\label{equation:Cestimate}
\hat{\bm{C}} = \bm{X}\hat{\bm{\Omega}}_{bc} + \hat{\bm{V}}\bm{Q}_X\hat{\bm{W}}^{-1} \hat{\bm{C}}_{\perp}.
\end{align}
Theorems \ref{theorem:ICaSE} and \ref{theorem:CBCV} and Lemma \ref{lemma:OmegaWLS} show that our estimate for $\bm{C}$ is accurate, but do not guarantee that our downstream generalized least squares estimates for $\bm{\beta}_g$ will be accurate. Therefore, we prove an additional theorem that states that the generalized least squares estimate for $\bm{\beta}_g$ obtained by plugging in $\hat{\bm{C}}$ for $\bm{C}$ has the same asymptotic distribution as the generalized least squares estimate when $\bm{C}$ is known.

\begin{theorem}
\label{theorem:CGLS}
Suppose the assumptions of Lemma \ref{lemma:OmegaWLS} hold, we estimate $\bm{C}$ according to \eqref{equation:Cestimate} and $\bm{V}_g$ and $\bm{\beta}_g$ via restricted maximum likelihood (REML) and generalized least squares (GLS) using the design matrix 
$\left[
\bm{X} \, \hat{\bm{C}}
\right]$. If the REML estimate $\hat{\bm{v}}_g$ is estimated on the parameter space $\Theta_*$ defined in Assumption \ref{assumption:ICaSE}, the following asymptotic relations hold for the GLS estimate $\hat{\bm{\beta}}_g$ and $\hat{\bm{V}}_g = \bm{V}\left( \hat{\bm{v}}_g\right)$:
\begin{align}
\label{equation:VgAsym}
\norm{\hat{\bm{V}}_g - \bm{V}_g}_2 &= o_P(1)\\
\label{equation:BetaAsym}
\sqrt{n}\left( \hat{\bm{\beta}}_g - \bm{\beta}_g\right) &\edist \bm{Z} + o_P(1)
\end{align}
where
\begin{align*}
&\bm{Z} \sim N\left( 0, \hat{\bm{M}}_n\right)\\
&\hat{\bm{M}}_n = \left( n^{-1}\bm{X}^T \hat{\bm{V}}_g^{-1} \bm{X}\right)^{-1} + \hat{\bm{\Omega}}_g \left( n^{-1}\hat{\bm{C}}_{\perp}^T \hat{\bm{W}}_g^{-1}\hat{\bm{C}}_{\perp}\right)^{-1}\hat{\bm{\Omega}}_g^{T}\\
&\hat{\bm{\Omega}}_g = \left( \bm{X}^T \hat{\bm{V}}_g^{-1}\bm{X}\right)^{-1}\bm{X}^T \hat{\bm{V}}_g^{-1} \hat{\bm{C}}.
\end{align*}
Further, $\norm{\hat{\bm{M}}_n - \bm{M}_n}_2 = o_P(1)$, where $\bm{M}_n$ is the finite sample variance for the generalized least squares estimate for $\bm{\beta}_g$ when $\bm{C}$ and $\bm{V}_g$ are known.
\end{theorem}

\section{Simulated multi-tissue gene expression data analysis}
\label{section:SimulatedData}
\subsection{Simulation setup and parameters}
\label{subsection:SimSetup}
We simulated the expression of $p = 15,000$ genes from 50 individuals across three tissues with a complicated tissue-by-tissue correlation structure to compare our method against other state of the art methods designed to estimate $K$ \citep{BujaFA,bcv}, the column space of $\bm{C}_{\perp}$ \citep{PCA,BaiLi} and $\bm{\beta}$ \citep{SVA2008,RUV,BiometrikaConfounding,CATE,ChrisANDDan}. We first randomly chose 25 individuals to be in the treatment group and set $\bm{X} \in \left\lbrace 0,1 \right\rbrace^{n}$ to be the treatment status for the $n=150$ samples. We then set $K = 10$ and for given $\bm{A} \in \mathbb{R}^{1 \times K}$ and $\bm{V}_g$ ($g=1,\ldots,p$), simulated $\bm{y}_{g} \in \mathbb{R}^n$, the expression of gene $g$ in the $n$ (individual, tissue) pairs, as follows:
\begin{subequations}
\label{equation:Simuation}
\begin{align}
&\bm{y}_g = \bm{X}\beta_g + \bm{C}\bm{\ell}_g + 2^{-1/2}\bm{V}_g^{1/2}\bm{e}_g, \quad \bm{e}_{gi} \sim T_4 \,\,\, (g=1,\ldots,p; i=1,\ldots,n)\\
&\beta_g \sim 0.8 \delta_0 + 0.2 N\left(0,0.4^2 \right) \,\,\, (g=1,\ldots,p)\\
& \bm{C} = \bm{X}\bm{A} + \bm{\Xi}\left\lbrace (n-4)^{-1} \bm{\Xi}^T \bm{Q}_{[X\, Z]}\bm{W}_*^{-1}\bm{Q}_{[X\, Z]}^T\bm{\Xi}\right\rbrace^{-1/2}, \quad \bm{\Xi} \sim MN_{n \times K}\left( 0, I_n, I_K\right)\\
& \bm{\ell}_{gk} \sim \pi_k\delta_0 + \left( 1-\pi_k\right)N\left( 0, \eta_k^2\right) \,\,\, (g=1,\ldots,p; k=1,\ldots,K)
\end{align}
\end{subequations}
\noindent where $\delta_0$ is the point mass at 0, $\bm{Z} = \bm{1}_{50} \otimes I_3 \in \mathbb{R}^{n \times 3}$ is the tissue-specific intercept and $T_4$ is the t-distribution with four degrees of freedom, and was chosen to emulate real data with heavy tails. We chose to simulate a non-sparse main effect $\bm{\beta}$ to show that we can violate Assumption \ref{assumption:B} and still do inference that is just as accurate as when $\bm{C}$ is known. The values for $\pi_k$, $\eta_k$ and the resulting $\gamma_k$ ($k=1,\ldots,10$) are provided in Table \ref{Table:L}.\par
\Ltab
\indent In order to mirror the complex and gene-specific nature of cross-tissue correlation patterns, we assumed tissues two and three were more similar to one another than to the first and set $\bm{V}_g = I_{50} \otimes \bm{M}_g$ to be
\begin{align*}
&\bm{M}_g = \C\left\lbrace \left( \epsilon_{g1}\,\, \epsilon_{g2} \,\, \epsilon_{g3} \right)^T\right\rbrace\\
&\epsilon_{g1} = \alpha_{g1} + \xi_{g1}, \,\, \epsilon_{g2} = \phi_{g2}\alpha_{g1} + \alpha_{g2} + \xi_{g2}, \,\, \epsilon_{g3} = \phi_{g3}\alpha_{g1} + \rho_{g3} \alpha_{g2} + \xi_{g3} \\
& \alpha_{g1} \sim \left( 0,v_{g1}^2\right), \,\, \alpha_{g2} \sim \left( 0,v_{g2}^2\right), \,\, \xi_{gj} \sim \left( 0, \sigma_{gj}^2\right) \,\,\, (j=1,2,3).
\end{align*}
The constants $v_{g1}^2, \phi_{g2}, v_{g2}^2, \phi_{g3}, \rho_{g3}$ and $\sigma_{gj}^2$ ($j=1,2,3$) were simulated from Gamma distributions with means 0.8, 1.25, 0.4, 0.75, 1 and 0.2, respectively, each with coefficient of variation equal to 0.2, and subsequently re-scaled so that $\delta_*^2 = \abs{p^{-1}\sum\limits_{g=1}^p \bm{Q}_{[X\, Z]}^T\bm{V}_g \bm{Q}_{[X\, Z]}} = 1$. The average correlation matrix using these parameters is given in Table \ref{Table:Corr} below.

\TissueCorr

\noindent This seemingly complex tissue-by-tissue covariance structure is amenable to the variance model assumed in \eqref{equation:Model_var}, since
\begin{align}
\bm{V}_g =& \sum\limits_{r=1}^3 \sum\limits_{s=r}^3 v_{grs}\left\lbrace I_{50}\otimes \left( \bm{a}_{rs} \bm{a}_{rs}^T\right)\right\rbrace = \sum\limits_{r=1}^3 \sum\limits_{s=r}^3 v_{grs} \bm{B}_{rs} \label{equation:BjSimulation}
\end{align}
where $\bm{a}_{rs} \in \mathbb{R}^3$ has a 1 in the $r^{\text{th}}$ and $s^{\text{th}}$ coordinates and 0 everywhere else and 
\begin{align*}
v_{g12}, \, v_{g13}, \, v_{g23} \geq 0 , \quad v_{g11}+v_{g12}+v_{g13}, \, v_{g22}+v_{g12}+v_{g23}, \, v_{g33}+v_{g13}+v_{g23} \geq 0.
\end{align*}
We lastly set $\bm{A} = \alpha \bm{1}_K^T$ and fixed $\alpha$ across all simulated datasets so $\left( \bm{Q}_Z^T \bm{V}_* \bm{Q}_Z\right)^{-1/2}\bm{Q}_Z^T\bm{C}$ explained approximately 30\% of the variability in $\left( \bm{Q}_Z^T \bm{V}_* \bm{Q}_Z\right)^{-1/2}\bm{Q}_Z^T\bm{X}$, on average. For each simulated data set, we set treatment status, $\bm{X}$, to be the covariate of interest and the tissue specific intercepts, $\bm{Z}$, to be nuisance covariates and estimated $K$ using CBCV with $F=3$ folds, $\bm{C}_{\perp}$ and $\bm{V}_*$ using ICaSE and $\bm{C}$ using \eqref{equation:Cestimate}. We then estimated $\bm{V}_g$ and $\bm{\beta}_g$ using restricted maximum likelihood and generalized least squares, respectively, as described in the statement of Theorem \ref{theorem:CGLS}.\par

\subsection{Comparison of estimators for $K$ and $\bm{C}$}
\label{subsection:SimData:EstKandC}
Since (as far as we are aware) our methods for estimating the factor dimension $K$ and $\bm{C}$ are the first methods designed to account for potential correlation among samples, we compared our estimates for $K$, $\bm{C}$ and subsequently $\bm{\beta}$ with state of the art methods designed for data with independent samples. We first evaluated our estimates for $K$ in 50 simulated datasets in Figure~\ref{Fig:EstK} and compared them with the two most widely used methods in the biological literature: parallel analysis \citep{BujaFA} and bi-cross validation \citep{bcv}, as implemented in the R-packages SVA \citep{SVAinR} and CATE \citep{CATEinR}, respectively. As predicted by Theorem \ref{theorem:CBCV}, our method estimates $K$ correctly in 48 out of the 50 simulated data sets (Figure~\ref{Fig:EstK}), while bi-cross validation and parallel analysis drastically overestimate it, since both methods are effectively treating the random effect as part of the latent fixed effect term $\bm{L}\bm{C}_{\perp}^T$. We underestimated $K$ in 2 datasets because the three components with the smallest $\gamma_k$'s were shrouded by the heavy-tailed residuals. When residuals were normally distributed, we estimated $K$ correctly in every dataset.\par
\indent An interesting feature of Figure~\ref{Fig:EstK} is parallel analysis' estimates for $K$ are smaller than bi-cross validation's, which is a manifestation of the more general phenomenon that parallel analysis fails to identify components with smaller eigenvalues $\gamma_k$ in the presence of components with larger eigenvalues. In fact, when we set $\gamma_k = n$ for $k = 1,\ldots,5$ and left the remaining eigenvalues as set in Table \ref{Table:L}, parallel analysis estimated $K$ to be only 5. This is because parallel analysis' approximations of the singular values of $\bm{Y}_2$ under the null hypothesis $\bm{L}\bm{C}_{\perp}^T = \bm{0}$ are obtained by independently permuting the entries in each row of $\bm{Y}_2$, and therefore tend to be large when the signal in $\bm{L}\bm{C}_{\perp}^T$ is large. We refer the reader to Section 3.1 of \cite{PermChooseK} for a more detailed discussion of this phenomenon.\par

\begin{figure}[ht]
\centering
\includegraphics[scale=0.4]{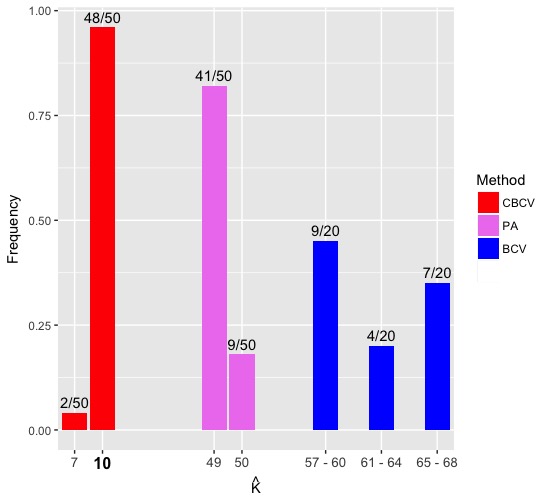}
\caption{Our proposed method's (CBCV) and parallel analysis' (PA) estimates for $K$ in 50 simulated datasets and bi-cross validation's (BCV) estimates in 20 simulated datasets with the true $K=10$. We only estimated $K$ with BCV in 20 datasets because of its slow run time.}\label{Fig:EstK}
\end{figure}

\indent Assuming $K=10$ was known, we next compared the accuracy of our estimates for the column space of $\bm{C}_{\perp}$ using ICaSE with the accuracy of $K$-partial SVD \citep{PCA} and maximum likelihood \citep{BaiLi}, which first estimates $\bm{L}$ and the diagonal matrix $\bm{\Sigma}$ under the quasi likelihood model $\bm{Y}_2 \sim MN_{p \times (n-4)} \left(\bm{0}, \bm{L}\bm{L}^T + \bm{\Sigma}, I_{n-4} \right)$, and sets $\hat{\bm{C}}_{\perp} = \bm{Y}_2^T \hat{\bm{\Sigma}}^{-1}\hat{\bm{L}}\left( \hat{\bm{L}}^T \hat{\bm{\Sigma}}^{-1}\hat{\bm{L}}\right)^{-1}$. For each estimate $\hat{\bm{C}}_{\perp}$, we measured the angle between the column space of $\bm{C}_{\perp}$, $\im\left(\bm{C}_{\perp} \right)$, and $\im\left( \hat{\bm{C}}_{\perp}\right)$ as
\begin{align*}
\measuredangle\left( \bm{C}_{\perp}, \hat{\bm{C}}_{\perp}\right) = \max_{\bm{v} \in \im\left( \bm{C}_{\perp}\right)} \left[\min_{\hat{\bm{v}} \in \im\left(\hat{\bm{C}}_{\perp} \right)}\left\lbrace \cos^{-1}\left( \frac{\hat{\bm{v}}^T \bm{v}}{\norm{\hat{\bm{v}}}_2 \norm{\bm{v}}_2} \right)\right\rbrace \right],
\end{align*}
which is a symmetric function, provided the dimensions of $\im\left(\bm{C}_{\perp} \right)$ and $\im\left( \hat{\bm{C}}_{\perp}\right)$ are the same. In order to benchmark the performance of ICaSE, we also simulated additional datasets $\bar{\bm{Y}} \in \mathbb{R}^{15,000 \times 150}$ with independent columns, which were generated with the parameters given in \eqref{equation:Simuation} and Table~\ref{Table:L}, except we fixed $\bm{V}_1=\cdots=\bm{V}_p = I_n$. The angles between the actual and estimated subspace for 50 simulated datasets $\bm{Y}$ and $\bar{\bm{Y}}$ are summarized in Figure~\ref{Figure:Sim:CtC}. Just as Theorem \ref{theorem:ICaSE} predicts and the discussion at the end of Section \ref{subsection:EstimationModel} anticipates, ICaSE accurately estimates the column space of $\bm{C}_{\perp}$, whereas naive singular value decomposition and maximum quasi likelihood that ignores the between-sample correlation cannot recover the latent subspace. Further, ICaSE's estimate for the column space of $\bm{C}_{\perp}$ when expression across samples exhibits a complex correlation structure is approximately as accurate as when it is known to be independent.\par 

\begin{figure}[ht]
\centering
\subcaptionbox{ \label{Figure:Sim:CtC}}{%
    \includegraphics[scale=0.35]{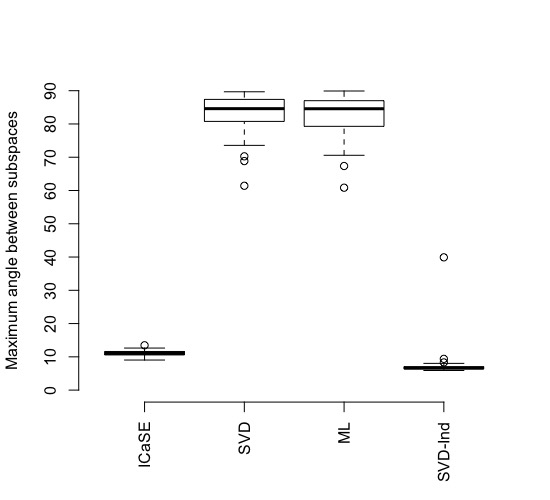} \hspace*{0.5em}%
  }
\subcaptionbox{ \label{Figure:Sim:Beta}}{%
    \includegraphics[scale=0.35]{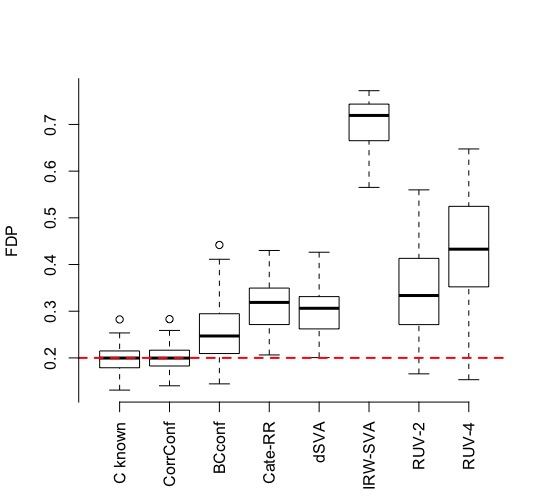}
  }
\caption{(a): The angle between the actual and estimated $\bm{C}_{\perp}$ for 50 simulated datasets $\bm{Y}$, where $\hat{\bm{C}}_{\perp}$ was estimated using our proposed method ICaSE, the first $K$ right singular vectors of $\bm{Y}_2$ (SVD) and maximum quasi likelihood (ML), all with $K=10$ known. Independent SVD (SVD-Ind) estimated $\bm{C}_{\perp}$ as the first $K=10$ right singular vectors of $\bar{\bm{Y}}_2$, which was generated with independent samples. (b): The false discovery proportion (FDP) among genes with a q-value no greater than 0.2 in 50 simulated datasets. We randomly chose 90 genes with $\beta_g=0$ as control genes in RUV-2 and RUV-4.}\label{Fig:Sim}
\end{figure}

\indent We lastly estimated $\bm{\beta}$, the effect of treatment on gene expression in each of the 50 simulated datasets, via generalized least squares with the design matrix $\hat{\bm{M}} = \left[ \bm{X}\, \bm{Z}\, \hat{\bm{C}}\right]$, where $\hat{\bm{C}}$ was estimated assuming $K=10$ was known using CorrConf (our estimator specified in \eqref{equation:Cestimate}), BCconf (the estimator proposed in \cite{ChrisANDDan}), Cate-RR (the software default robust regression estimator suggested in \cite{CATE}), dSVA \citep{BiometrikaConfounding}, iteratively re-weighted SVA (IRW-SVA) \citep{SVA2008}, RUV-2 \citep{RUV}, RUV-4 \citep{RUV4} and when $\bm{C}$ was known (i.e. $\hat{\bm{C}} = \bm{C}$). In order to make the estimation of $\bm{V}_g$ ($g=1,\ldots,p$) computationally tractable, we modeled the simulated residuals $\bm{e}_1,\ldots,\bm{e}_p$ as
\begin{align*}
\bm{e}_g \sim N\left( \bm{0}, v_g \bm{V}\right) \,\,\, (g=1,\ldots,p), \quad \bm{V} = \sum\limits_{r=1}^3 \sum\limits_{s=r}^3 \tau_{rs} \bm{B}_{rs},
\end{align*}
where $\bm{B}_{rs}$ was defined in \eqref{equation:BjSimulation}, and estimated $v_1,\ldots,v_p$ and $\tau_{rs}$ with restricted maximum likelihood for each of the eight methods using the estimated design matrix $\hat{\bm{M}}$. We then computed the \textit{P} value for the null hypothesis $\beta_g = 0$ ($g=1,\ldots,p$) by comparing $\hat{\beta}_g$ to a t-distribution with $n-4-K = n-14$ degrees of freedom, used these \textit{P} values as input into q-value \citep{qvalue} to control the false discovery rate (as this is a software popular among biologists) and deemed a gene to have significantly different expression across the two treatment conditions if its q-value was no greater than 0.2. Figure~\ref{Figure:Sim:Beta} plots the true false discovery proportion in 50 simulated datasets among genes with a q-value less than or equal to 0.2 for each of the eight methods.\par
\indent The performance of our method (CorrConf), as the statement of Theorem \ref{theorem:CGLS} suggests, is nearly indistinguishable from the generalized least squares estimator when $\bm{C}$ is known, with nearly identical power. On the other hand, the other six methods tend to introduce more type I errors because their estimates for $\bm{C}$ do not account for the dependence between the residuals and therefore fail to recover $\bm{\Omega}$, the variation in $\bm{C}$ due to treatment. When $K$ is underestimated, these six methods have even poorer performance than when $K=10$ is known, and still fail to control the false discovery rate and exhibit a decrease in power when $K$ is overestimated to the extent suggested by bi-cross validation and parallel analysis in Figure~\ref{Fig:EstK} (see Figure~\ref{Fig:Supp:Sim} in the Supplement).\par
\indent The anomalous behavior of IRW-SVA in Figure~\ref{Figure:Sim:Beta} is contingent on the size of $\bm{\Omega}$ and the eigenvalues $\gamma_1,\ldots,\gamma_K$. Instead of using the paradigm employed in BCconf, Cate-RR and dSVA of separating $\bm{\beta}$ from $\bm{L}$ by estimating $\bm{L}$ and $\bm{C}_{\perp}$ from the residuals $\bm{Y}_2$ and $\bm{\Omega}$ from $\hat{\bm{L}}$ and $\bm{Y}_1$ (with $\bm{G} = I_n$ in \eqref{equation:DataSplitting}), IRW-SVA circumvents estimating $\bm{\Omega}$ with the noisy design matrix $\hat{\bm{L}}$ by first identifying the genes with no main effect (i.e. $\beta_g=0$) and then estimating $\bm{C}$ with factor analysis on the reduced data matrix. If $K$ is known, the proof of Theorem \ref{theorem:CorrPCA} in the Supplement suggests that if sufficiently many null genes were correctly identified, $\bm{\Omega}$ is sufficiently small, $\gamma_K$ is sufficiently large and $n^{-1}\bm{C}^T \bm{V}_* \bm{C}$ is approximately a multiple of $n^{-1}\bm{C}^T \bm{C}$, then IRW-SVA can accurately recover $\bm{C}$ when samples are correlated. This is not necessarily the case for BCconf, Cate-RR and dSVA, since these methods must additionally use $\hat{\bm{L}}$ to estimate $\bm{\Omega}$, which can be inaccurate depending on the accuracy of $\hat{\bm{L}}$. However, when both $\gamma_k$ and the $k^{\text{th}}$ column of $\bm{\Omega}$ are moderate to large, IRW-SVA attributes a disproportionate amount of the variability in $\bm{C}$ as arising from the direct effect of treatment on expression compared to the other three methods. We refer the reader to Section 5.3 of \cite{CATE} for a more detailed discussion.

\subsection{Modifying existing methods that estimate $K$ to account for sample correlation}
\label{subsection:SimData:ModK}
Given the importance of the choice of $K$ in estimating $\bm{\beta}$ and factor analysis in general \citep{ChooseK1,ChooseK2,PermChooseK}, we discuss three possible adjustments one might suggest to attempt to ameliorate bi-cross validation's and parallel analysis' estimates for $K$ in these simulated data. The simplest would be to merely estimate $K$, and subsequently $\bm{C}$, separately for each tissue, since gene expression is assumed to be independent across individuals. If the data within each tissue were sufficiently informative for $K$ and $\bm{C}$, this procedure should estimate the within-tissue factor dimension to be 10 and $\hat{\bm{C}}_{t} \approx \bm{C}_{t} \bm{R}_{t}$ where $\bm{C}_{t}$ is $\bm{C}$ restricted to the $t^{\text{th}}$ tissue and $\bm{R}_{t}$ is an invertible matrix ($t=1,2,3$). The final estimate for $K$ would then be $3 \times 10=30$ and $\hat{\bm{C}} \approx \bm{\Pi}\left( \bm{C}_{1} \bm{R}_{1} \oplus \bm{C}_{2}\bm{R}_{2} \oplus \bm{C}_{3}\bm{R}_{3}\right)$, where $\bm{\Pi}$ is a permutation matrix that reorders (individual, tissue) pairs. While $\hat{\bm{C}}$ will give approximately unbiased estimates for $\bm{\beta}$, there will be a reduction in the residual degrees of freedom, and therefore power. However, analyzing each tissue separately effectively reduces the sample size (and therefore the eigenvalues $\gamma_1,\ldots,\gamma_{K}$) by 67\% in this simulation example, which is why depending on the analyzed tissue, bi-cross validation and parallel analysis only estimate the within-tissue factor dimension to be anywhere from 1 to 4, which is a marked underestimate of $K$.\par
\indent To discuss the remaining two adjustments, we note that in this simulation example, $\bm{V}_g \approx a_gI_n + b_g \bm{B} = a_gI_n + b_g \bm{Z}_B \bm{Z}_B^T$, where $\bm{B}$ is a partition matrix that groups the $n$ samples into $n/3$ individuals and the columns of $\bm{Z}_B \in \mathbb{R}^{n \times n/3}$ are indicators specifying from which individual the sample originated. The second alteration would be to include $\bm{Z}_B$ in the set of nuisance covariates and restrict $\bm{Y}$ to the set of $2n/3$ within-individual contrasts. However, this effectively reduces the sample size from $n$ to $2n/3$ and shrinks the $\gamma_k$'s by at least 33\%, thereby making it harder to differentiate the latent signal $\bm{L}\bm{C}_{\perp}^T$ from the noise. In fact, bi-cross validation and parallel analysis had median estimates of $K$ equal to 8 and 6, respectively, using this alteration. The third adjustment, which avoids dramatically reducing the sample size, is to rotate $\bm{Y}_2$ by the eigenvectors of $\bm{Q}_{[X\, Z]}^T\bm{B}\bm{Q}_{[X\, Z]}$, which in this simulation example shrinks the between-sample dependence but increases the heterogeneity of the sample-specific residual variances. Since parallel analysis only compares the singular values of $\bm{Y}_2$ with singular values under a bootstrapped null model, rotating $\bm{Y}_2$ did not change parallel analysis' estimates for $K$. On the other hand, bi-cross validation's estimates for $K$ should change because cross validation will not be as sensitive to correlations between samples. However, its estimates for $\bm{C}_{\perp}$ will still be inaccurate because of the heterogeneity in sample-specific residual variances, which is why its median estimate for $K$ was only 5 in this simulation example.

\section{Sex-specific DNA methylation patterns in a longitudinal twin study}
\label{section:RealData}
We next applied our method to identify sex-specific DNA methylation patterns from a longitudinal twin study using data previously published in \cite{Martino}. The authors measured the DNA methylation of 10 monozygotic (MZ) and 5 dizygotic (DZ) Australian twin pairs (all DZ twins were both male or female) at birth and 18 months on the Infinium HumanMethylation450 BeadChip platform in buccal epithelium, a relatively homogeneous tissue. After probe and sample quality filtering and data-normalization, the authors were left with $p = 330,168$ methylation sites (CpGs) whose methylation was quantified in 29 male and 24 female ($n = 53$) samples as the difference between log-methylated and log-unmethylated probe intensity \citep{Mvalue} (see \cite{Martino} for all pre-processing steps). We then used our proposed method to choose $K$ and estimate $\bm{\beta}$ (CBCV-CorrConf), BCconf, Cate-RR, dSVA and IRW-SVA to identify sex-associated CpGs (CpGs whose methylation differed in males and females), and subsequently validated each method's findings using sex-associated CpGs identified at birth in previous studies with substantially larger sample sizes. We did not compare our method with RUV-2 or RUV-4, since we did not have access to control CpGs.\par 
\indent We first show that we can write the covariance of the 53 observations at each CpG as a linear combination of six positive semi-definite matrices. Let $y_{m,t,a}$ be the measured DNA methylation for twin $t \in \left\lbrace 1,2 \right\rbrace$, from mother $m \in [15]$ at age $a \in \left\lbrace 0, \text{18m} \right\rbrace$, where samples with different mothers were assumed to be independent and twin 1 and twin 2 from the same mother were assumed to be exchangeable. A preliminary analysis showed that the correlation between MZ and DZ twins at both ages was approximately the same, which is consistent with the observation that methylation patterns are in large part determined by environmental exposures \citep{eLifeGalanter,MethEnv}. Therefore, the $4 \times 4$ covariance matrix for $\bm{y}_m = \left(y_{m,1,0} \, y_{m,1,\text{18m}} \, y_{m,2,0} \, y_{m,2,\text{18m}} \right)^T$ completely determined the $n \times n$ covariance matrix for each CpG. We avoid assuming a generative model for $\C\left( \bm{y}_m\right)$ by only making assumptions on the pairwise covariances, which averts potential biases in our estimate for $\bm{V}_*$, and therefore $\bm{C}$. First, we would expect the covariance between observations made on the same individual (or sample) to be at least as large as those made on different individuals (or samples). Second, one might also expect that the shared variance for twins at the same age be as least as large as that at different ages. That is, for $a_1 \neq a_2$ and $t_1 \neq t_2$,
\begin{subequations}
\label{equation:MartinoCovariance}
\begin{align}
0\leq \C\left( y_{m,t_1,a_1}, y_{m,t_2,a_2}\right) \leq \C\left( y_{m,t_1,a_1}, y_{m,t_1,a_2}\right) &\leq \V\left( y_{m,t_1,a_i}\right) \,\,\, (i=1,2)\\
0 \leq \C\left( y_{m,t_1,a_1}, y_{m,t_2,a_2}\right) \leq \C\left( y_{m,t_1,a_i}, y_{m,t_2,a_i}\right) & \leq \V\left( y_{m,t_1,a_i}\right) \,\,\, (i=1,2).
\end{align}
\end{subequations}
Therefore, we can write the covariance matrix for $\bm{y}_m$ as
\begin{align*}
\C\left(\bm{y}_m\right) =& v_{\alpha} \bm{1}_4 \bm{1}_4^T + v_{\eta}\left(\bm{1}_2\bm{1}_2^T \oplus \bm{1}_2\bm{1}_2^T\right) + v_{\phi,0} \bm{a}_0 \bm{a}_0^T + v_{\phi,\text{18m}} \bm{a}_{\text{18m}}\bm{a}_{\text{18m}}^T + v_{0} \diag\left( \bm{a}_0\right) + v_{\text{18m}} \diag\left( \bm{a}_{\text{18m}}\right)
\end{align*}
where $\bm{a}_{0}^T = \left(1 , 0 , 1 , 0 \right)$, $\bm{a}_{\text{18m}}^T = \left(0 , 1 , 0 , 1 \right)$ and
\begin{align*}
& v_{\alpha} = \C\left( y_{m,t_1,0}, y_{m,t_2,\text{18m}}\right)\\
& v_{\alpha} + v_{\eta} = \C\left( y_{m,t_1,0}, y_{m,t_1,\text{18m}}\right), \,\, v_{\alpha} + v_{\phi,a} = \C\left( y_{m,t_1,a}, y_{m,t_2,a}\right) \,\,\, (a=0, \text{18m})\\
& v_{\alpha} + v_{\eta} + v_{\phi,a} + v_{a} = \V\left( y_{m,t_1,a}\right)\,\,\, (a=0, \text{18m}).
\end{align*}
By \eqref{equation:MartinoCovariance}, the variance multipliers also lie in a convex polytope that can be written in the form of \eqref{equation:Theta}, and are such that
\begin{align*}
v_{\alpha}\geq 0, \, v_{\eta} \geq 0, \, v_{\phi,a} \geq 0, \, v_{\eta}+v_{a}\geq 0,  \, v_{\phi, a}+v_{a} \geq 0 \,\,\, (a = 0, \text{18m}),
\end{align*}
meaning we can apply Algorithm \ref{algorithm:EsimationProcedure} to estimate $K$ and $\bm{C}$ and subsequently identify sex-specific DNA methylation patterns.\par
\indent Since there was no evidence that the difference in methylation between males and females changed from birth to 18 months, we assumed the methylation at each CpG was a linear combination of the subject's age (birth or 18 months), sex and other unobserved factors to be estimated, where age was a nuisance covariate and sex was the phenotype of interest. We first used CBCV with $F=5$ folds and estimated $K$ to be 2 and then estimated $\bm{C}$ and $\bm{V}_*$ with CorrConf. Our estimates for the six average variance multipliers were all strictly greater than 0 and were consistent with previous observations that one's methylome reflects one's environmental exposures \citep{eLifeGalanter,MethEnv}. That is, the average residual variance at 18 months was 25\% larger than that at birth and the correlation between methylation for twins at 18 months was nearly 20\% larger than that at birth, indicating this set of twin's methylomes tended to converge over the first 18 months of life.\par
\indent We next computed each of the other four method's estimates for $\bm{C}$ using each method's default software to choose $K$: bi-cross validation, the default for Cate-RR and BCconf, or parallel analysis, the default for dSVA and IRW-SVA. Using the full $p \times n$ data matrix, bi-cross validation and parallel analysis estimated $K$ to be 4 and 15, respectively (CBCV also estimated $K$ to be 4 when we applied it assuming $\bm{V}_g$ was a multiple of the identity). The fact that both estimated the latent factor dimension to be greater than CBCV's estimate of 2 is not surprising, as these methods tend to overestimate $K$ when samples are correlated. Since these methods are not designed for dependent data and given the complexity of the sample correlation structure for samples from infants with the same mother, we discuss two adjustments designed to alleviate potential biases in their estimates for $K$ and subsequent test statistics. The first is to split the data matrix into a set of samples measured at birth and another set measured at 18 months, and subsequently rotate the two data matrices to nullify between-twin correlations, which should help to mitigate biases in bi-cross validation's estimates for $K_0$ and $K_{\text{18m}}$ (the number of latent factors at birth and 18 months), but leaves parallel analysis' estimates unchanged. While data splitting removes between-individual correlations, it effectively reduces the sample size by 50\% when estimating $\bm{C}$ because we are forced to estimate the latent factors at birth and 18 months separately. Bi-cross validation estimated $K_0=3$, $K_{\text{18m}}=2$ and parallel analysis estimated $K_0=9$, $K_{\text{18m}}=7$. Lastly, one could split the data by age and twin id into four data matrices, which would ostensibly eliminate the correlation between samples. However, since twin ids are arbitrary, estimates for $K$, $\bm{C}$ and subsequently $\bm{\beta}$ were heavily dependent on how twins were grouped, so we did not include comparisons with this data splitting technique.\par
\indent Once we estimated $\bm{C}$ for all five methods, we estimated the effect due to sex on methylation, and corresponding q-values to control the false discovery rate, exactly as we did for the simulated data in Section \ref{section:SimulatedData} and deemed a CpG a sex-associated CpG if its q-value in that method was no greater than 0.2. Since we obviously did not know the ground truth, we used sex-associated CpGs identified at birth in \cite{BMCValidation} and \cite{NatureValidation} as a validation set to help judge the veracity of each method's findings. \cite{BMCValidation} and \cite{NatureValidation} measured DNA methylation in umbilical cord blood on the Infinium HumanMethylation450 BeadChip platform in children born to 111 unrelated Brazilian and 71 unrelated Mexican American mothers, respectively. The authors of both studies measured and corrected for cord blood cellular composition and identified 2,355 and 1,928 sex-associated CpGs that were also among the 330,168 CpGs studied in \cite{Martino}. Table \ref{Table:RealData} gives the fraction of sex-associated CpGs identified using $\hat{\bm{C}}$ estimated with our method (CBCV-CorrConf), along with the other four methods applied to the full $p \times n$ data matrix, that are also among the 3,532 sex-associated CpGs identified in \cite{NatureValidation} or \cite{BMCValidation}. BCconf's, Cate-RR's, dSVA's and IRW-SVA's results were nearly identical when we used the data splitting method described in the previous paragraph.

\RealData

\indent While it may be the case that most of BCconf, Cate-RR, dSVA and IRW-SVA are actual sex-associated CpGs, the results in Table \ref{Table:RealData} mirror the trends observed in Figure~\ref{Figure:Sim:Beta} (as well as those observed in Figure~\ref{Fig:Supp:Sim} in the Supplement). That is, while these four methods nominally identify more sex-associated CpGs, we are less confident in their results because their estimates for the latent factors reduce the residual variance but likely do not suitably account for the variability in sex explainable by $\bm{C}$, thereby making their results less reproducible.\par 
\indent These results also highlight the importance of the choice of $K$. Estimating $K$ with CBCV (and cross-validation in general) tends to yield more reproducible results because we only include a latent factor if prediction performs suitably well on new, held-out data. When we applied all five methods with $K=2$, BCconf, Cate-RR and dSVA performed similarly with overlaps no greater than 30\% (349 out of dSVA's 1168 sex-associated CpGs ($< 30\%$) were in the validation set). However, 272 out of IRW-SVA's 662 sex-associated CpGs (41\%) were in the validation set, which is nearly identical to CorrConf's results in Table~\ref{Table:RealData}. Similarly, when we set $K=4$ for all methods, dSVA performed nearly identically to Cate-RR, whereas CorrConf and IRW-SVA had overlaps of 27\% and 26\%, respectively, and both ostensibly identified approximately 1,500 sex-associated CpGs. The similarity between CorrConf and IRW-SVA in this dataset is not surprising, since the estimated $\bm{C}$, $\gamma_K$, $\bm{\Omega}$ and $\bm{V}_*$ satisfied the sufficient conditions for IRW-SVA to accurately recover $\bm{C}$ discussed at the end of Section \ref{subsection:SimData:EstKandC}. We believe $K=2$ is the most appropriate choice of $K$ for this dataset because CorrConf's estimate for $\bm{C}$ appears to explain enough of the variance in methylation to achieve reasonable power, while also accurately recovering $\bm{\Omega}$ to control for false discoveries and ensuring the results are reproducible.


\section{Discussion}
\label{section:Discussion}
To the best of our knowledge, we have provided the first method to account for latent factors in high dimensional data with correlated observations. We proved that our proposed method correlated bi-cross validation (CBCV) tends to accurately choose the latent factor dimension and that our estimate for $\bm{C}$ is accurate enough so that asymptotically, inference on the main effect $\bm{\beta}$ is just as accurate as when $\bm{C}$ is known. We also demonstrated our method's finite sample properties by analyzing complex, multi-tissue simulated gene expression data, and also used a real longitudinal DNA methylation data from a twin study to show our method tends to give more reproducible results compared to other state of the art methods.\par
\indent Our proposed procedure is certainly not a panacea for data with arbitrary correlation structure, and relies on the residual variance $\bm{V}_g$ being a linear combination of known positive semi-definite matrices. Data with more complex, non-linear sample correlation structure like lengthy auto regressive processes may not be amenable to \eqref{equation:Model_var} without requiring $b$ to be large, since a linear combination of $p$ non-linear functions will not necessarily have an apriori known functional form. It may be possible to use the intuition we have developed to first estimate the latent factors with the largest effects, estimate the average unit-specific covariances $\bm{V}_1,\ldots,\bm{V}_p$ and use these to subsequently fix $\bm{V}_*$ to approximate $K$ and $\bm{C}$, since ICaSE, CBCV and \eqref{equation:Cestimate} only depend on the unit-specific variances through $\bm{V}_*$. This may be an interesting area of future application.\par 
\indent Another important point is the estimation model we use in \eqref{equation:WrongModel} is not the only way to approximate the data generating model. In fact, one could argue that $\bm{Y} \approx MN_{p \times n}\left\lbrace \bm{\mu}, \diag\left( \delta_1^2, \cdots,\delta_p^2 \right), \bm{V}\right\rbrace$ is a better approximation to \eqref{equation:Model_yANDvar}. However, ICaSE, and subsequently CBCV, have substantially longer run times with this model (and are intractable with ultra high dimensional DNA methylation data), since we need to repeatedly manipulate and decompose a large $p \times n$ matrix as opposed to the relatively small $n \times n$ empirical sample covariance matrix $p^{-1}\bm{Y}^T \bm{Y}$. We also did not see any improvement with this model over model \eqref{equation:WrongModel} in simulations.

\section{Software}
\label{section:Software}

An R package is available for download at https://github.com/chrismckennan/CorrConf. To install, type the following into the R console:
\begin{verbatim}
install.packages("devtools")
devtools::install_github("chrismckennan/CorrConf/CorrConf")
library(CorrConf)
\end{verbatim}

\section*{Acknowledgments}

We thank Carole Ober and Marcus Soliai for providing data and for comments that have substantially improved the manuscript. This research is supported in part by NIH grants R01-HL129735 and R01-MH101820.


\bibliographystyle{rss}
\bibliography{References}

\newpage

\endgroup

\setcounter{equation}{0}
\setcounter{theorem}{0}
\setcounter{lemma}{0}
\setcounter{corollary}{0}
\setcounter{figure}{0}
\renewcommand{\thefigure}{S\arabic{figure}}
\renewcommand{\theequation}{S\arabic{equation}}
\renewcommand{\thesection}{S\arabic{section}}
\renewcommand{\thetheorem}{S\arabic{theorem}}
\renewcommand{\thelemma}{S\arabic{lemma}}
\renewcommand{\thecorollary}{S\arabic{corollary}}

\section{Supplementary Material}
\label{Section:Supplement}
\subsection{Additional simulation results}
\label{subsection:Supp:SimResults}
Here we include additional results regarding the estimation of the effect of interest, $\bm{\beta}$, from the simulations discussed in Section \ref{section:SimulatedData} when $K$ is not known to be 10. Since all competing methods perform worse when $K$ is underestimated, we estimated $\bm{\beta}$ using BCconf, Cate-RR, dSVA, IRW-SVA, RUV-2 and RUV-4 in the 50 simulated datasets with $\hat{K}$ set to 61, which was the median of all 20 estimates for $K$ using bi-cross validation. We chose the median bi-cross validation estimate for $K$ as opposed to the median parallel analysis estimate so that each method had an opportunity to estimate the factors with smallest $\gamma_k$'s, which were not recoverable with parallel analysis' estimate for $K$. The results are given in Figure~\ref{Fig:Supp:Sim}, which shows no competing method is able to control for false discovery. Even BCconf, which is the only competing method able to come close to controlling false discovery, has less power than CorrConf to detect true signals. This is presumably because of bi-cross validation's large estimate of $K$, which reduces residual degrees of freedom and increases the estimated variation in $\bm{X}$ attributable to $\bm{C}$ (i.e. increases $\hat{\bm{\Omega}}$).\par 

\begin{figure}[ht]
\centering
\subcaptionbox{ \label{Figure:Supp:Sim:CtC}}{%
    \includegraphics[scale=0.35]{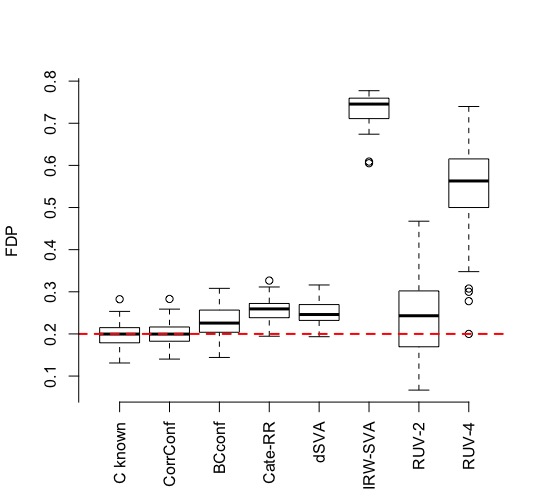} \hspace*{0.5em}%
  }
\subcaptionbox{ \label{Figure:Supp:Sim:Beta}}{%
    \includegraphics[scale=0.35]{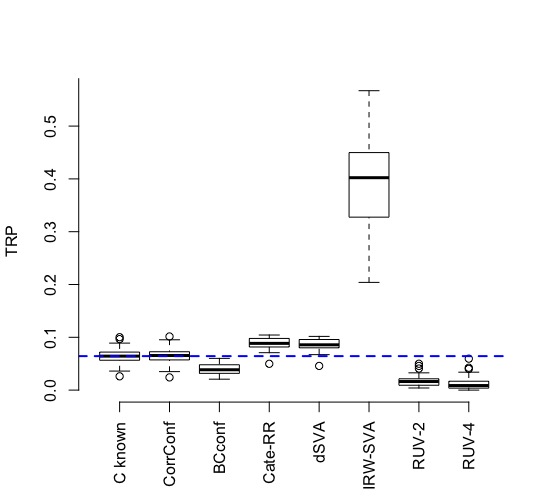}
  }
\caption{The false discovery proportion, (a), and true recovery proportion, (b), among genes with a q-value no greater than 0.2 in 50 simulated datasets, where $\text{TRP} = \left\lbrace \text{\# of discoveries} \right\rbrace/\left\lbrace \text{Total \# of genes with $\beta_g \neq 0$} \right\rbrace$. The dashed red and blue lines in (a) and (b) are $\text{FDP} = 0.2$ and $\text{TRP} = \left\lbrace\text{Median TRP when $\bm{C}$ is known}\right\rbrace$. We set $\hat{K}=10$ in CorrConf (which was CBCV's median estimate for $K$), and $\hat{K}=61$ in BCconf, Cate-RR, dSVA, IRW-SVA, RUV-2 and RUV-4.}\label{Fig:Supp:Sim}
\end{figure}


\newpage
\subsection{Proofs of all propositions, lemmas and theorems}
\label{subsection:Supp:Proofs}
Unless otherwise stated, we use the notation that a matrix or vector $\bm{A} = O_P\left( a_n\right)$ if $\norm{\bm{A}}_2 = O_P\left( a_n\right)$. We also define
\begin{align*}
\tilde{\bm{B}}_j = \bm{Q}_X^T \bm{B}_j \bm{Q}_X \quad (j=1,\ldots,b).
\end{align*}
\label{subsection:Proofs}


\subsubsection{Proof of Proposition \ref{proposition:KLVdelta}}
\label{subsubsection:Proof:KLVdelta}

\begin{proof}[of Proposition \ref{proposition:KLVdelta}]
Since the mean of $F_{\eqref{equation:WrongModel}, \left( \bm{\tau}, \delta^2\right)}$ that minimizes the KL divergence is the mean of $F_{\eqref{equation:Model_yANDvar}}$, we can write the KL divergence up to constants that do not depend on $\bm{\tau}$ or $\delta^2$ as
\begin{align*}
KL\left\lbrace F_{\eqref{equation:Model_yANDvar}}  \mid \mid F_{\eqref{equation:WrongModel}, \left( \bm{\tau}, \delta^2\right)}\right\rbrace = -\log \abs{\bm{Z}} + \Tr\left( \bm{Z}\bar{\bm{V}}\right)
\end{align*}
where $\bm{Z} = \left( \delta^2 \bm{V}\right)^{-1}$ and $\bar{\bm{V}} = p^{-1}\sum\limits_{g=1}^p \bm{V}_g$. This attains its only minimum when $\delta_*^2 \bm{V}_* = \bm{Z}_*^{-1} = \bar{\bm{V}}$, since the KL-divergence is strictly convex in $\bm{Z}$. We then have
\begin{align*}
\sum\limits_{j=1}^b \delta_*^2 \tau_{*_j}\bm{B}_j = \delta_*^2 \bm{V}_* = p^{-1}\sum\limits_{g=1}^p \left( \sum\limits_{j=1}^b v_{g,j} \bm{B}_j\right) = \sum\limits_{j=1}^b \left( p^{-1}\sum\limits_{g=1}^p v_{g,j} \right)\bm{B}_j
\end{align*}
Since the matrix $\bm{H}$ defined in the statement of Proposition \ref{proposition:KLVdelta} is positive definite, $\left\lbrace \vecM\left( \bm{B}_j\right) \right\rbrace_{j=1}^b$ is a linearly independent set, which proves \eqref{equation:KLVResult}. To prove $\bm{\tau}_* \in \Theta$,
\begin{align*}
\bm{A}_{\mathcal{I}}\bm{\tau}_* = \delta_*^{-2}p^{-1}\sum\limits_{g=1}^p \bm{A}_{\mathcal{I}}\bm{v}_g \geq \bm{0}
\end{align*}
since $\bm{A}_{\mathcal{I}}\bm{v}_g \geq \bm{0}$ and
\begin{align*}
\bm{A}_{\mathcal{E}}\bm{\tau}_* = \delta_*^{-2}p^{-1}\sum\limits_{g=1}^p \bm{A}_{\mathcal{E}}\bm{v}_g = \bm{0}
\end{align*}
since $\bm{A}_{\mathcal{E}}\bm{v}_g = \bm{0}$. This completes the proof.
\end{proof}

\subsubsection{Proof of Theorem \ref{theorem:ICaSE}}
\label{subsubsection:Proof:ICaSE}
\indent We next prove Theorem \ref{theorem:ICaSE}, which will be subsequently used in the proofs of Theorems \ref{theorem:CBCV} and \ref{theorem:CGLS} and Lemma \ref{lemma:OmegaWLS}. For ease of notation we will assume in this section that the data $\bm{Y}$ are generated according to model \eqref{equation:Model_yANDvar} with $d=0$ and will use the incorrect, but simpler model in \eqref{equation:WrongModel} with $d=0$ as well in the proofs of Lemma \ref{lemma:EigsEtE}, Theorem \ref{theorem:CorrPCA}, Lemma \ref{lemma:Vhat} and Theorem \ref{theorem:ICaSE} below. We first prove a lemma about the extreme singular values of a Gaussian random matrix with independent rows.

\begin{lemma}
\label{lemma:EigsEtE}
Let $\bm{E} \in \mathbb{R}^{p \times n}$ be a random matrix with independent rows where the $g^{\text{th}}$ row is distributed as $\bm{e}_g \sim N\left( 0, \bm{V}_g \right)$ where $p^{-1}\sum\limits_{g=1}^p \bm{V}_g = \bar{\bm{V}}$ and suppose Assumption \ref{assumption:Basics} holds. Then
\begin{align*}
\norm{p^{-1}\bm{E}^T \bm{E} - \bar{\bm{V}}}_2 = O_P\left( n^{1/2}p^{-1/2}\right).
\end{align*}
\end{lemma}
\begin{proof}
The proof of this is a simple extension of Theorem 5.39 in \citet{Vershynin} and is omitted.
\end{proof}

\begin{theorem}
\label{theorem:CorrPCA}
Suppose Assumptions \ref{assumption:Basics} and \ref{assumption:ICaSE} hold and define $\epsilon = \norm{\hat{\bm{V}} - \bm{V}_*}_2$, where $\hat{\bm{V}}$ is an estimate of $\bm{V}_*$ with $\log \abs{\hat{\bm{V}}} = 0$, and
\begin{subequations}
\label{equation:CbarQbarlambda}
\begin{align}
\label{equation:Cbar}
\bar{\bm{C}} &= \hat{\bm{V}}^{-1/2}\bm{C}\left( \bm{C}^T \hat{\bm{V}}^{-1}\bm{C} \right)^{-1/2}\\
\label{equation:Qbar}
\bar{\bm{Q}} &= \bm{Q}_{\bar{\bm{C}}}\\
\label{equation:lambda}
\lambda_k &= \bm{\lambda}_k\left( p^{-1}\hat{\bm{V}}^{-1/2}\bm{C} \bm{L}^T \bm{L}\bm{C}^T \hat{\bm{V}}^{-1/2} \right).
\end{align}
\end{subequations}
where $\bm{\lambda}_k\left( \bm{A}\right)$ is the $k^{\text{th}}$ largest eigenvector of the matrix $\bm{A}$. Define the estimates $\hat{\bar{\bm{C}}} \in \mathbb{R}^{n \times K}$ and $\hat{\lambda}_1,\ldots,\hat{\lambda}_K$ to be the first $K$ eigenvectors and eigenvalues of $\hat{\bm{V}}^{-1/2}p^{-1}\bm{Y}^T\bm{Y}\hat{\bm{V}}^{-1/2}$, respectively. Then if $\epsilon/\gamma_K = o_P(1)$,
\begin{align}
\label{equation:corrPCA:subspace}
\norm{ P_{\bar{\bm{C}}} - P_{\hat{\bar{\bm{C}}}} }_F^2 &= O_P\left[\left\lbrace n^{1/2}\left( \gamma_K p\right)^{-1/2} + \epsilon \gamma_K^{-1}\right\rbrace^2\right].
\end{align}
Further, if $\epsilon = o_P(1)$,
\begin{align}
\label{equation:corrPCA:lambda}
&\hat{\lambda}_k/\lambda_k = 1 + \delta_*^2/\lambda_k + O_P\left\lbrace n\left( \gamma_K p\right)^{-1} + n^{1/2}\left( \gamma_K p\right)^{-1/2}\epsilon + \left( \gamma_K p\right)^{-1/2} + \epsilon \gamma_K^{-1}\right\rbrace\\
\label{equation:corrPCA:CtChat.IK}
&\norm{\bar{\bm{C}}^T \hat{\bar{\bm{C}}} - I_K}_2 = O_P\left\lbrace n\left( \gamma_K p\right)^{-1} + n^{1/2}\left( \gamma_K p\right)^{-1/2}\epsilon + \left( \gamma_K p\right)^{-1/2} + \epsilon \gamma_K^{-1}\right\rbrace.
\end{align}
\end{theorem}

\begin{proof}
First, $(\lambda_k - \lambda_{k+1})/\lambda_{k+1} \geq c_1^{-1} + o_P(1)$ when $\epsilon = o_P(1)$ by item \ref{item:assumptionBasics:V} of Assumption \ref{assumption:Basics}. Next, since $\bm{L}$ and $\bm{C}$ are unique up to a $K \times K$ invertible matrix, it suffices to re-define $\bm{L}$ as $\bm{L}\left( n^{-1}\bm{C}^T \hat{\bm{V}}^{-1}\bm{C} \right)^{+1/2}$ and assume $np^{-1}\bm{L}^T \bm{L} = \diag\left( \lambda_1, \ldots, \lambda_K \right)$. We let $\gamma = \gamma_K$ and use a technique developed in \citet{Debashis} and define the rotated matrix $\bar{\bm{S}}$ to be
\begin{align}
\label{equation:Debashis}
\bar{\bm{S}} &= \begin{pmatrix}
\bar{\bm{C}}^T\\
\bar{\bm{Q}}^T
\end{pmatrix}\left( \gamma p\right)^{-1}\hat{\bm{V}}^{-1/2}\bm{Y}^T \bm{Y} \hat{\bm{V}}^{-1/2} \begin{pmatrix}
\bar{\bm{C}} & \bar{\bm{Q}}
\end{pmatrix} = \left\lbrace\begin{matrix}
\left( \bar{\bm{L}} + \bar{\bm{E}}_1 \right)^T \left( \bar{\bm{L}} + \bar{\bm{E}}_1 \right) & \left( \bar{\bm{L}} + \bar{\bm{E}}_1 \right)^T \bar{\bm{E}}_2\\
\bar{\bm{E}}_2^T \left( \bar{\bm{L}} + \bar{\bm{E}}_1 \right) & \bar{\bm{E}}_2^T\bar{\bm{E}}_2
\end{matrix}\right\rbrace\\
\label{equation:LBar}
\bar{\bm{L}} &= n^{1/2}\left(\gamma p \right)^{-1/2}\bm{L}\\
\label{equation:E1}
\bar{\bm{E}}_1 &= \left(\gamma p \right)^{-1/2}\bm{E}\hat{\bm{V}}^{-1/2}\bar{\bm{C}}\\
\label{equation:E2}
\bar{\bm{E}}_2 &= \left(\gamma p \right)^{-1/2}\bm{E}\hat{\bm{V}}^{-1/2}\bar{\bm{Q}}.
\end{align}
We now get explicit error bounds for terms in $\bar{\bm{S}}$.
\begin{enumerate}
\item \begin{align*}
\bar{\bm{L}}^T \bar{\bm{E}}_1 &= \left(\gamma p \right)^{-1/2} \bar{\bm{L}}^T \bm{E}\hat{\bm{V}}^{-1}\left( n^{-1/2}\bm{C}\right)\underbrace{\left( n^{-1}\bm{C}^T \hat{\bm{V}}^{-1}\bm{C}\right)^{-1/2}}_{O_P(1)}\\
&= \left( \underbrace{\left(\gamma p \right)^{-1/2}\bar{\bm{L}}^T \bm{E}\bm{V}^{-1}_*\left( n^{-1/2}\bm{C}\right)}_{O_P\left\lbrace \left(\gamma p \right)^{-1/2}\right\rbrace} + \underbrace{\left(\gamma p \right)^{-1/2}\bar{\bm{L}}^T \bm{E} \left( \hat{\bm{V}}^{-1}-\bm{V}^{-1}_*\right)\left( n^{-1/2}\bm{C}\right)}_{O_P\left\lbrace n^{1/2}\left(\gamma p \right)^{-1/2} \epsilon\right\rbrace} \right)O_P(1).
\end{align*}
Therefore, $\norm{\bar{\bm{L}}^T \bar{\bm{E}}_1}_2 = O_P\left\lbrace \left(\gamma p \right)^{-1/2} + n^{1/2}\left(\gamma p \right)^{-1/2} \epsilon\right\rbrace$.
\item \begin{align*}
\bar{\bm{L}}^T \bar{\bm{E}}_2 &= \underbrace{\left(\gamma p \right)^{-1/2}\bar{\bm{L}}^T \bm{E}\hat{\bm{V}}^{-1/2} \hat{\bm{V}}^{1/2}\bm{Q}_{\bm{C}}}_{O_P\left\lbrace n^{1/2}\left(\gamma p \right)^{-1/2}\right\rbrace}\underbrace{\left( \bm{Q}_{\bm{C}}^T \hat{\bm{V}} \bm{Q}_{\bm{C}}\right)^{-1/2}}_{O_P(1)}.
\end{align*}
Therefore, $\norm{\bar{\bm{L}}^T \bar{\bm{E}}_2}_2 = O_P\left\lbrace n^{1/2}\left(\gamma p \right)^{-1/2}\right\rbrace$.
\item \begin{align*}
\bar{\bm{E}}_1^T \bar{\bm{E}}_1 = \left(\gamma p \right)^{-1} \bar{\bm{C}}^T \bm{V}^{-1/2}_* \bm{E}^T \bm{E}\bm{V}^{-1/2}_* \bar{\bm{C}} + O_P\left( \epsilon \gamma^{-1} \right).
\end{align*}
Define $\tilde{\bm{C}} = \bm{V}^{-1/2}_*\bm{C}\left( \bm{C}^T \bm{V}^{-1}_*\bm{C}\right)^{-1/2}$.
\begin{align*}
\left(\gamma p \right)^{-1} \bar{\bm{C}}^T \hat{\bm{V}}^{-1/2} \bm{E}^T \bm{E}\hat{\bm{V}}^{-1/2} \bar{\bm{C}} &= \gamma^{-1}\delta_*^2 I_K + \gamma^{-1} \bar{\bm{C}}^T \left( \hat{\bm{V}}^{-1/2} p^{-1}\bm{E}^T \bm{E}\hat{\bm{V}}^{-1/2} - \delta_*^2 I_n \right)\bar{\bm{C}}\\
&= \gamma^{-1}\delta_*^2 I_K + \underbrace{\gamma^{-1}\tilde{\bm{C}}^T\left( \bm{V}^{-1/2}_* p^{-1}\bm{E}^T \bm{E}\bm{V}^{-1/2}_* - \delta_*^2 I_n \right)\tilde{\bm{C}}}_{O_P\left( \gamma^{-1}p^{-1/2}\right)} + O_P\left( \epsilon \gamma^{-1}\right).
\end{align*}
Therefore,
\begin{align*}
\norm{\bar{\bm{E}}_1^T \bar{\bm{E}}_1 - \gamma^{-1}\delta_*^2 I_K}_2 = O_P\left( \epsilon \gamma^{-1} + \gamma^{-1}p^{-1/2}\right).
\end{align*}
\item Let $\tilde{\bm{Q}} = \bm{Q}_{\tilde{\bm{C}}}$. Then
\begin{align*}
\bar{\bm{E}}_1^T \bar{\bm{E}}_2 = \left( \gamma p\right)^{-1} \bar{\bm{C}}^T \hat{\bm{V}}^{-1/2}\bm{E}^T \bm{E}\hat{\bm{V}}^{-1/2}\bar{\bm{Q}} = \left( \gamma p\right)^{-1} \tilde{\bm{C}}^T \bm{V}^{-1/2}_*\bm{E}^T \bm{E}\bm{V}^{-1/2}_* \tilde{\bm{Q}} + O_P\left(\epsilon \gamma^{-1}\right).
\end{align*}
Let $\tilde{\bm{E}} = p^{-1/2}\bm{E}\bm{V}^{-1/2}_*, \tilde{\bm{E}}_1 = \tilde{\bm{E}} \tilde{\bm{C}}, \tilde{\bm{E}}_2 = \tilde{\bm{E}}\tilde{\bm{Q}}$ and $\bm{A}_g = \bm{V}^{-1/2}_* \bm{V}_g \bm{V}^{-1/2}_*$. Then the $k^{\text{th}}$ column of $\tilde{\bm{E}}_2^T \tilde{\bm{E}}_1$ ($k=1,\ldots,K$) is
\begin{align*}
\bm{b}_k &= p^{-1}\tilde{\bm{Q}}^T\tilde{\bm{E}}^T \tilde{\bm{E}} \tilde{\bm{c}}_k = \tilde{\bm{Q}}^T\left( p^{-1}\tilde{\bm{E}}^T \tilde{\bm{E}} - p^{-1}\sum\limits_{g=1}^p \bm{A}_g \right)\tilde{\bm{c}}_k + \tilde{\bm{Q}}^T \underbrace{p^{-1}\sum\limits_{g=1}^p \bm{A}_g}_{=\delta_*^2 I_n} \tilde{\bm{c}}_k\\
&= p^{-1}\sum\limits_{g=1}^p \tilde{\bm{Q}}^T\bm{A}_g^{1/2}\left( \bm{r}_g\bm{r}_g^T - \delta_*^2 I_n \right)\bm{A}_g^{1/2}\tilde{\bm{c}}_k = p^{-1}\sum\limits_{g=1}^p \bm{b}_{g,k}
\end{align*}
where $\bm{r}_g \sim N\left( \bm{0},I_n\right)$, $\bm{r}_g$ and $\bm{r}_{g'}$ are independent for $g \neq g'$ ($g=1,\ldots,p; g' = 1,\ldots,p$) and $\E\left(\bm{b}_k \right) = \bm{0}$. Therefore,
\begin{align*}
\E\left(\norm{\bm{b}_k}_2^2\right) = \sum\limits_{i=1}^{n-K} \E\left( \bm{b}_{ki}^2\right) = \sum\limits_{i=1}^{n-K}\V\left( \bm{b}_{ki}\right) = O\left\lbrace n \V\left( \bm{b}_{k1}\right)\right\rbrace = O\left( np^{-1}\right)
\end{align*}
where the second to last and last equalities follow because $\norm{\bm{A}_g}_2$ is uniformly bounded by item \ref{item:assumptionBasics:V} of Assumption \ref{assumption:Basics}. Therefore, $\norm{\bar{\bm{E}}_2^T\bar{\bm{E}}_1}_2 = O_P\left( \gamma^{-1} n^{1/2}p^{-1/2} + \epsilon \gamma^{-1} \right)$.
\end{enumerate}
\indent Let $\mu_k = \lambda_k/\gamma$ and define $\begin{pmatrix}
\hat{\bm{v}}_k^T & \hat{\bm{z}}_k^T
\end{pmatrix}^T \in \mathbb{R}^{n}$ to be the $k^{\text{th}}$ normalized eigenvector of $\bar{\bm{S}}$, where $\bm{v}_k \in \mathbb{R}^{K}$ and $\bm{z}_k \in \mathbb{R}^{n-K}$. All of this proves that $\hat{\mu}_k = \mu_k + \delta_*^2/\gamma + o_P(1)$ by Weyl's theorem. To prove sharper bounds, we set set up the eigenvalue equations
\begin{align*}
\hat{\mu}_k \hat{\bm{v}}_k &= \left( \bar{\bm{L}} + \bar{\bm{E}}_1 \right)^T \left( \bar{\bm{L}} + \bar{\bm{E}}_1 \right)\hat{\bm{v}}_k + \left( \bar{\bm{L}} + \bar{\bm{E}}_1 \right)^T \bar{\bm{E}}_2\hat{\bm{z}}_k\\
\hat{\mu}_k \bm{z}_k &= \bar{\bm{E}}_2^T \left( \bar{\bm{L}} + \bar{\bm{E}}_1 \right)\hat{\bm{v}}_k + \bar{\bm{E}}_2^T\bar{\bm{E}}_2\hat{\bm{z}}_k.
\end{align*}
A little algebra shows that
\begin{align}
\label{equation:Eigsv}
\hat{\mu}_k \hat{\bm{v}}_k &= \left\lbrace\left( \bar{\bm{L}} + \bar{\bm{E}}_1 \right)^T \left( \bar{\bm{L}} + \bar{\bm{E}}_1 \right) + \left( \bar{\bm{L}} + \bar{\bm{E}}_1 \right)^T \bar{\bm{E}}_2 \left( \hat{\mu}_k I_{n} - \bar{\bm{E}}_2^T\bar{\bm{E}}_2 \right)^{-1}\bar{\bm{E}}_2^T \left( \bar{\bm{L}} + \bar{\bm{E}}_1 \right)\right\rbrace\hat{\bm{v}}_k\\
\label{equation:Eigsz}
\hat{\bm{z}}_k &= \left( \hat{\mu}_k I_{n} - \bar{\bm{E}}_2^T\bar{\bm{E}}_2 \right)^{-1}\bar{\bm{E}}_2^T \left( \bar{\bm{L}} + \bar{\bm{E}}_1 \right)\hat{\bm{v}}_k
\end{align}
where $\hat{\mu}_k I_{n} - \bar{\bm{E}}_2^T\bar{\bm{E}}_2$ is invertible because $\hat{\mu}_k = \mu_k + \delta_*^2/\gamma + o_P(1)$ and by Lemma \ref{lemma:EigsEtE}. By what we showed above, we then have that
\begin{subequations}
\label{equation:corrPCA:EigsResults}
\begin{align}
\label{equation:EigsResults:zhat}
\norm{\hat{\bm{z}}_k}_2 &= O_P\left\lbrace n^{1/2}\left( \gamma p\right)^{-1/2} + \epsilon \gamma^{-1}\right\rbrace\\
\hat{\mu}_k &= \bm{\lambda}_k\left\lbrace \left( \bar{\bm{L}} + \bar{\bm{E}}_1 \right)^T \left( \bar{\bm{L}} + \bar{\bm{E}}_1 \right)\right\rbrace + O_P\left[\left\lbrace n^{1/2}\left( \gamma p\right)^{-1/2} + \epsilon \gamma^{-1}\right\rbrace^2\right] \nonumber\\
\label{equation:EigsResults:lambda.hat}
&= \mu_k + \delta_*^2 \gamma^{-1} + O_P\left\lbrace n\left( \gamma p\right)^{-1} + n^{1/2}\left( \gamma p\right)^{-1/2}\epsilon + \left( \gamma p\right)^{-1/2} + \epsilon \gamma^{-1}\right\rbrace\\
\label{equation:EigsResults:vhatv}
\norm{\hat{\bm{v}}_k - \bm{v}_k}_2 &= O_P\left[\left\lbrace n^{1/2}\left( \gamma p\right)^{-1/2} + \epsilon \gamma^{-1}\right\rbrace^2\right]\\
\label{equation:EigsResults:vhata}
\norm{\hat{\bm{v}}_k - \bm{a}_k}_2 &= O_P\left\lbrace n\left( \gamma p\right)^{-1} + n^{1/2}\left( \gamma p\right)^{-1/2}\epsilon + \left( \gamma p\right)^{-1/2} + \epsilon \gamma^{-1}\right\rbrace
\end{align}
\end{subequations}
where $\bm{v}_k$ is the $k^{\text{th}}$ eigenvector of $\left( \bar{\bm{L}} + \bar{\bm{E}}_1 \right)^T \left( \bar{\bm{L}} + \bar{\bm{E}}_1 \right)$ and $\bm{a}_k \in \mathbb{R}^K$ is the $k^{\text{th}}$ standard basis vector. Equation \eqref{equation:EigsResults:zhat} follow from Weyl's Theorem (i.e. $\hat{\mu}_k = \mu_k + \delta_*^2/\lambda_k + o_P(1)$) and \eqref{equation:EigsResults:lambda.hat} follows from Theorem 3.5 of \citet{Tang} and the assumption that $(\lambda_k - \lambda_{k+1})/\lambda_{k+1} \geq c_1^{-1}+o_P(1)$. Since the left and right singular vectors of any matrix are unique up to sign parity, a second application of Theorem 3.5 of \citet{Tang}, along with the assumption that $(\lambda_k - \lambda_{k+1})/\lambda_{k+1} \geq c_1^{-1}+o_P(1)$, proves \eqref{equation:EigsResults:vhatv} and \eqref{equation:EigsResults:vhata}. This proves \eqref{equation:corrPCA:lambda}.\par
\indent Define $\hat{\bm{v}} = \begin{pmatrix}
\hat{\bm{v}}_1 & \cdots & \hat{\bm{v}}_K
\end{pmatrix}$ and $\hat{\bm{z}} = \begin{pmatrix}
\hat{\bm{z}}_1 & \cdots & \hat{\bm{z}}_K
\end{pmatrix}$. Since $\hat{\bm{v}}^T\hat{\bm{v}} + \hat{\bm{z}}^T\hat{\bm{z}} = I_K$ and by \eqref{equation:EigsResults:zhat},
\begin{align*}
\norm{\hat{\bm{v}} \hat{\bm{v}}^T - I_K}_2 = O_P\left[\left\lbrace n^{1/2}\left( \gamma p\right)^{-1/2} + \epsilon \gamma^{-1}\right\rbrace^2\right].
\end{align*}
Then
\begin{align}
\label{equation:Cbarhat}
\hat{\bar{\bm{C}}} = \bar{\bm{C}}\hat{\bm{v}} + \bar{\bm{Q}}\hat{\bm{z}}
\end{align}
with
\begin{align*}
\norm{P_{\hat{\bar{\bm{C}}}} - P_{\bar{\bm{C}}}}_F^2 = 2K - 2\Tr\left[ \bar{\bm{C}}^T \hat{\bar{\bm{C}}} \hat{\bar{\bm{C}}}^T \bar{\bm{C}}\right] = 2K - 2\Tr\left[ \hat{\bm{v}} \hat{\bm{v}}^T\right] = O_P\left[\left\lbrace n^{1/2}\left( \gamma p\right)^{-1/2} + \epsilon \gamma^{-1}\right\rbrace^2\right],
\end{align*}
which proves \eqref{equation:corrPCA:subspace}. Lastly, \eqref{equation:corrPCA:CtChat.IK} follows from \eqref{equation:EigsResults:vhatv} and \eqref{equation:EigsResults:vhata}.
\end{proof}

\begin{center}
\line(1,0){425}
\end{center}

We note that because $\hat{\delta}^2 \hat{\bm{\tau}} \in \Theta_*$ by Assumption \ref{assumption:ICaSE}, we may assume that $(n-d)^{-1}\bm{C}_{\perp}^T \hat{\bm{W}}^{-1}\bm{C}_{\perp} = I_{n-K}$, $(n-d)p^{-1}\bm{L}^T \bm{L} = \diag\left( \lambda_1, \ldots,\lambda_K\right)$ where $\lambda_K \leq \cdots \leq \lambda_1$ and $\lambda_1, \lambda_K \asymp \gamma_K$ for any fixed estimates $\hat{\delta}^2, \hat{\bm{\tau}}$ in the proofs of Theorems \ref{theorem:ICaSE}, \ref{theorem:CBCV} and \ref{theorem:CGLS} because the statement of those theorems only involves the columns space of $\bm{C}$. Since we will eventually prove \eqref{equation:ConfSub:V} in Theorem \ref{theorem:ICaSE} when $K$ is known and
\begin{align*}
&n^{1/2}\norm{(n-d)^{-1}\bm{C}_{\perp}^T \bm{W}_*^{-1} \bm{C}_{\perp} - (n-d)^{-1}\bm{C}_{\perp}^T \hat{\bm{W}}_*^{-1} \bm{C}_{\perp}}_2 = o_P(1),\\
& n^{1/2}\norm{\bm{\Omega} \left\lbrace (n-d)^{-1}\bm{C}_{\perp}^T \hat{\bm{W}}_*^{-1} \bm{C}_{\perp}\right\rbrace^{1/2} - \bm{\Omega}}_2 = o_P(1).
\end{align*}
Therefore, we assume $(n-d)^{-1}\bm{C}_{\perp}^T \hat{\bm{W}}^{-1}\bm{C}_{\perp} = I_{n-K}$ and $(n-d)p^{-1}\bm{L}^T \bm{L}$ is diagonal with decreasing elements in the remainder of the supplement, including the proof of Lemma \ref{lemma:OmegaWLS}.

\begin{corollary}[used in the proof of Theorems \ref{theorem:ICaSE} and \ref{theorem:CGLS}]
\label{corollary:ConfSub:CtC}
Suppose the assumptions of Theorem \ref{theorem:CorrPCA} hold (including the assumption that $\epsilon = o_P(1)$) and let $\bm{M}_g \in \mathbb{R}^{n \times n}$ be a random symmetric positive definite matrix that is only a function of $\bm{e}_g$, the $g^{\text{th}}$ row of $\bm{E}$, with $\norm{\bm{M}_g}_2 = 1$ and smallest eigenvalue uniformly bounded away from 0. Define $\hat{\bm{C}} = n^{1/2}\hat{\bm{V}}^{1/2} \hat{\bar{\bm{C}}}$. Then using the same notation as the statement of Theorem \ref{theorem:CorrPCA},
\begin{align}
\label{equation:Cor:ChatMChat}
\norm{n^{-1}\hat{\bm{C}}^T \bm{M}_g \hat{\bm{C}} - n^{-1}\bm{C}^T \bm{M}_g \bm{C}}_2 &= O_P\left( \eta\right)\\
\label{equation:Cor:ChatMC}
\norm{n^{-1}\hat{\bm{C}}^T \bm{M}_g \bm{C} - n^{-1}\bm{C}^T \bm{M}_g \bm{C}}_2 &=  O_P\left( \eta\right)\\
\label{equation:Cor:Proj}
\norm{P_{\bm{C}} - P_{\hat{\bm{C}}}}_F^2 &= O_P\left( \eta\right).
\end{align}
where $\eta = n\left( \gamma_K p\right)^{-1} + \left( \gamma_K p\right)^{-1/2} + n^{1/2}\left( \gamma_K p\right)^{-1/2} \epsilon + \epsilon \gamma_K^{-1}$.
\end{corollary}
\begin{proof}
The proof utilizes objects defined in Theorem \ref{theorem:CorrPCA}. We see that
\begin{align*}
n^{-1}\hat{\bm{C}}^T \bm{M}_g \hat{\bm{C}} =&  \hat{\bar{\bm{C}}}^T \hat{\bm{V}}^{1/2}\bm{M}_g\hat{\bm{V}}^{1/2}\hat{\bar{\bm{C}}} = \hat{\bm{v}}^T \bar{\bm{C}}^T\hat{\bm{V}}^{1/2}\bm{M}_g\hat{\bm{V}}^{1/2} \bar{\bm{C}} \hat{\bm{v}} + \hat{\bm{v}}^T \bar{\bm{C}}^T\hat{\bm{V}}^{1/2}\bm{M}_g\hat{\bm{V}}^{1/2} \bar{\bm{Q}} \hat{\bm{z}} \\
& + \left( \hat{\bm{v}}^T\bar{\bm{C}}^T\hat{\bm{V}}^{1/2}\bm{M}_g\hat{\bm{V}}^{1/2} \bar{\bm{Q}} \hat{\bm{z}}\right)^T + \underbrace{\hat{\bm{z}}^T\bar{\bm{Q}}^T\hat{\bm{V}}^{1/2}\bm{M}_g\hat{\bm{V}}^{1/2} \bar{\bm{Q}} \hat{\bm{z}}}_{O_P\left[\left\lbrace n\left(\gamma_K p \right)^{-1} + \epsilon\gamma_K^{-1}\right\rbrace^2\right] = O_P\left(\eta \right)}.
\end{align*}
First,
\begin{align*}
\norm{\hat{\bm{v}}^T \bar{\bm{C}}^T\hat{\bm{V}}^{1/2}\bm{M}_g\hat{\bm{V}}^{1/2} \bar{\bm{C}} \hat{\bm{v}} - \bar{\bm{C}}^T\hat{\bm{V}}^{1/2}\bm{M}_g\hat{\bm{V}}^{1/2} \bar{\bm{C}}}_2 = O_P\left( \eta\right).
\end{align*}
Next, by the proof of Theorem \ref{theorem:CorrPCA},
\begin{align*}
\norm{\hat{\bm{z}} - \bar{\bm{E}}_2^T \left( \bar{\bm{L}} + \bar{\bm{E}}_1\right)}_2 = O_P\left( \eta\right).
\end{align*}
Therefore, we need only understand how 
\begin{align*}
&\bar{\bm{C}}^T\hat{\bm{V}}^{1/2}\bm{M}_g\hat{\bm{V}}^{1/2} \bar{\bm{Q}} \bar{\bm{E}}_2^T \left( \bar{\bm{L}} + \hat{\bm{E}}_1\right) = \left( n\gamma_K p\right)^{-1/2}\bm{C}^T \bm{M}_g \hat{\bm{V}}\bm{Q}_{\bm{C}} \left( \bm{Q}_{\bm{C}}^T \hat{\bm{V}} \bm{Q}_{\bm{C}}\right)^{-1}\bm{Q}_{\bm{C}}^T \bm{E}^T \bar{\bm{L}} \\
& + \gamma_K^{-1}n^{-1}\bm{C}^T \bm{M}_g \hat{\bm{V}}\bm{Q}_{\bm{C}} \left( \bm{Q}_{\bm{C}}^T \hat{\bm{V}} \bm{Q}_{\bm{C}}\right)^{-1} \bm{Q}_{\bm{C}}^T\left(p^{-1}\bm{E}^T\bm{E} \right)\hat{\bm{V}}^{-1}\bm{C}\\
=& \left( \gamma_K p n\right)^{-1/2}\bm{C}^T \bm{M}_g \bm{V}_*\bm{Q}_{\bm{C}} \left( \bm{Q}_{\bm{C}}^T \bm{V}_* \bm{Q}_{\bm{C}}\right)^{-1}\bm{Q}_{\bm{C}}^T \bm{E}^T \bar{\bm{L}}\\
& + \gamma_K^{-1}n^{-1}\bm{C}^T \bm{M}_g \bm{V}_*\bm{Q}_{\bm{C}} \left( \bm{Q}_{\bm{C}}^T \bm{V}_* \bm{Q}_{\bm{C}}\right)^{-1} \bm{Q}_{\bm{C}}^T\left(p^{-1}\bm{E}^T\bm{E} \right)\bm{V}_*^{-1}\bm{C} + O_P\left\lbrace \epsilon\gamma_K^{-1} + n^{1/2}\left( \gamma_K p\right)^{-1/2}\epsilon \right\rbrace 
\end{align*}
behaves. We can write the first term as
\begin{align*}
&\underbrace{\left( \gamma_K p\right)^{-1}\bm{C}^T \bm{M}_g \bm{V}_*\bm{Q}_{\bm{C}} \left( \bm{Q}_{\bm{C}}^T \bm{V}_* \bm{Q}_{\bm{C}}\right)^{-1}\bm{Q}_{\bm{C}}^T \bm{e}_g \bm{\ell}_g^T}_{O_P\left\lbrace n\left( \gamma_K p\right)^{-1}\right\rbrace} +\\
&\underbrace{\left( \gamma_K n p\right)^{-1/2}\bm{C}^T \bm{M}_g \bm{V}_*\bm{Q}_{\bm{C}} \left( \bm{Q}_{\bm{C}}^T \bm{V}_* \bm{Q}_{\bm{C}}\right)^{-1}\bm{Q}_{\bm{C}}^T \bm{E}_{-g}^T \bar{\bm{L}}_{-g}}_{O_P\left\lbrace \left( \gamma_K p\right)^{-1/2}\right\rbrace}\\
=& O_P\left( \eta\right)
\end{align*}
where the subscript $-g$ means we remove the $g^{\text{th}}$ row from the matrix. The second term can also be decomposed in the same way:
\begin{align*}
& \underbrace{\gamma_K^{-1}n^{-1}\bm{C}^T \bm{M}_g \bm{V}_*\bm{Q}_{\bm{C}} \left( \bm{Q}_{\bm{C}}^T \bm{V}_* \bm{Q}_{\bm{C}}\right)^{-1}\bm{Q}_{\bm{C}}^T\left(p^{-1}\bm{e}_g \bm{e}_g^T \right)\bm{V}_*^{-1}\bm{C}}_{O_P\left\lbrace n^{1/2} \left( \gamma_K p\right)^{-1} \right\rbrace} +\\
& \underbrace{\gamma_K^{-1}n^{-1}\bm{C}^T \bm{M}_g \bm{V}_*\bm{Q}_{\bm{C}} \left( \bm{Q}_{\bm{C}}^T \bm{V}_* \bm{Q}_{\bm{C}}\right)^{-1}\bm{Q}_{\bm{C}}^T\left(p^{-1}\bm{E}_{-g}^T\bm{E}_{-g} \right)\bm{V}_*^{-1}\bm{C}}_{O_P\left\lbrace \left( \gamma_K p\right)^{-1/2}\right\rbrace}\\
&= O_P\left( \eta\right).
\end{align*}
This completes the proof for \eqref{equation:Cor:ChatMChat}. The proof of \eqref{equation:Cor:ChatMC} is identical to the analysis above and is not shown.\par
\indent To prove \eqref{equation:Cor:Proj}, we see that
\begin{align*}
\norm{P_{\bm{C}} - P_{\hat{\bm{C}}}}_F^2 = 2I_K - 2\Tr\left( P_{\bm{C}}P_{\hat{\bm{C}}}\right) = 2I_K - 2\Tr\left( \bm{C}^T \hat{\bm{C}} \left( \hat{\bm{C}}^T\hat{\bm{C}}\right)^{-1} \hat{\bm{C}}^T \bm{C}\left( \bm{C}^T\bm{C}\right)^{-1}\right) = O_P\left( \eta\right).
\end{align*}
\end{proof}

\begin{corollary}
\label{corollary:ConfSub:S}
Suppose the conditions of Corollary \ref{corollary:ConfSub:CtC} hold. Then using the same notation as Theorem \ref{theorem:CorrPCA},
\begin{align}
\label{equation:DebashiRemoved}
\gamma_K \left[\bar{\bm{S}} - \sum\limits_{k=1}^K \hat{\mu}_k \begin{pmatrix}
\hat{\bm{v}}_k\\
\hat{\bm{z}}_k
\end{pmatrix} \begin{pmatrix}
\hat{\bm{v}}_k^T & \hat{\bm{z}}_k
\end{pmatrix}\right] = \begin{pmatrix}
\underbrace{o_p(1)}_{K\times K} & \underbrace{o_P(1)}_{K\times (n-K)}\\
\underbrace{o_P(1)}_{(n-K)\times K} & \underbrace{\bar{\bm{Q}}^T \hat{\bm{V}}^{-1/2}p^{-1}\bm{E}^T \bm{E}\hat{\bm{V}}^{-1/2}\bar{\bm{Q}} + o_P(1)}_{(n-K)\times (n-K)}
\end{pmatrix}
\end{align}
where $o_P(1)$ here is short for a rank $K$ (i.e. $K$ non-zero singular values) matrix whose 2-norm converges to 0 in probability.
\end{corollary}
\begin{proof}
We see that
\begin{align*}
\hat{\mu}_k \begin{pmatrix}
\hat{\bm{v}}_k\\
\hat{\bm{z}}_k
\end{pmatrix} \begin{pmatrix}
\hat{\bm{v}}_k^T & \hat{\bm{z}}_k
\end{pmatrix} = \hat{\mu}_k\begin{pmatrix}
\hat{\bm{v}}_k \hat{\bm{v}}_k^T & \hat{\bm{v}}_k \hat{\bm{z}}_k^T\\
\hat{\bm{z}}_k\hat{\bm{v}}_k^T & \hat{\bm{z}}_k\hat{\bm{z}}_k^T
\end{pmatrix}\,\,\, (k=1,\ldots,K).
\end{align*}
First, let $\bar{\bm{N}} = \bar{\bm{L}} + \bar{\bm{E}}_1$. In \eqref{equation:corrPCA:EigsResults} we defined $\bm{v}_k$ to be the $k^{\text{th}}$ eigenvector of $\bar{\bm{N}}^T \bar{\bm{N}}$. We then have that
\begin{align*}
\norm{\hat{\mu}_k \hat{\bm{v}}_k \hat{\bm{v}}_k^T - \bm{\lambda}_k\left( \bar{\bm{N}}^T \bar{\bm{N}}\right)\bm{v}_k\bm{v}_k^T}_2 = o_P\left( \gamma_K^{-1}\right),
\end{align*}
meaning
\begin{align*}
\norm{\bar{\bm{N}}^T \bar{\bm{N}} - \sum\limits_{k=1}^K \hat{\mu}_k \hat{\bm{v}}_k \hat{\bm{v}}_k^T}_2 = o_P\left( \gamma_K^{-1}\right).
\end{align*}
Next, we have that
\begin{align*}
\hat{\mu}_k \hat{\bm{z}}_k =& \hat{\mu}_k \left( \hat{\mu}_k I_{n-K} - \bar{\bm{E}}_2^T \bar{\bm{E}}\right)^{-1}\bar{\bm{E}}_2^T \bar{\bm{N}}\hat{\bm{v}}_k = \frac{\hat{\mu}_k}{\hat{\mu}_k - \delta_*^2/\gamma_K}\bar{\bm{E}}_2^T \bar{\bm{N}}\hat{\bm{v}}_k\\
&+ O_P\left\lbrace \gamma_K^{-1}\left(n^{1/2}p^{-1/2} + \epsilon\right)\left(n^{1/2}p^{-1/2}\gamma_K^{-1/2} + \epsilon \gamma_K^{-1}\right)\right\rbrace\\
=& \frac{\mu_k + \delta_*^2/\gamma_K}{\mu_k}\bar{\bm{E}}_2^T \bar{\bm{N}}\hat{\bm{v}}_k + o_P\left( \gamma_K^{-1} \right) = \bar{\bm{E}}_2^T \bar{\bm{N}}\hat{\bm{v}}_k + o_P\left( \gamma_K^{-1} \right)
\end{align*}
Therefore,
\begin{align*}
\norm{\bar{\bm{E}}_2^T \bar{\bm{N}} - \sum\limits_{k=1}^K \hat{\mu}_k\hat{\bm{z}}_k\hat{\bm{v}}_k^T}_2 = \norm{\bar{\bm{E}}_2^T \bar{\bm{N}}\left( I_K - \hat{\bm{v}}\hat{\bm{v}}^T\right)}_2 + o_P\left( \gamma_K^{-1} \right) = o_P\left( \gamma_K^{-1} \right).
\end{align*}
Lastly, $\norm{\hat{\bm{z}}_k}_2^2 = o_P\left( 1/\gamma_K\right)$, which completes the proof.
\end{proof}

\begin{center}
\line(1,0){425}
\end{center}

\indent We now prove a crucial lemma that states that we can accurately estimate $\bm{V}_*$ when the starting point is sufficiently close to $\bm{V}_*$.

\begin{lemma}
\label{lemma:Vhat}
Suppose Assumptions \ref{assumption:Basics} and \ref{assumption:ICaSE} hold and let $k \geq K$. Define $\hat{\bm{V}}$ to be an initial estimate of $\bm{V}_*$ such that $\norm{\hat{\bm{V}} - \bm{V}_*}_2 = \epsilon = o_P(1)$ and $\hat{\delta}^{(f) 2}, \hat{\bm{V}}^{(f)}$ to be the estimates after we run step \ref{ICaSE:Subspace} and then \ref{ICaSE:Variance} of Algorithm \ref{algorithm:ICaSE} using $\hat{\bm{V}}$ as a starting point. Then the maximum of step \ref{ICaSE:Variance} is such that
\begin{align*}
\norm{\hat{\bm{V}}^{(f)} - \bm{V}_*}_2 &= O_P\left( n^{-1} \right)\\
\hat{\delta}^{(f) 2} - \delta_*^2 &= O_P\left( n^{-1} \right).
\end{align*}
\end{lemma}
\begin{proof}
It suffices to prove this lemma by proving that $\norm{\delta^{(f) 2} \hat{\bm{V}}^{(f)} - \delta_*^2 \bm{V}_*} = \mathcal{O}_P\left( n^{-1}\right)$, since $\delta^2 = \exp\left\lbrace n^{-1}\log\abs{\delta^2 \bm{V}\left(\bm{\tau}\right)}\right\rbrace$ is Lipschitz continuous in $\delta^2 \bm{\tau}$ and bounded away from 0 in $\Theta_*$. Therefore, we ignore the requirement that $\log \abs{\bm{V}} = 0$ and re-define $\bm{\tau}$ to be $\bm{\tau} \leftarrow \delta^2 \bm{\tau}$ in the proof of this lemma. We continue to use the objects $\bar{\bm{C}}$, $\bar{\bm{Q}}$ and $\hat{\bar{\bm{C}}}$ defined in \eqref{equation:Cbar}, \eqref{equation:Qbar} and \eqref{equation:Cbarhat}, and also define $\hat{\bm{Q}} = \bm{Q}_{\hat{\bm{C}}}$ and $\hat{\bar{\bm{Q}}} = \bm{Q}_{\hat{\bar{\bm{C}}}}$.\par
\indent We first assume that $k = K$. Recall step \ref{ICaSE:Subspace} in Algorithm \ref{algorithm:ICaSE} is to estimate $\bar{\bm{C}}$, and step \ref{ICaSE:Variance} computes $\hat{\bm{V}}^{(f)}$ as
\begin{align}
\label{equation:LikeVfinal}
\hat{\bm{V}}^{(f)} = \argmax_{\substack{\bm{V} = \bm{V}\left( \bm{\tau} \right) \\ \bm{\tau} \in \Theta_*}}\left[ -\left( n-K\right)^{-1}\log \abs{\hat{\bm{Q}}^T \bm{V} \hat{\bm{Q}}} - \left( n-K\right)^{-1}\Tr\left\lbrace \hat{\bm{Q}}^T p^{-1}\bm{Y}^T \bm{Y}\hat{\bm{Q}}\left( \hat{\bm{Q}}^T \bm{V} \hat{\bm{Q}} \right)^{-1} \right\rbrace\right] .
\end{align}
where $\Theta_*$ is defined in Assumption \ref{assumption:ICaSE}. We therefore need to understand how $\hat{\bm{Q}}^T p^{-1}\bm{Y}^T \bm{Y}\hat{\bm{Q}}$ behaves. First, note that we can express $\hat{\bar{\bm{Q}}}$ in terms of $\hat{\bm{Q}}$:
\begin{align*}
\hat{\bar{\bm{Q}}} = \hat{\bm{V}}^{1/2} \hat{\bm{Q}}\left( \hat{\bm{Q}}^T \hat{\bm{V}} \hat{\bm{Q}}\right)^{-1/2}.
\end{align*}
Using the results of Corollary \ref{corollary:ConfSub:S}, we get that
\begin{align}
\hat{\bm{Q}}^T p^{-1}\bm{Y}^T \bm{Y}\hat{\bm{Q}} &= \left( \hat{\bm{Q}}^T \hat{\bm{V}} \hat{\bm{Q}} \right)^{1/2}\hat{\bar{\bm{Q}}}^T \hat{\bm{V}}^{-1/2} p^{-1}\bm{Y}^T \bm{Y} \hat{\bm{V}}^{-1/2}\hat{\bar{\bm{Q}}}\left( \hat{\bm{Q}}^T \hat{\bm{V}} \hat{\bm{Q}} \right)^{1/2}\nonumber\\
=& \left( \hat{\bm{Q}}^T \hat{\bm{V}} \hat{\bm{Q}} \right)^{1/2} \hat{\bar{\bm{Q}}}^T\begin{pmatrix}
\bar{\bm{C}} & \bar{\bm{Q}}
\end{pmatrix}\left( \gamma_K \bar{\bm{S}} \right) \begin{pmatrix}
\bar{\bm{C}}^T\\
\bar{\bm{Q}}^T
\end{pmatrix}\hat{\bar{\bm{Q}}}\left( \hat{\bm{Q}}^T \hat{\bm{V}} \hat{\bm{Q}} \right)^{1/2}\nonumber\\
=& \left( \hat{\bm{Q}}^T \hat{\bm{V}} \hat{\bm{Q}} \right)^{1/2} \hat{\bar{\bm{Q}}}^T\begin{pmatrix}
\bar{\bm{C}} & \bar{\bm{Q}}
\end{pmatrix} \left\lbrace \gamma_K \bar{\bm{S}} - \sum\limits_{j=1}^K \gamma_K\hat{\mu}_j \begin{pmatrix}
\hat{\bm{v}}_j\\
\hat{\bm{z}}_j
\end{pmatrix}\begin{pmatrix}
\hat{\bm{v}}_j^T & \hat{\bm{z}}_j^T
\end{pmatrix} \right\rbrace\begin{pmatrix}
\bar{\bm{C}}^T\\
\bar{\bm{Q}}^T
\end{pmatrix} \hat{\bar{\bm{Q}}} \left( \hat{\bm{Q}}^T \hat{\bm{V}} \hat{\bm{Q}} \right)^{1/2}\nonumber\\
=& \left( \hat{\bm{Q}}^T \hat{\bm{V}} \hat{\bm{Q}} \right)^{1/2} \hat{\bar{\bm{Q}}}^T \bar{\bm{Q}} \left( \bar{\bm{Q}}^T \hat{\bm{V}}^{-1/2}p^{-1}\bm{E}^T \bm{E} \hat{\bm{V}}^{-1/2} \bar{\bm{Q}} \right) \bar{\bm{Q}}^T \hat{\bar{\bm{Q}}}\left( \hat{\bm{Q}}^T \hat{\bm{V}} \hat{\bm{Q}} \right)^{1/2}\nonumber\\
\label{equation:QhatTrace}
& + \underbrace{o_P\left(1\right)}_{\text{rank $K$ matrix with this  2-norm}}.
\end{align}
Since the likelihood function in \eqref{equation:LikeVfinal} is depends on $\hat{\bm{Q}}^T p^{-1}\bm{Y}^T \bm{Y}\hat{\bm{Q}}$ through $1/(n-K)\Tr\left( \cdot \right)$, the rank $K$ matrix with Frobenius norm $o_P\left( 1 \right)$ will contribute $o_P\left( 1/n \right)$ to the likelihood, score function and Hessian, and can therefore be ignored. Let $\sigma_i$ be the $i^{\text{th}}$ singular value of $\bar{\bm{Q}}^T \hat{\bar{\bm{Q}}}$. Since $\bar{\bm{Q}}$ and $\hat{\bar{\bm{Q}}}$ have orthonormal columns, $\sigma_i \leq 1$. Further, the proof of Theorem \ref{theorem:CorrPCA} shows that
\begin{align*}
0 \leq & \sum\limits_{i=1}^{n-K} \left( 1-\sigma_i\right) \leq \sum\limits_{i=1}^{n-K}\left( 1-\sigma_i^2\right) = \Tr\left( I_{n-K}\right) - \Tr\left( \bar{\bm{Q}}^T \hat{\bar{\bm{Q}}} \hat{\bar{\bm{Q}}}^T\bar{\bm{Q}}\right) = 2^{-1}\norm{P_{\bar{\bm{Q}}} - P_{\hat{\bar{\bm{Q}}}}}_F^2\\
=& 2^{-1}\norm{P_{\bar{\bm{C}}} - P_{\hat{\bar{\bm{C}}}}}_F^2 = o_P(1).
\end{align*}
Since $\bar{\bm{Q}}$, $\hat{\bar{\bm{Q}}}$ are unique up to $(n-K) \times (n-K)$ rotation matrices, it suffices to assume $\bar{\bm{Q}}^T \hat{\bar{\bm{Q}}}$ is diagonal. And since the singular values of
\begin{align*}
\left( \bar{\bm{Q}}^T \hat{\bm{V}}^{-1/2}p^{-1}\bm{E}^T \bm{E} \hat{\bm{V}}^{-1/2} \bar{\bm{Q}} \right) \bar{\bm{Q}}^T \hat{\bar{\bm{Q}}} \left( \hat{\bm{Q}}^T \hat{\bm{V}} \hat{\bm{Q}} \right)^{1/2} \left\lbrace \hat{\bm{Q}}^T \bm{V}\left( \bm{\tau} \right) \hat{\bm{Q}} \right\rbrace^{-1} \left( \hat{\bm{Q}}^T \hat{\bm{V}} \hat{\bm{Q}} \right)^{1/2}
\end{align*}
are uniformly bounded from above in probability for all $\bm{\tau} \in \bm{\Theta}_*$,
\begin{align*}
&\left(n-K\right)^{-1}\left[ \left( \hat{\bm{Q}}^T \hat{\bm{V}} \hat{\bm{Q}} \right)^{1/2} \hat{\bar{\bm{Q}}}^T \bar{\bm{Q}} \left( \bar{\bm{Q}}^T \hat{\bm{V}}^{-1/2}p^{-1}\bm{E}^T \bm{E} \hat{\bm{V}}^{-1/2} \bar{\bm{Q}} \right) \bar{\bm{Q}}^T \hat{\bar{\bm{Q}}}\left( \hat{\bm{Q}}^T \hat{\bm{V}} \hat{\bm{Q}} \right)^{1/2}\left\lbrace \hat{\bm{Q}}^T \bm{V}\left( \bm{\tau} \right) \hat{\bm{Q}} \right\rbrace^{-1}  \right]\\
&= \left(n-K\right)^{-1}\Tr\left[ \left( \hat{\bm{Q}}^T \hat{\bm{V}} \hat{\bm{Q}} \right)^{1/2}\left( \bar{\bm{Q}}^T \hat{\bm{V}}^{-1/2}p^{-1}\bm{E}^T \bm{E} \hat{\bm{V}}^{-1/2} \bar{\bm{Q}} \right)\left( \hat{\bm{Q}}^T \hat{\bm{V}} \hat{\bm{Q}} \right)^{1/2}\left\lbrace \hat{\bm{Q}}^T \bm{V}\left( \bm{\tau} \right) \hat{\bm{Q}} \right\rbrace^{-1}  \right] + o_P\left( n^{-1} \right).
\end{align*}
Next, we can re-write $\bar{\bm{Q}}$ as
\begin{align*}
\bar{\bm{Q}} = \left( \hat{\bar{\bm{Q}}}\hat{\bar{\bm{Q}}}^T + \hat{\bar{\bm{C}}}\hat{\bar{\bm{C}}}^T \right)\bar{\bm{Q}} = \hat{\bar{\bm{Q}}} \left( \hat{\bar{\bm{Q}}}^T \bar{\bm{Q}} \right) + \underbrace{o_P\left( 1 \right)}_{\text{rank $K$ matrix with this 2-norm}}.
\end{align*}
Therefore, 
\begin{align*}
&\left(n-K\right)^{-1}\Tr\left[ \left( \hat{\bm{Q}}^T \hat{\bm{V}} \hat{\bm{Q}} \right)^{1/2}\left( \bar{\bm{Q}}^T \hat{\bm{V}}^{-1/2}p^{-1}\bm{E}^T \bm{E} \hat{\bm{V}}^{-1/2} \bar{\bm{Q}} \right)\left( \hat{\bm{Q}}^T \hat{\bm{V}} \hat{\bm{Q}} \right)^{1/2}\left\lbrace \hat{\bm{Q}}^T \bm{V}\left( \bm{\tau} \right) \hat{\bm{Q}} \right\rbrace^{-1}  \right]\\
&= \left(n-K\right)^{-1}\Tr\left[ \left( \hat{\bm{Q}}^T \hat{\bm{V}} \hat{\bm{Q}} \right)^{1/2}\hat{\bar{\bm{Q}}}^T \hat{\bm{V}}^{-1/2}\left( p^{-1}\bm{E}^T \bm{E} \right)\hat{\bm{V}}^{-1/2}\hat{\bar{\bm{Q}}}\left( \hat{\bm{Q}}^T \hat{\bm{V}} \hat{\bm{Q}} \right)^{1/2}\left\lbrace \hat{\bm{Q}}^T \bm{V}\left( \bm{\tau} \right) \hat{\bm{Q}} \right\rbrace^{-1}  \right] + o_P\left( n^{-1} \right)\\
&= \left(n-K\right)^{-1} \Tr\left[ \hat{\bm{Q}}^T\left( p^{-1}\bm{E}^T \bm{E} \right)\hat{\bm{Q}}\left\lbrace \hat{\bm{Q}}^T \bm{V}\left( \bm{\tau} \right) \hat{\bm{Q}} \right\rbrace^{-1}  \right] + o_P\left( n^{-1} \right).
\end{align*}
where the last equality follows from the fact that $\hat{\bar{\bm{Q}}}$ and $\hat{\bm{Q}}$ satisfy $\hat{\bar{\bm{Q}}} = \hat{\bm{V}}^{1/2}\hat{\bm{Q}}\left( \hat{\bm{Q}}^T \hat{\bm{V}}\hat{\bm{Q}}\right)^{-1/2}$. The likelihood function in \eqref{equation:LikeVfinal} (which is $O_P(1)$) can then be re-written up to factors that are $o_P(1/n)$ as
\begin{align*}
\hat{l}\left( \bm{\tau} \right) = -\left(n-K\right)^{-1}\log\abs{\hat{\bm{Q}}^T \bm{V}\left( \bm{\tau} \right) \hat{\bm{Q}}} - \left(n-K\right)^{-1}\Tr\left[ \hat{\bm{Q}}^T\left( p^{-1}\bm{E}^T \bm{E} \right)\hat{\bm{Q}}\left\lbrace \hat{\bm{Q}}^T \bm{V}\left( \bm{\tau} \right) \hat{\bm{Q}} \right\rbrace^{-1}\right].
\end{align*}
We will now show that $\hat{l}\left( \bm{\tau} \right) - l\left( \bm{\tau} \right) = O_P(1/n)$, where
\begin{align}
\label{equation:LikeVfinalFull}
l\left( \bm{\tau}\right) = -n^{-1}\log\abs{\bm{V}\left( \bm{\tau} \right)} - n^{-1}\Tr\left\lbrace \left( p^{-1}\bm{E}^T \bm{E} \right) \bm{V}\left( \bm{\tau} \right)^{-1}\right\rbrace.
\end{align}
First,
\begin{align*}
\log\abs{\bm{V}} = \log{\abs{ \hat{\bm{Q}}^T \bm{V} \hat{\bm{Q}} }} + \log\abs{ \hat{\bm{C}}^T \bm{V} \hat{\bm{C}} - \hat{\bm{C}}^T \bm{V} \hat{\bm{Q}} \left( \hat{\bm{Q}}^T \bm{V} \hat{\bm{Q}} \right)^{-1}\hat{\bm{Q}}^T \bm{V} \hat{\bm{C}} } = \log{\abs{ \hat{\bm{Q}}^T \bm{V} \hat{\bm{Q}} }} + O(1).
\end{align*}
where we have abused notation here and defined $\hat{\bm{C}}$ to be a $n \times K$ matrix with orthonormal columns that is orthogonal to $\hat{\bm{Q}}$. Let $\bm{S} = p^{-1}\bm{E}^T \bm{E}$. Then
\begin{align*}
\bm{S} = \hat{\bm{Q}}\hat{\bm{Q}}^T \bm{S} \hat{\bm{Q}}\hat{\bm{Q}}^T + \text{rank $K$ matrix eigenvalues that are $O_P(1)$}.
\end{align*}
Therefore,
\begin{align*}
l\left( \bm{\tau} \right) = -\left( n-K\right)^{-1}\log{\abs{ \hat{\bm{Q}}^T \bm{V} \hat{\bm{Q}} }} - \left( n-K\right)^{-1}\Tr\left[ \hat{\bm{Q}}^T \bm{S} \hat{\bm{Q}} \hat{\bm{Q}}^T \bm{V}^{-1}\hat{\bm{Q}} \right] + O_P\left( n^{-1} \right).
\end{align*}
Further, we can re-write $\bm{V}^{-1}$ as
\begin{align}
\bm{V}^{-1} &= \begin{bmatrix}
\hat{\bm{Q}} & \hat{\bm{C}}
\end{bmatrix}\begin{bmatrix}
\hat{\bm{Q}}^T \bm{V} \hat{\bm{Q}} & \hat{\bm{Q}}^T \bm{V} \hat{\bm{C}}\\
\hat{\bm{C}}^T \bm{V}\hat{\bm{Q}} & \hat{\bm{C}}^T \bm{V}\hat{\bm{C}}
\end{bmatrix}^{-1} \begin{bmatrix}
\hat{\bm{Q}}^T\\
\hat{\bm{C}}^T
\end{bmatrix}\nonumber\\
\label{equation:Vhat:Vinv}
&= \hat{\bm{Q}} \left( \hat{\bm{Q}}^T \bm{V} \hat{\bm{Q}} \right)^{-1}\hat{\bm{Q}}^T + \text{Matrix of rank $K$ with bounded singular values}.
\end{align}
We then get that
\begin{align*}
l\left( \bm{\tau} \right) = -\left(n-K\right)^{-1}\log{\abs{ \hat{\bm{Q}}^T \bm{V} \hat{\bm{Q}} }} - \left(n-K\right)^{-1}\Tr\left[ \hat{\bm{Q}}^T \bm{S} \hat{\bm{Q}} \left(\hat{\bm{Q}}^T \bm{V}\hat{\bm{Q}}\right)^{-1} \right] + O_P\left( n^{-1} \right),
\end{align*}
as desired.\par
\indent Since the eigenvalues of $\bm{V}(\bm{\tau})$ are uniformly bounded from above and below on $\Theta_*$, we have that $\mathop{\sup}\limits_{\bm{x}\in\bm{\Theta}_*}\abs{\hat{l}^{(n)}(\bm{x}) - l^{(n)}(\bm{x})} = O_P\left(n^{-1}\right)$ (the superscript $(n)$ is to indicate this is with sample size $n$). And since $\mathop{\sup}\limits_{\bm{x} \in \Theta_*} \abs{h_n\left( \bm{x}\right) - l^{(n)}\left( \bm{x}\right)} = O_P\left( n^{1/2}p^{-1/2}\right) = o_P(1)$, $\mathop{\sup}\limits_{\bm{x} \in \Theta_*} \abs{h\left( \bm{x}\right) - \hat{l}^{(n)}(\bm{x})} = o_P(1)$. This proves that $\norm{\hat{\bm{\tau}}^f - \bm{\tau}_*}_2 = o_P(1)$, since $h$ achieves a global maximum at $\bm{\tau}_*$ by Assumption \ref{assumption:ICaSE}.\par
\indent To find the rate, we simply use an identical analysis to write the score and Hessian of the objective function in \eqref{equation:LikeVfinal} up to terms that are $O_P(1/n)$, which gives us
\begin{align*}
\bm{s}^{(n)}\left( \bm{\tau}\right)_j &= \frac{d}{d\bm{\tau}_j}l^{(n)}\left( \bm{\tau}\right) = -n^{-1}\Tr\left\lbrace \bm{V}\left( \bm{\tau}\right)^{-1} \bm{B}_j\right\rbrace + n^{-1}\Tr\left\lbrace \bm{V}\left( \bm{\tau}\right)^{-1} \bm{B}_j\bm{V}\left( \bm{\tau}\right)^{-1} \bm{S} \right\rbrace \,\,\, (j=1,\ldots,b)\\
\hat{\bm{s}}^{(n)}\left( \bm{\tau}\right)_j &= \frac{d}{d\bm{\tau}_j}\hat{l}^{(n)}\left( \bm{\tau}\right) = -n^{-1}\Tr\left\lbrace \hat{\bm{A}}(\bm{\tau})^{-1} \hat{\bm{B}}_j\right\rbrace + n^{-1}\Tr\left\lbrace \hat{\bm{A}}(\bm{\tau})^{-1}\hat{\bm{B}}_j \hat{\bm{A}}(\bm{\tau})^{-1} \hat{\bm{Q}}^T \bm{S}\hat{\bm{Q}} \right\rbrace\\
&  + O_P\left( n^{-1}\right)  \,\,\, (j=1,\ldots,b)\\
H^{(n)}\left( \bm{\tau}\right)_{ij} &= \frac{d^2}{d\bm{\tau}_j d\bm{\tau}_i}l^{(n)}\left( \bm{\tau}\right) = n^{-1}\Tr\left\lbrace \bm{V}\left( \bm{\tau}\right)^{-1} \bm{B}_i \bm{V}\left( \bm{\tau}\right)^{-1} \bm{B}_j \right\rbrace\\
&- 2n^{-1}\Tr\left\lbrace \bm{V}\left( \bm{\tau}\right)^{-1} \bm{B}_i \bm{V}\left( \bm{\tau}\right)^{-1} \bm{B}_j \bm{V}\left( \bm{\tau}\right)^{-1} \bm{S} \right\rbrace \,\,\, (i=1,\ldots,b; j=1,\ldots,b)\\
\hat{H}^{(n)}\left( \bm{\tau}\right)_{ij} &= \frac{d^2}{d\bm{\tau}_j d\bm{\tau}_i}\hat{l}^{(n)}\left( \bm{\tau}\right) = n^{-1}\Tr\left\lbrace \hat{\bm{A}}\left( \bm{\tau}\right)^{-1} \hat{\bm{B}}_i \hat{\bm{A}}\left( \bm{\tau}\right)^{-1} \hat{\bm{B}}_j \right\rbrace\\
& - 2n^{-1}\Tr\left\lbrace \hat{\bm{A}}\left( \bm{\tau}\right)^{-1} \hat{\bm{B}}_i \hat{\bm{A}}\left( \bm{\tau}\right)^{-1} \hat{\bm{B}}_j \hat{\bm{A}}\left( \bm{\tau}\right)^{-1} \hat{\bm{Q}}^T\bm{S}\hat{\bm{Q}} \right\rbrace + O_P\left( n^{-1}\right) \,\,\, (i=1,\ldots,b; j=1,\ldots,b)
\end{align*}
where $\hat{\bm{B}}_j = \hat{\bm{Q}}^T\bm{B}_j \hat{\bm{Q}}$ and $\hat{\bm{A}} = \hat{\bm{Q}}^T\bm{V} \hat{\bm{Q}}$. An identical analysis as the one used to show the likelihoods differ by $O_P\left( n^{-1}\right)$ shows that $\norm{\bm{s}^{(n)} - \hat{\bm{s}}^{(n)}}_2,  \norm{\bm{H}^{(n)} - \hat{\bm{H}}^{(n)}}_2= O_P\left( n^{-1}\right)$. Define the matrix $\hat{\bm{\Gamma}} \in \mathbb{R}^{b \times r} = \bm{Q}_{\left[ A_{\mathcal{E}}\, A_{\mathcal{I},\hat{s}}\right]}$, $r \leq b$, where $\bm{A}_{\mathcal{I},\hat{s}}$ contains the rows of $\bm{A}_{\mathcal{I}}$ that satisfy the equality constraints at the point $\hat{\bm{\tau}}^f$. The column space of $\hat{\bm{\Gamma}}$ will contain the column space $\bm{\Gamma}_*$ with probability tending to 1 because $\hat{\bm{\tau}}^f$ converges in probability to $\bm{\tau}_*$. Since $\norm{\bm{s}^{(n)}\left( \bm{\tau}_*\right)}_2 = O_P\left\lbrace \left(np\right)^{-1/2}\right\rbrace$ and $\bm{H}^{(n)}\left( \bm{\tau}_*\right)$ is positive definite with eigenvalues uniformly bounded away from 0, the KKT conditions at the optimum imply
\begin{align*}
\bm{0} &= \hat{\bm{\Gamma}}^T \hat{\bm{s}}^{(n)}\left( \hat{\bm{\tau}}^f\right) = \hat{\bm{\Gamma}}^T \hat{\bm{s}}^{(n)}\left( \bm{\tau}_*\right) + \hat{\bm{\Gamma}}^T\hat{\bm{H}}^{(n)}\left( \bm{\tau}_*\right)\left( \hat{\bm{\tau}}^f - \bm{\tau}_* \right) + o_P\left( \norm{\hat{\bm{\tau}}^f - \bm{\tau}_*}_2\right)\\
& = \hat{\bm{\Gamma}}^T \bm{s}^{(n)}\left( \bm{\tau}_*\right) + \hat{\bm{\Gamma}}^T \bm{H}^{(n)}\left( \bm{\tau}_*\right) \hat{\bm{\Gamma}} \left( \hat{\bm{\theta}}^f - \bm{\theta}_*\right) + O_P\left( n^{-1}\right) + o_P\left( \norm{\hat{\bm{\tau}}^f - \bm{\tau}_*}_2\right)\\
&= O_P\left\lbrace \left(np\right)^{-1/2} + n^{-1} \right\rbrace +  \hat{\bm{\Gamma}}^T\bm{H}^{(n)}\left( \bm{\tau}_*\right) \hat{\bm{\Gamma}} \left( \hat{\bm{\theta}}^f - \bm{\theta}_*\right) + o_P\left( \norm{\hat{\bm{\tau}}^f - \bm{\tau}_*}_2\right)
\end{align*}
where $\hat{\bm{\Gamma}} \left( \hat{\bm{\theta}}^f - \bm{\theta}_*\right) = \hat{\bm{\tau}}^f - \bm{\tau}_*$, since the column space of $\bm{\Gamma}_*$ is contained in the column space of $\hat{\bm{\Gamma}}$ with probability tending to 1. Therefore, $\norm{\hat{\bm{\tau}}^f - \bm{\tau}_*}_2 = O_P\left( n^{-1}\right)$.\par
\indent When $k > K$, we need only go back to \eqref{equation:DebashiRemoved}. First, the maximum of the remaining eigenvalues of \eqref{equation:DebashiRemoved} is $\delta_*^2 + o_P(1)$ by Lemma \ref{lemma:EigsEtE}. Second, our new $\hat{\bar{\bm{Q}}} \in \mathbb{R}^{n \times (n-k)}$ is simply our old $\hat{\bar{\bm{Q}}}$ in \eqref{equation:QhatTrace} with $k-K$ basis vectors removed. These two facts simply mean that we replace all of the $o_P(1)$'s and $o_P\left( 1/n\right)$'s after \eqref{equation:QhatTrace} with $O_P(1)$ and $O_P(1/n)$. Everything else about the proof remains the same.
\end{proof}

\begin{center}
\line(1,0){425}
\end{center}

We can now prove Theorem \ref{theorem:ICaSE}.

\begin{proof}[of Theorem \ref{theorem:ICaSE}]
Again, for values $\delta^2$ and $\bm{\tau}$ such that $\delta^2 \bm{\tau} \in \Theta_*$, we redefine $\bm{\tau} \leftarrow \delta^2 \bm{\tau}$ for notational convenience to prove $\norm{\bm{\tau}_* - \hat{\bm{\tau}}}_2 = O_P\left( 1/n\right)$, which will prove \eqref{equation:ConfSub:V} by the analysis in the first paragraph of the proof of Lemma \ref{lemma:Vhat}. Results \eqref{equation:ConfSub:C} and \eqref{equation:ConfSub:CKlarge} will follow by Corollary \ref{corollary:ConfSub:CtC} because the column space of $\bm{C}_{\perp}$ and $\hat{\bm{C}}_{\perp}$ is invariant to scalar multiplication.\par
\indent For any value of $\hat{\bm{V}}$, define $\epsilon = \norm{\hat{\bm{V}} - \bm{V}_*}_2$. We use a similar technique to the proof of Lemma \ref{lemma:Vhat} and work with the $\Tr\left( \cdot\right)$ component in the likelihood given by \eqref{equation:LikeVfinal}, where
\begin{align*}
p^{-1}\bm{Y}^T \bm{Y} = p^{-1}\bm{C}\bm{L}^T \bm{L}\bm{C} + p^{-1}\bm{C}\bm{L}^T \bm{E} + \left( p^{-1}\bm{C}\bm{L}^T \bm{E}\right)^T + p^{-1}\bm{E}^T \bm{E}.
\end{align*}
Suppose first that $\gamma_1 = o(n)$. Then for any matrix $\hat{\bm{Q}} \in \mathbb{R}^{n \times k}$ with orthonormal columns ($k$ bounded from above),
\begin{align*}
n^{-1}\Tr\left\lbrace \hat{\bm{Q}}^T \left( p^{-1}\bm{Y}^T \bm{Y} \right)\hat{\bm{Q}} \left( \hat{\bm{Q}}^T \bm{V} \hat{\bm{Q}}\right)^{-1} \right\rbrace = n^{-1}\Tr\left\lbrace \hat{\bm{Q}}^T \left( p^{-1}\bm{E}^T \bm{E} \right)\hat{\bm{Q}} \left( \hat{\bm{Q}}^T \bm{V} \hat{\bm{Q}}\right)^{-1}\right\rbrace + o_P(1)
\end{align*}
since the rank $k$ matrix $\bm{C}\bm{L}^T \bm{E}$ has maximum singular value $\norm{p^{-1}\bm{C}\bm{L}^T \bm{E}}_2 = O_P\left\lbrace \left( \gamma_1 n \right)^{1/2}p^{-1/2} \right\rbrace$ and $\norm{p^{-1}\bm{C}\bm{L}^T \bm{L}\bm{C}}_2 = O(\gamma_1)$. Therefore, $\epsilon = o_P(1)$ for all estimates of $\bm{V}_*$ when $\gamma_1 = o(n)$, meaning all conditions of Theorem \ref{theorem:CorrPCA} and Lemma \ref{lemma:Vhat} will be satisfied on the first iteration of Algorithm \ref{algorithm:ICaSE}, meaning $\epsilon = O_P\left( n^{-1}\right)$ in step \ref{ICaSE:Variance} of the first iteration when $k=K$. Therefore, the estimate for $\bm{C}$ in the second iteration will satisfy
\eqref{equation:corrPCA:subspace}, \eqref{equation:corrPCA:lambda} and \eqref{equation:corrPCA:CtChat.IK} from Theorem \ref{theorem:CorrPCA} and \eqref{equation:Cor:ChatMChat}, \eqref{equation:Cor:ChatMC} for Corollary \ref{corollary:ConfSub:CtC} with $\epsilon = O_P\left( n^{-1}\right)$ when $k=K$. And since we use $\hat{\bm{V}}$ as a starting point for $k+1$, $\epsilon = O_P\left( n^{-1}\right)$ for all subsequent $k's$ ($k \geq K$).\par
\indent Next, suppose $\gamma_K \asymp n$ and $k = K$. Then after step \ref{ICaSE:Subspace} in the first iteration, we have for $\bar{\bm{C}}$ defined in \eqref{equation:Cbar} and $\hat{\bm{v}} = \begin{bmatrix}
\hat{\bm{v}}_1 & \cdots & \hat{\bm{v}}_K
\end{bmatrix}, \hat{\bm{z}} = \begin{bmatrix}
\hat{\bm{z}}_1 & \cdots & \hat{\bm{z}}_K
\end{bmatrix}$ defined in \eqref{equation:corrPCA:EigsResults},
\begin{align*}
\norm{n^{-1/2}\hat{\bm{Q}}^T \bm{C}}_2 &= \norm{n^{-1/2}\hat{\bm{Q}}^T \hat{\bm{V}}^{1/2} \hat{\bm{V}}^{-1/2} \bm{C}}_2 = O\left(\norm{\hat{\bar{\bm{Q}}}^T \bar{\bm{C}}}_2 \right) = O\left( \norm{ \hat{\bar{\bm{Q}}}^T \left( \bar{\bm{C}} - \hat{\bar{\bm{C}}} \hat{\bm{v}}^T \right) }_2\right)\\
 &= O\left( \norm{ \hat{\bar{\bm{Q}}}^T \left\lbrace \bar{\bm{C}}\left(I_K - \hat{\bm{v}}\hat{\bm{v}}^T \right) + \bar{\bm{Q}}\hat{\bm{z}}\hat{\bm{v}}^T \right\rbrace }_2 \right) = O_P\left\lbrace n^{1/2}\left( \gamma_K p\right)^{-1/2} + \epsilon \gamma_K^{-1}\right\rbrace.
\end{align*}
When completing step \ref{ICaSE:Variance} of the first iteration to re-estimate $\bm{V}_*$, we have
\begin{align*}
n^{-1}\Tr\left\lbrace \hat{\bm{Q}}^T \left( p^{-1}\bm{Y}^T \bm{Y} \right)\hat{\bm{Q}} \left( \hat{\bm{Q}}^T \bm{V} \hat{\bm{Q}}\right)^{-1} \right\rbrace = n^{-1}\Tr\left\lbrace \hat{\bm{Q}}^T \left( p^{-1}\bm{E}^T \bm{E} \right)\hat{\bm{Q}} \left( \hat{\bm{Q}}^T \bm{V} \hat{\bm{Q}}\right)^{-1}\right\rbrace + O_P\left( n^{-1}\right),
\end{align*}
meaning $\epsilon = O_P\left( n^{-1}\right)$ after this step. Therefore, step \ref{ICaSE:Subspace} of the second iteration will give us a $\bm{C}$ that satisfies \eqref{equation:corrPCA:subspace}, \eqref{equation:corrPCA:lambda} and \eqref{equation:corrPCA:CtChat.IK} from Theorem \ref{theorem:CorrPCA} and \eqref{equation:Cor:ChatMChat}, \eqref{equation:Cor:ChatMC} for Corollary \ref{corollary:ConfSub:CtC} with $\epsilon = O_P\left( n^{-1}\right)$. This is also true when $k > K$ by Lemma \ref{lemma:Vhat}, since we use the previous estimate of $\bm{V}_*$ as a starting point for step \ref{ICaSE:Subspace}. Equation \eqref{equation:ConfSub:CKlarge} then follows because
\begin{align*}
\norm{P_{C_{\perp}} - P_{\hat{C}_{\perp,k}} P_{C_{\perp}}}_F^2 &= \Tr\left( P_{C_{\perp}}\right) + \Tr\left( P_{\hat{C}_{\perp,k}} P_{C_{\perp}}\right) - 2\Tr\left( P_{\hat{C}_{\perp,k}} P_{C_{\perp}}\right) = K - \Tr\left( P_{\hat{C}_{\perp,k}} P_{C_{\perp}}\right)\\
& \leq K - \Tr\left( P_{\hat{C}_{\perp,K}} P_{C_{\perp}}\right) = K - \Tr\left\lbrace \bm{C}_{\perp}^T \hat{\bm{C}}_{\perp,K}\left( \hat{\bm{C}}_{\perp,K}^T \hat{\bm{C}}_{\perp,K}\right)^{-1} \hat{\bm{C}}_{\perp,K}^T \bm{C}_{\perp} \left( \bm{C}_{\perp}^T \bm{C}_{\perp}\right)^{-1} \right\rbrace\\
&= O_P\left\lbrace n\left( \gamma_K p\right)^{-1} + \left( \gamma_K p\right)^{-1/2} + n^{-1} \gamma_K^{-1}\right\rbrace
\end{align*}
where the column space of $\hat{\bm{C}}_{\perp,K} \in \mathbb{R}^{n \times K}$ is a subspace of the column space of $\hat{\bm{C}}_{\perp,k} \in \mathbb{R}^{n \times k}$, since $k \geq K$. The final equality follows from Corollary \ref{corollary:ConfSub:CtC}.
\end{proof}

\subsubsection{Proof of Theorem \ref{theorem:CBCV}}
\label{subsubsection:CBCV}
In this section use Theorem \ref{theorem:ICaSE} to prove Theorem \ref{theorem:CBCV}. Just as we did in Section \ref{subsubsection:Proof:ICaSE}, we assume $\bm{Y}$ is generated according to model \eqref{equation:Model_yANDvar} with $d=0$ and will use the incorrect, but simpler model in \eqref{equation:WrongModel} with $d=0$.

\begin{proof}[of Theorem \ref{theorem:CBCV}]
First,
\begin{align*}
\E_{\bm{\pi}}\left( \frac{n}{\gamma_K p_f} \bm{L}_{f}^T \bm{L}_{f}\right) = \frac{n}{\gamma_K p}\bm{L}^T \bm{L}
\end{align*}
and for $r,s \in [K]$,
\begin{align*}
\V\left( \frac{n}{\gamma_K p_f} \sum\limits_{g=1}^{p_f} \bm{\ell}_{\bm{\pi}(g)}[r]\bm{\ell}_{\bm{\pi}(g)}[s]\right) &\leq \sum\limits_{g=1}^{p_f}  \V\left( \frac{n}{\gamma_K p_f}\bm{\ell}_{\bm{\pi}(g)}[r]\bm{\ell}_{\bm{\pi}(g)}[s] \right) \leq \left(\frac{n}{\gamma_K p_f}\right)^2 \sum\limits_{g=1}^{p_f} \E\left( \bm{\ell}_{\bm{\pi}(g)}[r]^2\bm{\ell}_{\bm{\pi}(g)}[s]^2\right)\\
& \leq \left(\frac{n}{\gamma_K p_f}\right)^2 c^2 \sum\limits_{g=1}^{p_f}\E\left(\bm{\ell}_{\bm{\pi}(g)}[r]^2\right) = \frac{n c^2}{\gamma_K p^2} \sum\limits_{g=1}^{p_f}\E\left( \frac{n}{\gamma_K}\bm{\ell}_{\bm{\pi}(g)}[r]^2\right) = \frac{n c^2\gamma_r p_f}{\gamma_K^2 p^2}\\
&  \to 0
\end{align*}
where the first inequality follows from the fact that $\bm{\ell}_g[r]\bm{\ell}_g[s]$ is being sampled without replacement from a finite population, meaning successive draws are negatively correlated, and the third inequality because the magnitude entries of $\bm{L}$ are uniformly bounded by a constant $c$. Therefore, the eigenvalues of $n/\left( \gamma_K p_f\right) \bm{L}_{f}^T \bm{L}_{f}$ are $\gamma_k \gamma_K^{-1} + o_P(1)$ as $n \to \infty$ ($k=1,\ldots,K$), meaning the results of Theorems \ref{theorem:CorrPCA} and \ref{theorem:ICaSE} and Lemma \ref{lemma:Vhat} hold for any subset of $p_f$ rows of $\bm{Y}$ that are chosen uniformly at random (where $p_f/p \asymp 1$).\par 
\indent We next note that by assumption and by the proof of Theorem \ref{theorem:CorrPCA} and Lemma \ref{lemma:Vhat}, the leverage scores $\hat{\bar{h}}_i$ are such that $\mathop{\max}\limits_{i \in [n]} \hat{\bar{h}}_i = o_P(1)$ when $k=K$. Therefore, to prove the theorem it suffices to assume that $\mathop{\max}\limits_{i \in [n]} \hat{\bar{h}}_i < 1-\eta/m_n \log\left( m_n\right)$ for all $k \leq K_{\max}$. Fix some fold $f$ and define $\delta_{f*}^2$ and $\bm{V}_{f*}$ to be the analogues of $\delta_*^2$ of $\bm{V}_*$ for fold $f$ and $g(k)$ to be
\begin{align*}
g(k) = \frac{1}{\delta_{f*}^2n p_f}\E\left\lbrace LOO_f(k) \mid \bm{Y}_{(-f)}, \bm{\pi} \right\rbrace = \frac{1}{\delta_{f*}^2n p_f}\E\left\lbrace \sum\limits_{i=1}^n \norm{\bar{\bm{y}}_{f,i} - \hat{\bm{L}}_{f,(-i)} \hat{\bar{\bm{c}}}_i}_2^2 \mid \bm{Y}_{(-f)}, \bm{\pi} \right\rbrace.
\end{align*}
We can re-write $g(k)$ as
\begin{align*}
g(k) =& \underbrace{\frac{1}{\delta_{f*}^2 n p_f} \E\left\lbrace \sum\limits_{i=1}^n\norm{\bm{L}_f \bar{\bm{c}}_i - \hat{\bm{L}}_{f,(-i)} \hat{\bar{\bm{c}}}_i}_2^2 \mid \bm{Y}_{(-f)}, \bm{\pi} \right\rbrace}_{(1)} + \underbrace{\frac{1}{\delta_{f*}^2 n} \E\left[\Tr\left\lbrace\hat{\bm{V}}_{(-f)}^{-1}\left( p_f^{-1}\bm{E}_f^T \bm{E}_f \right)\right\rbrace  \mid \bm{Y}_{(-f)}, \bm{\pi}\right]}_{(2)}\\
&- \underbrace{\frac{2}{\delta_{f*}^2 n}\sum\limits_{i=1}^n \E \left\lbrace p_f^{-1} \left( \bm{L}_f \bar{\bm{c}}_i - \hat{\bm{L}}_{f,(-i)} \hat{\bar{\bm{c}}}_i \right)^T \bm{E}_f \hat{\bm{V}}_{(-f)}^{-1/2} \bm{a}_i \mid \bm{Y}_{(-f)}, \bm{\pi} \right\rbrace}_{(3)}
\end{align*}
We will go through each one of these terms to evaluate $g(k)$.
\begin{enumerate}[label=(\arabic*)]
\item \begin{align*}
\frac{1}{\delta_{f*}^2 n p_f} \E\left\lbrace \sum\limits_{i=1}^n\norm{\bm{L}_f \bar{\bm{c}}_i - \hat{\bm{L}}_{f,(-i)} \hat{\bar{\bm{c}}}_i}_2^2 \mid \bm{Y}_{(-f)}, \bm{\pi} \right\rbrace =& \frac{1}{\delta_{f*}^2 n p_f}\sum\limits_{i=1}^n\norm{ \frac{\bm{L}_f \bar{\bm{c}}_i - \bm{L}_f \bar{\bm{C}}^T \hat{\bar{\bm{C}}}\left( \hat{\bar{\bm{C}}}^T\hat{\bar{\bm{C}}} \right)^{-1} \hat{\bar{\bm{c}}}_i}{\left(1-\hat{\bar{h}}_i\right)} }_2^2\\
&+ n^{-1}\sum\limits_{i=1}^n\left( \frac{1}{1-\hat{\bar{h}}_i} \right)^2 \hat{\bar{\bm{H}}}_i^T \bm{A}_{(-i)} \hat{\bm{V}}_{(-f)}^{-1/2}\bm{V}_{f*}\hat{\bm{V}}_{(-f)}^{-1/2}\bm{A}_{(-i)}\hat{\bar{\bm{H}}}_i
\end{align*}
When \underline{$k = K$}, we have
\begin{align*}
\frac{1}{\delta_{f*}^2 n p_f}\sum\limits_{i=1}^n\norm{ \frac{\bm{L}_f \bar{\bm{c}}_i - \bm{L}_f \bar{\bm{C}}^T \hat{\bar{\bm{C}}}\left( \hat{\bar{\bm{C}}}^T\hat{\bar{\bm{C}}} \right)^{-1} \hat{\bar{\bm{c}}}_i}{\left(1-\hat{\bar{h}}_i\right)} }_2^2 &= \frac{1}{\delta_{f*}^2} \Tr\left\lbrace p_f^{-1}\bm{L}_f^T\bm{L}_f \left( n^{-1}\bar{\bm{C}}^T P_{\hat{\bar{\bm{C}}}}^{\perp}\bar{\bm{C}}\right)\right\rbrace\left\lbrace 1+o_P(1) \right\rbrace\\
&\underbrace{=}_{\text{Theorem \ref{theorem:CorrPCA} and Lemma \ref{lemma:Vhat}}} O_P\left( n^{-1/2}p^{-1/2} + p^{-1} + n^{-2}\right)\\
& = o_P\left( n^{-1}\right)
\end{align*}
and
\begin{align*}
& n^{-1}\sum\limits_{i=1}^n\left( \frac{1}{1-\hat{\bar{h}}_i} \right)^2 \hat{\bar{\bm{H}}}_i^T \bm{A}_{(-i)} \hat{\bm{V}}_{(-f)}^{-1/2}\bm{V}_{f*}\hat{\bm{V}}_{(-f)}^{-1/2}\bm{A}_{(-i)}\hat{\bar{\bm{H}}}_i \\
=& \frac{1+o_P(1)}{n}\sum\limits_{i=1}^n\left( \frac{1}{1-\hat{\bar{h}}_i} \right) \hat{\bar{\bm{H}}}_i^T \bm{A}_{(-i)} \hat{\bm{V}}_{(-f)}^{-1/2}\bm{V}_{f*}\hat{\bm{V}}_{(-f)}^{-1/2}\bm{A}_{(-i)}\hat{\bar{\bm{H}}}_i\\
=& \frac{1+o_P(1)}{n}\Tr\left\lbrace\hat{\bm{V}}_{(-f)}^{-1/2}\bm{V}_{f*}\hat{\bm{V}}_{(-f)}^{-1/2} \sum\limits_{i=1}^n \frac{1}{1-\hat{\bar{h}}_i} \bm{A}_{(-i)}\hat{\bar{\bm{H}}}_i \hat{\bar{\bm{H}}}_i^T \bm{A}_{(-i)} \right\rbrace.
\end{align*}
Note that $\sum\limits_{i=1}^n 1/\left( 1-\hat{\bar{h}}_i\right) \bm{A}_{(-i)}\hat{\bar{\bm{H}}}_i \hat{\bar{\bm{H}}}_i^T \bm{A}_{(-i)}$ is positive semi-definite with
\begin{align*}
\Tr\left\lbrace \sum\limits_{i=1}^n \frac{1}{1-\hat{\bar{h}}_i} \bm{A}_{(-i)}\hat{\bar{\bm{H}}}_i \hat{\bar{\bm{H}}}_i^T \bm{A}_{(-i)}\right\rbrace = \Tr\left\lbrace \sum\limits_{i=1}^n \frac{1}{1-\hat{\bar{h}}_i} \left(\hat{\bar{\bm{H}}}_i - \hat{\bar{h}}_i \bm{a}_i \right) \left(\hat{\bar{\bm{H}}}_i - \hat{\bar{h}}_i \bm{a}_i \right)^T\right\rbrace = \sum\limits_{i=1}^n \hat{\bar{h}}_i = K.
\end{align*}
Therefore,
\begin{align*}
n^{-1}\sum\limits_{i=1}^n\left( \frac{1}{1-\hat{\bar{h}}_i} \right)^2 \hat{\bar{\bm{H}}}_i^T \bm{A}_{(-i)} \hat{\bm{V}}_{(-f)}^{-1/2}\bm{V}_{f*}\hat{\bm{V}}_{(-f)}^{-1/2}\bm{A}_{(-i)}\hat{\bar{\bm{H}}}_i = K/n\left\lbrace 1+o_P(1)\right\rbrace.
\end{align*}
When \underline{$k < K$}, we have
\begin{align*}
\frac{1}{\delta_{f*}^2 n p_f}\sum\limits_{i=1}^n\norm{ \frac{\bm{L}_f \bar{\bm{c}}_i - \bm{L}_f \bar{\bm{C}}^T \hat{\bar{\bm{C}}}\left( \hat{\bar{\bm{C}}}^T\hat{\bar{\bm{C}}} \right)^{-1} \hat{\bar{\bm{c}}}_i}{\left(1-\hat{\bar{h}}_i\right)} }_2^2 &\geq \frac{1}{\delta_{f*}^2} \Tr\left\lbrace p_f^{-1}\bm{L}_f^T\bm{L}_f \left( n^{-1}\bar{\bm{C}}^T P_{\hat{\bar{\bm{C}}}}^{\perp}\bar{\bm{C}}\right)\right\rbrace\\
& \asymp \gamma_K/n
\end{align*}
and
\begin{align*}
n^{-1}\sum\limits_{i=1}^n\left( \frac{1}{1-\hat{\bar{h}}_i} \right)^2 \hat{\bar{\bm{H}}}_i^T \bm{A}_{(-i)} \hat{\bm{V}}_{(-f)}^{-1/2}\bm{V}_{f*}\hat{\bm{V}}_{(-f)}^{-1/2}\bm{A}_{(-i)}\hat{\bar{\bm{H}}}_i \geq 0.
\end{align*}
Lastly, when \underline{$k > K$},
\begin{align*}
\frac{1}{\delta_{f*}^2 n p_f}\sum\limits_{i=1}^n\norm{ \frac{\bm{L}_f \bar{\bm{c}}_i - \bm{L}_f \bar{\bm{C}}^T \hat{\bar{\bm{C}}}\left( \hat{\bar{\bm{C}}}^T\hat{\bar{\bm{C}}} \right)^{-1} \hat{\bar{\bm{c}}}_i}{\left(1-\hat{\bar{h}}_i\right)} }_2^2 \geq 0
\end{align*}
and
\begin{align*}
&n^{-1}\sum\limits_{i=1}^n\left( \frac{1}{1-\hat{\bar{h}}_i} \right)^2 \hat{\bar{\bm{H}}}_i^T \bm{A}_{(-i)} \hat{\bm{V}}_{(-f)}^{-1/2}\bm{V}_{f*}\hat{\bm{V}}_{(-f)}^{-1/2}\bm{A}_{(-i)}\hat{\bar{\bm{H}}}_i \\
\geq & n^{-1}\Tr\left\lbrace\hat{\bm{V}}_{(-f)}^{-1/2}\bm{V}_{f*}\hat{\bm{V}}_{(-f)}^{-1/2} \sum\limits_{i=1}^n \frac{1}{1-\hat{\bar{h}}_i} \bm{A}_{(-i)}\hat{\bar{\bm{H}}}_i \hat{\bar{\bm{H}}}_i^T \bm{A}_{(-i)} \right\rbrace = k/n\left\lbrace 1+o_P(1)\right\rbrace.
\end{align*}
Putting this all together, we get that
\begin{align*}
\frac{1}{\delta_{f*}^2 n p_f} \E\left\lbrace \sum\limits_{i=1}^n\norm{\bm{L}_f \bar{\bm{c}}_i - \hat{\bm{L}}_{f,(-i)} \hat{\bar{\bm{c}}}_i}_2^2 \mid \bm{Y}_{(-f)}, \bm{\pi} \right\rbrace\,\,\, \begin{cases}
\gtrsim \gamma_K/n & \text{if $k < K$}\\
= K/n\left\lbrace 1+o_P(1)\right\rbrace & \text{if $k = K$}\\
\geq k/n\left\lbrace 1+o_P(1)\right\rbrace & \text{if $k > K$}.
\end{cases}
\end{align*}

\item \begin{align*}
\frac{1}{\delta_{f*}^2 n} \E\left[\Tr\left\lbrace \hat{\bm{V}}_{(-f)}^{-1}\left( p_f^{-1}\bm{E}_f^T \bm{E}_f \right)\right\rbrace\mid \bm{Y}_{(-f)}, \bm{\pi}\right] = n^{-1}\Tr\left\lbrace \hat{\bm{V}}_{(-f)}^{-1} \bm{V}_{f*}\right\rbrace.
\end{align*}
Since $\log\abs{\hat{\bm{V}}_{(-f)}} = \log\abs{\bm{V}_{f*}} = 0$, $n^{-1}\Tr\left\lbrace \hat{\bm{V}}_{(-f)}^{-1} \bm{V}_{f*}\right\rbrace \geq 1$ with equality if and only if $\hat{\bm{V}}_{(-f)} = \bm{V}_{f*}$ by Jensen's inequality. Therefore, for \underline{$k < K$},
\begin{align*}
n^{-1}\Tr\left\lbrace \hat{\bm{V}}_{(-f)}^{-1} \bm{V}_{f*}\right\rbrace \geq 1.
\end{align*}
When \underline{$k \geq K$}, by Taylor's Theorem we have
\begin{align*}
n^{-1}\Tr\left\lbrace \hat{\bm{V}}_{(-f)}^{-1} \bm{V}_{f*}\right\rbrace &= 1- n^{-1}\Tr\left\lbrace \bm{V}_{f*}^{-1}\left( \hat{\bm{V}}_{(-f)} - \bm{V}_{f*}\right)\bm{V}_{f*}^{-1} \bm{V}_{f*}\right\rbrace + O_P\left\lbrace \left( n^{-1} + p^{-1/2}\right)^2\right\rbrace\\
&= 1- n^{-1}\Tr\left\lbrace \bm{V}_{f*}^{-1}\left( \hat{\bm{V}}_{(-f)} - \bm{V}_{f*}\right)\right\rbrace + o_P\left( n^{-1}\right).
\end{align*}
Therefore,
\begin{align*}
1 &\leq n^{-1}\Tr\left\lbrace \hat{\bm{V}}_{(-f)}^{-1}\bm{V}_{f*}\right\rbrace = 1 - n^{-1}\Tr\left\lbrace \bm{V}_{f*}^{-1} \left( \hat{\bm{V}}_{(-f)} - \bm{V}_{f*}\right)\right\rbrace + o_P\left( n^{-1}\right)\\
 &= 2 - n^{-1}\Tr\left\lbrace \bm{V}_{f*}^{-1}\hat{\bm{V}}_{(-f)}\right\rbrace + o_P\left( n^{-1}\right) \leq 1 + o_P\left( n^{-1}\right),
\end{align*}
which implies
\begin{align*}
n^{-1}\Tr\left\lbrace \hat{\bm{V}}_{(-f)}^{-1} \bm{V}_{f*}\right\rbrace = 1 + o_P\left( n^{-1} \right).
\end{align*}
Putting this all together, we get that
\begin{align*}
\frac{1}{\delta_{f*}^2 n} \E\left[ \Tr\left\lbrace \hat{\bm{V}}_{(-f)}^{-1}\left( p_f^{-1}\bm{E}_f^T \bm{E}_f \right) \right\rbrace \mid \bm{Y}_{(-f)}, \bm{\pi}\right]\,\,\, \begin{cases}
\geq 1 & \text{if $k < K$}\\
= 1 + o_P\left( n^{-1}\right) & \text{if $k \geq K$}.
\end{cases}
\end{align*}

\item \begin{align*}
&\frac{1}{\delta_{f*}^2 n}\sum\limits_{i=1}^n \E \left\lbrace p_f^{-1} \left( \bm{L}_f \bar{\bm{c}}_i - \hat{\bm{L}}_{f,(-i)} \hat{\bar{\bm{c}}}_i \right)^T \bm{E}_f \hat{\bm{V}}_{(-f)}^{-1/2} \bm{a}_i \mid \bm{Y}_{(-f)}, \bm{\pi} \right\rbrace \\
=& n^{-1}\sum\limits_{i=1}^n \frac{1}{1-\hat{\bar{h}}_i}\hat{\bar{\bm{H}}}_i^T \bm{A}_{(-i)}\hat{\bm{V}}_{(-f)}^{-1/2}\bm{V}_{f*} \hat{\bm{V}}_{(-f)}^{-1/2} \bm{a}_i = n^{-1}\Tr\left\lbrace \hat{\bm{V}}_{(-f)}^{-1/2}\bm{V}_{f*} \hat{\bm{V}}_{(-f)}^{-1/2} \sum\limits_{i=1}^n \frac{1}{1-\hat{\bar{h}}_i} \bm{A}_{(-i)}\hat{\bar{\bm{H}}}_i \bm{a}_i^T\right\rbrace
\end{align*}
where
\begin{align*}
\Tr\left\lbrace\sum\limits_{i=1}^n \frac{1}{1-\hat{\bar{h}}_i} \bm{A}_{(-i)}\hat{\bar{\bm{H}}}_i \bm{a}_i^T\right\rbrace &= 0\\
\norm{\sum\limits_{i=1}^n \frac{1}{1-\hat{\bar{h}}_i} \bm{A}_{(-i)}\hat{\bar{\bm{H}}}_i \bm{a}_i^T}_F^2 &= \sum\limits_{i=1}^n \frac{\hat{\bar{h}}_i}{1-\hat{\bar{h}}_i}.
\end{align*}
By Cauchy-Schwartz, we get that
\begin{align*}
&\abs{n^{-1}\Tr\left\lbrace \hat{\bm{V}}_{(-f)}^{-1/2}\bm{V}_{f*} \hat{\bm{V}}_{(-f)}^{-1/2} \sum\limits_{i=1}^n \frac{1}{1-\hat{\bar{h}}_i} \bm{A}_{(-i)}\hat{\bar{\bm{H}}}_i \bm{a}_i^T\right\rbrace}\\
&\leq \left( n^{-1}\Tr\left[ \left\lbrace I-\hat{\bm{V}}_{(-f)}^{-1/2}\bm{V}_{f*} \hat{\bm{V}}_{(-f)}^{-1/2}\right\rbrace^2\right]\right)^{1/2} \left( n^{-1}\norm{\sum\limits_{i=1}^n \frac{1}{1-\hat{\bar{h}}_i} \bm{A}_{(-i)}\hat{\bar{\bm{H}}}_i \bm{a}_i^T}_F^2\right)^{1/2}\\
&= O_P\left( n^{-3/2} + n^{-1/2}p^{-1/2}\right)M_k^{1/2}
\end{align*}
where $M_k = \mathop{\max}\limits_{i \in [n]}1/\left(1-\hat{\bar{h}}_i\right)$. When \underline{$k=K$}, $M_k = 1+o_P(1)$. Otherwise, we have $M_k = O\left[\frac{\min\left( n, p/n\right)}{\log\left\lbrace \min\left( n, p/n\right)\right\rbrace}\right]$ and
\begin{align*}
&M_k^{1/2}\left( n^{-3/2} + n^{-1/2}p^{-1/2}\right) = o_P\left( n^{-1}\right)\\
\Leftrightarrow & M_k^{1/2}\left( n^{-1/2} + n^{1/2}p^{-1/2}\right) = o_P\left( 1\right)\\
\Leftrightarrow & M_k\left( n^{-1} + np^{-1}\right) = o_P(1),
\end{align*}
where the last equality holds by our assumption on $M_k$. Putting this all together, we get that
\begin{align*}
\frac{1}{\delta_{f*}^2 n}\sum\limits_{i=1}^n \E \left\lbrace p_f^{-1} \left( \bm{L}_f \bar{\bm{c}}_i - \hat{\bm{L}}_{f,(-i)} \hat{\bar{\bm{c}}}_i \right)^T \bm{E}_f \hat{\bm{V}}_{(-f)}^{-1/2} \bm{a}_i \mid \bm{Y}_{(-f)},\bm{\pi} \right\rbrace = \begin{cases}
o_P\left( n^{-1}\right) & \text{if $k = K$}\\
o_P\left( n^{-1}\right) & \text{if $k \neq K$, $\mathop{\max}\limits_{i \in [n]} \hat{\bar{h}}_i < 1 - \eta \frac{\log m_n}{m_n}$}.
\end{cases}
\end{align*}
\end{enumerate}
Putting all three of these together, we finally get that
\begin{align*}
\frac{1}{\delta_{f*}^2n p_f}\E\left\lbrace LOO_f(k) \mid \bm{Y}_{(-f)},\bm{\pi} \right\rbrace \,\,\, \begin{cases}
\geq 1 + \Omega\left( \gamma_K/n\right) & \text{if $k < K$}\\
= 1 + K/n\left\lbrace 1+o_P(1)\right\rbrace & \text{if $k = K$}\\
\geq 1 + k/n\left\lbrace 1+o_P(1)\right\rbrace & \text{if $k > K$}
\end{cases}
\end{align*}
which completes the proof.
\end{proof}


\subsubsection{Lemma \ref{lemma:OmegaWLS} and Theorem \ref{theorem:CGLS}}
\label{subsubsection:Omega}
In this section we prove Lemma \ref{lemma:OmegaWLS} and Theorem \ref{theorem:CGLS} which justify performing inference on $\bm{\beta}$ using the estimated design matrix $\hat{\bm{C}}$. Recall from our discussion after the proof of Theorem \ref{theorem:CorrPCA} that once we have estimated $\delta_*^2$ and $\bm{V}_*$, it suffices to assume that $(n-d)^{-1}\bm{C}_{\perp}^T \hat{\bm{W}} \bm{C}_{\perp} = I_{K}$ and $(n-d)p^{-1}\bm{L}^T \bm{L}$ is diagonal with decreasing elements in the proofs Lemma \ref{lemma:OmegaWLS} and Theorem \ref{theorem:CGLS}.

\begin{proof}[of Lemma \ref{lemma:OmegaWLS}]
Let $m = n-d$. Once we have estimated $\bm{C}_{\perp}$, $\delta_*^2$ and $\bm{V}_*$ from Algorithm \ref{algorithm:ICaSE}, we define 
\begin{align*}
\tilde{\bm{X}} &= \hat{\bm{V}}^{-1/2}\bm{X} \in \mathbb{R}^{n \times d}\\
\tilde{\bm{C}} &= \hat{\bm{V}}^{-1/2}\bm{C} \in \mathbb{R}^{n \times K}\\
\tilde{\bm{C}}_{\perp} &= \bm{Q}_{\tilde{X}}^T \tilde{\bm{C}} \in \mathbb{R}^{m \times K}.
\end{align*}
From our discussion at the beginning of Section \ref{subsubsection:Omega}, it suffices to assume that $m^{-1}\tilde{\bm{C}}_{\perp}^T \tilde{\bm{C}}_{\perp} = I_K$ and $mp^{-1}\bm{L}^T \bm{L} = \diag\left( \lambda_1,\ldots,\lambda_K\right)$ where $\lambda_1,\lambda_K \asymp \gamma_K$ and $\left( \lambda_k - \lambda_{k+1}\right) / \lambda_{k+1} \geq c_1^{-1} + o_P(1)$ where $c_1$ is defined in Assumption \ref{assumption:Basics}. Therefore, all conditions of Theorem \ref{theorem:CorrPCA} and Lemma \ref{lemma:Vhat} are satisfied and \eqref{equation:corrPCA:EigsResults} holds with $\epsilon = 1/n$. This means that for $\hat{\bm{L}}$ defined in \eqref{equation:Lhat} and $\hat{\mu}_k$ defined the proof of Theorem \ref{theorem:CorrPCA} (see \eqref{equation:Eigsv} and \eqref{equation:Eigsz}),
\begin{align*}
m\left( \gamma_K p\right)^{-1}\hat{\bm{L}}^T \hat{\bm{L}} - m \gamma_K^{-1}\hat{\delta}^2 \left( \hat{\bm{C}}_{\perp}^T \hat{\bm{W}}^{-1} \hat{\bm{C}}_{\perp}\right)^{-1} &= \diag\left( \hat{\mu}_1,\ldots,\hat{\mu}_K \right) - \gamma_K^{-1}\hat{\delta}^2 I_K\\
& = m\left( \gamma_K p\right)^{-1}\bm{L}^T \bm{L} + O_P\left\lbrace n\left( p \gamma_K\right)^{-1} + \left( p \gamma_K\right)^{-1/2} + \left( n \gamma_K\right)^{-1} \right\rbrace\\
&= m\left( \gamma_K p\right)^{-1}\bm{L}^T \bm{L} + o_P\left( n^{-1/2}\right)
\end{align*}
Therefore, to prove Lemma \ref{lemma:OmegaWLS}, we only need to show that
\begin{align}
\label{equation:Lemma1ToShow}
\norm{m \left( \gamma_K p\right)^{-1}\hat{\bm{L}}^T\bm{Y}_1 - \left\lbrace m\left( \gamma_K p\right)^{-1}\bm{L}^T \bm{L}\bm{\Omega}^{T}\right\rbrace}_2 = o_P\left( n^{-1/2}\right).
\end{align}
\indent We first note that 
\begin{align*}
\tilde{\bm{C}}_{\perp} = \bm{Q}_{\tilde{X}}^T \tilde{\bm{C}} = \left( \bm{Q}_{X}^T \hat{\bm{V}} \bm{Q}_{X}\right)^{-1/2} \left( \hat{\bm{V}}^{1/2} \bm{Q}_{X}\right)^T \hat{\bm{V}}^{-1/2}\bm{C} = \hat{\bm{W}}^{-1/2} \bm{C}_{\perp}.
\end{align*}
Algorithm \ref{algorithm:ICaSE} will return
\begin{align*}
\hat{\tilde{\bm{C}}}_{\perp} = \tilde{\bm{C}}_{\perp}\hat{\bm{v}} + \sqrt{m}\bm{Q}_{\tilde{C}_{\perp}}\hat{\bm{z}},
\end{align*}
an estimate of $\tilde{\bm{C}}_{\perp}$, where $\hat{\bm{v}}$ and $\hat{\bm{z}}$ are defined in \eqref{equation:Eigsv} and \eqref{equation:Eigsz} with asymptotic properties given in \eqref{equation:corrPCA:EigsResults}. Next, since $\bm{Q}_{\tilde{X}} = \hat{\bm{V}}^{1/2}\bm{Q}_X \hat{\bm{W}}^{-1/2}$,
\begin{align}
\hat{\bm{L}} &= \bm{Y}\bm{Q}_{X}\hat{\bm{W}}^{-1}\hat{\bm{C}}_{\perp} \left( \hat{\bm{C}}_{\perp}^T \hat{\bm{W}}^{-1}\hat{\bm{C}}_{\perp} \right)^{-1} = \bm{Y}\bm{Q}_{X}\hat{\bm{W}}^{-1/2}\hat{\tilde{\bm{C}}}_{\perp}\left( \hat{\tilde{\bm{C}}}_{\perp}^T\hat{\tilde{\bm{C}}}_{\perp}\right)^{-1}\nonumber\\
\label{equation:OmegaWLS:L}
&= m^{-1}\bm{L}\tilde{\bm{C}}_{\perp}^T \hat{\tilde{\bm{C}}}_{\perp} + m^{-1}\bm{E}\hat{\bm{V}}^{-1/2}\bm{Q}_{\tilde{X}}\hat{\tilde{\bm{C}}}_{\perp} = \bm{L} \hat{\bm{v}} + m^{-1}\bm{E}\hat{\bm{V}}^{-1/2}\bm{Q}_{\tilde{X}}\hat{\tilde{\bm{C}}}_{\perp}
\end{align}
and
\begin{align*}
m \left( \gamma_K p\right)^{-1}\hat{\bm{L}}^T\bm{Y}_1 &= m \left( \gamma_K p\right)^{-1}\hat{\bm{L}}^T \bm{Y}\hat{\bm{V}}^{-1}\bm{X}\left( \tilde{\bm{X}}^T \tilde{\bm{X}}\right)^{-1}\\
 &= \underbrace{m \left( \gamma_K p\right)^{-1}\hat{\bm{L}}^T \bm{B}}_{(1)} + \underbrace{m \left( \gamma_K p\right)^{-1}\hat{\bm{L}}^T \bm{L}\bm{\Omega}^{T}}_{(2)} + \underbrace{m \left( \gamma_K p\right)^{-1}\hat{\bm{L}}^T \bm{E}\hat{\bm{V}}^{-1}\bm{X}\left( \tilde{\bm{X}}^T \tilde{\bm{X}}\right)^{-1}}_{(3)}.
\end{align*}
We will go through each of these terms to prove \eqref{equation:Lemma1ToShow}.

\begin{enumerate}[label=(\arabic*)]
\item \begin{align*}
m \left( \gamma_K p\right)^{-1}\hat{\bm{L}}^T \bm{B} = m \left( \gamma_K p\right)^{-1}\hat{\bm{v}}^T \bm{L}^T \bm{B} + \left( \gamma_K p\right)^{-1} \hat{\tilde{\bm{C}}}_{\perp}^T \bm{Q}_{\tilde{X}}^T \hat{\bm{V}}^{-1/2}\bm{E}^T \bm{B}.
\end{align*}
By the assumptions on $\bm{B}$ and $\bm{L}$, the first term is $o_P\left( n^{-1/2}\right)$. For the second term, it suffices to assume $d = 1$. Then
\begin{align*}
\text{Var}\left( p^{-1/2}\bm{E}^T \bm{B}\right) = p^{-1}\sum\limits_{g=1}^p \beta_g^2 \bm{V}_g.
\end{align*}
Therefore,
\begin{align*}
\E \norm{p^{-1/2}\bm{E}^T \bm{B}}_2^2 \asymp n s_1 
\end{align*}
where $s_1$ is defined in Assumption \ref{assumption:B}. By Assumption \ref{assumption:B},
\begin{align*}
n^{1/2}\left\lbrace\left( \gamma_K^{-1}n^{1/2}p^{-1/2}\right) \left( n s_1\right)^{1/2}\right\rbrace \lesssim n^{3/4}\left( \gamma_K p\right)^{-1/2} = \left\lbrace n^{3/2}\left( \gamma_K p\right)^{-1} \right\rbrace^{1/2} \to 0.
\end{align*}
Therefore, the second term is $o_P\left( n^{-1/2}\right)$.

\item \begin{align*}
m\left( \gamma_K p\right)^{-1} \hat{\bm{L}}^T \bm{L} = \hat{\bm{v}}^T m\left( \gamma_K p\right)^{-1}\bm{L}^T \bm{L} + \left( \gamma_K mp\right)^{-1/2}\hat{\tilde{\bm{C}}}_{\perp}^T \bm{Q}_{\tilde{X}}^T \hat{\bm{V}}^{-1/2}\bm{E}^T \bar{\bm{L}}.
\end{align*}
By \eqref{equation:EigsResults:vhata}, the first term is such that
\begin{align*}
\norm{\hat{\bm{v}}^T m\left( \gamma_K p\right)^{-1}\bm{L}^T \bm{L} - m\left( \gamma_K p\right)^{-1}\bm{L}^T \bm{L}}_2 = o_P\left( n^{-1/2}\right).
\end{align*}
For the second term,
\begin{align*}
\left( \gamma_K mp\right)^{-1/2} \hat{\tilde{\bm{C}}}_{\perp}^T \bm{Q}_{\tilde{X}}^T \hat{\bm{V}}^{-1/2}\bm{E}^T \bar{\bm{L}} =& \left( \gamma_K mp\right)^{-1/2} \hat{\bm{v}}^T \tilde{\bm{C}}_{\perp}^T \bm{Q}_{\tilde{X}}^T \hat{\bm{V}}^{-1/2}\bm{E}^T \bar{\bm{L}}\\
&+ \underbrace{\hat{\bm{z}}^T \bm{Q}_{\tilde{\bm{C}}}^T \bm{Q}_{\tilde{X}}^T \hat{\bm{V}}^{-1/2}}_{O_P\left\lbrace n^{1/2}\left(\gamma_K p\right)^{-1/2} + \left(n\gamma_K\right)^{-1} \right\rbrace}\underbrace{\left( \gamma_K p\right)^{-1/2}\bm{E}^T \bar{\bm{L}}}_{O_P\left\lbrace n^{1/2}\left(\gamma_K p\right)^{-1/2}\right\rbrace}\\
=& \left( \gamma_K mp\right)^{-1/2} \hat{\bm{v}}^T \bm{C}_{\perp}^T \left( \bm{Q}_{X}^T \hat{\bm{V}}\bm{Q}_{X}\right)^{-1}\bm{Q}_{X}^T \bm{E}^T\bar{\bm{L}} + o_P\left( n^{-1/2}\right) \\
=& \left( \gamma_K p\right)^{-1/2} \hat{\bm{v}}^T \left( m^{-1/2}\bm{C}_{\perp}\right)^T \left( \bm{Q}_{X}^T \bm{V}_*\bm{Q}_{X}\right)^{-1}\bm{Q}_{X}^T \bm{E}^T\bar{\bm{L}} + o_P\left( n^{-1/2}\right)\\
=& o_P\left( n^{-1/2}\right).
\end{align*}

\item Lastly,
\begin{align*}
m\left( \gamma_K p\right)^{-1}\hat{\bm{L}}^T \bm{E}\hat{\bm{V}}^{-1}\bm{X}\left( \tilde{\bm{X}}^T \tilde{\bm{X}}\right)^{-1} =& m\left( \gamma_K p\right)^{-1}\hat{\bm{v}}\bm{L}^T \bm{E}\hat{\bm{V}}^{-1}\bm{X}\left( \tilde{\bm{X}}^T \tilde{\bm{X}}\right)^{-1}\\
& + \left( \gamma_K p\right)^{-1}\hat{\tilde{\bm{C}}}_{\perp}^T \bm{Q}_{\tilde{X}}^T \hat{\bm{V}}^{-1/2}\bm{E}^T \bm{E}\hat{\bm{V}}^{-1}\bm{X}\left( \tilde{\bm{X}}^T\tilde{\bm{X}}\right)^{-1}\\
=& \left( \gamma_K p\right)^{-1/2} \bar{\bm{L}}^T \bm{E}\hat{\bm{V}}^{-1}\left\lbrace m^{1/2}\bm{X}\left( \tilde{\bm{X}}^T \tilde{\bm{X}}\right)^{-1} \right\rbrace\\
&+ \gamma_K^{-1}\hat{\bm{v}}^T \bm{C}_{\perp}^T \left( \bm{Q}_{X}^T \hat{\bm{V}}\bm{Q}_{X}\right)^{-1}\bm{Q}_{X}^T p^{-1}\bm{E}^T \bm{E}\hat{\bm{V}}^{-1}\bm{X}\left( \tilde{\bm{X}}^T\tilde{\bm{X}}\right)^{-1}\\
&+ \gamma_K^{-1}\hat{\bm{z}}^T \bm{Q}_{\tilde{C}_{\perp}}^T \bm{Q}_{\tilde{X}}^T \hat{\bm{V}}^{-1/2} p^{-1}\bm{E}^T \bm{E}\hat{\bm{V}}^{-1/2}\left\lbrace m^{1/2}\tilde{\bm{X}}\left( \tilde{\bm{X}}^T \tilde{\bm{X}}\right)^{-1}\right\rbrace.
\end{align*}
First,
\begin{align*}
\left( \gamma_K p\right)^{-1/2} \bar{\bm{L}}^T \bm{E}\hat{\bm{V}}^{-1}\left\lbrace m^{1/2}\bm{X}\left( \tilde{\bm{X}}^T \tilde{\bm{X}}\right)^{-1} \right\rbrace &= \left( \gamma_K p\right)^{-1/2} \bar{\bm{L}}^T \bm{E}\bm{V}_*^{-1}\left\lbrace m^{1/2}\bm{X}\left( \tilde{\bm{X}}^T \tilde{\bm{X}}\right)^{-1} \right\rbrace + o_P\left( n^{-1/2}\right)\\
&= o_P\left( n^{-1/2}\right).
\end{align*}
Next,
\begin{align*}
&\gamma_K^{-1}\hat{\bm{v}}^T \bm{C}_{\perp}^T \left( \bm{Q}_{X}^T \hat{\bm{V}}\bm{Q}_{X}\right)^{-1}\bm{Q}_{X}^T \frac{1}{p}\bm{E}^T \bm{E}\hat{\bm{V}}^{-1}\bm{X}\left( \tilde{\bm{X}}^T\tilde{\bm{X}}\right)^{-1} \\
=& \gamma_K^{-1}\hat{\bm{v}}^T \bm{C}_{\perp}^T \bm{W}_*^{-1}\bm{Q}_{X}^T p^{-1}\bm{E}^T \bm{E}\bm{V}_*^{-1}\bm{X}\left( \tilde{\bm{X}}^T\tilde{\bm{X}}\right)^{-1} + O_P\left\lbrace \left( \gamma_K n\right)^{-1}\right\rbrace.
\end{align*}
Let $\bm{A} \in \mathbb{R}^{m \times K}$ be a matrix whose columns form an orthonormal basis for the column space of $\bm{C}_{\perp}$, which is not dependent on $\hat{\bm{V}}$, and define the non-random matrices
\begin{align*}
\bm{f} &= \bm{Q}_{X}\bm{W}_*^{-1}\bm{A} \in \mathbb{R}^{n \times K}\\
\bm{u} &= \bm{X}\left( \bm{X}^T \bm{X}\right)^{-1/2} \in \mathbb{R}^{n \times d},
\end{align*}
both of which have bounded 2-norm and $\bm{f}^T \bm{u} = \bm{0}$. Since $K,d$ are fixed, it suffices to assume that $K=d=1$. We first see that
\begin{align*}
\E \left( \bm{f}^T p^{-1} \bm{E}^T \bm{E} \bm{V}_*^{-1}\bm{u}\right) = \bm{f}^T \left( p^{-1}\sum\limits_{g=1}^p \bm{V}_g\right)\bm{V}_*^{-1}\bm{u} = \bm{f}^T \bm{V}_{*} \bm{V}_*^{-1}\bm{u} = 0
\end{align*}
and
\begin{align*}
\V\left(  \bm{f}^T p^{-1} \bm{E}^T \bm{E} \bm{V}_*^{-1}\bm{u}\right) = p^{-2} \sum\limits_{g=1}^p \text{Var}\left( \tilde{e}_{g,1} \tilde{e}_{g,2}\right) \asymp p^{-1}
\end{align*}
where $\tilde{e}_{g,1} = \bm{f}^T\bm{e}_g \sim N\left( 0,\bm{f}^T \bm{V}_g\bm{f}\right)$ and $\tilde{e}_{g,2} = \bm{u}^T \bm{V}_*^{-1} \bm{e}_g \sim N\left( 0, \bm{u}^T \bm{V}_*^{-1}\bm{V}_g \bm{V}_*^{-1}\bm{u}\right)$.
Therefore,
\begin{align*}
\norm{\gamma_K^{-1}\hat{\bm{v}}^T \bm{C}_{\perp}^T \left( \bm{Q}_{X}^T \hat{\bm{V}}\bm{Q}_{X}\right)^{-1}\bm{Q}_{X}^T p^{-1}\bm{E}^T \bm{E}\hat{\bm{V}}^{-1}\bm{X}\left( \tilde{\bm{X}}^T\tilde{\bm{X}}\right)^{-1}}_2 = O_P\left\lbrace \gamma_K^{-1}p^{-1/2} + \left( \gamma_K n\right)^{-1}\right\rbrace = o_P\left( n^{-1/2}\right). 
\end{align*}
Finally,
\begin{align*}
&\gamma_K^{-1}\hat{\bm{z}}^T \bm{Q}_{\tilde{C}_{\perp}}^T \bm{Q}_{\tilde{X}}^T \hat{\bm{V}}^{-1/2} p^{-1}\bm{E}^T \bm{E}\hat{\bm{V}}^{-1/2}\left\lbrace m^{1/2}\tilde{\bm{X}}\left( \tilde{\bm{X}}^T \tilde{\bm{X}}\right)^{-1}\right\rbrace\\
=& \gamma_K^{-1}\hat{\bm{z}}^T \bm{Q}_{\tilde{C}_{\perp}}^T \bm{Q}_{\tilde{X}}^T \bm{V}_*^{-1/2} p^{-1}\bm{E}^T \bm{E}\bm{V}_*^{-1/2}\left\lbrace m^{1/2}\tilde{\bm{X}}\left( \tilde{\bm{X}}^T \tilde{\bm{X}}\right)^{-1}\right\rbrace + o_P\left\lbrace \left( n\gamma_K\right)^{-1}\right\rbrace\\
=& \gamma_K^{-1}\hat{\bm{z}}^T \bm{Q}_{\tilde{C}_{\perp}}^T \bm{Q}_{\tilde{X}}^T \bm{V}_*^{-1/2} \left(p^{-1}\bm{E}^T \bm{E} - \bm{V}_*\right)\bm{V}_*^{-1/2}\left\lbrace m^{1/2}\tilde{\bm{X}}\left( \tilde{\bm{X}}^T \tilde{\bm{X}}\right)^{-1}\right\rbrace + o_P\left\lbrace \left( n\gamma_K\right)^{-1}\right\rbrace
\end{align*}
where
\begin{align*}
\norm{\gamma_K^{-1}\hat{\bm{z}}^T \bm{Q}_{\tilde{C}_{\perp}}^T \bm{Q}_{\tilde{X}}^T \bm{V}_*^{-1/2} \left(p^{-1}\bm{E}^T \bm{E} - \bm{V}_*\right)\bm{V}_*^{-1/2}\left\lbrace m^{1/2}\tilde{\bm{X}}\left( \tilde{\bm{X}}^T \tilde{\bm{X}}\right)^{-1}\right\rbrace}_2 = O_P\left( n\gamma_K^{-3/2}p^{-1}\right) = o_P\left( n^{-1/2}\right).
\end{align*}
\end{enumerate}
This proves \eqref{equation:Lemma1ToShow} and completes the proof.

\end{proof}

\begin{center}
\line(1,0){425}
\end{center}

\begin{proof}[of Theorem \ref{theorem:CGLS}]
Fix a $g \in [p]$. We will first show \eqref{equation:VgAsym}. Item \ref{item:Vg:Technical} in Assumption \ref{assumption:ICaSE} ensures that the likelihood
\begin{align*}
l_{n,g}\left( \bm{\theta}\right) = n^{-1}\log\abs{\bm{V}\left( \bm{\theta}\right)} - n^{-1}\bm{e}_g^T \bm{V}\left( \bm{\theta}\right)^{-1}\bm{e}_g
\end{align*}
satisfies the weak uniform law of large numbers
\begin{align*}
\sup_{\bm{\theta} \in \Theta_*} \abs{l_{n,g}\left( \bm{\theta}\right) - f_{g}\left( \bm{\theta}\right)} = o_P(1)
\end{align*}
with $f_g$ being uniquely maximized at $\bm{\theta} = \bm{v}_g$. The restricted maximum likelihood problem we are interested in solving is
\begin{align}
\label{equation:Argmax.lnghat}
\argmax_{\bm{\theta} \in \Theta_*} \hat{l}_{n,g}\left( \bm{\theta}\right), \quad \hat{l}_{n,g}\left( \bm{\theta}\right) = -n^{-1}\log\abs{\bm{Q}_{\hat{C}_{\perp}}^T \bm{W}\left( \bm{\theta}\right)\bm{Q}_{\hat{C}_{\perp}}} - n^{-1}\bm{y}_{g_2}^T \bm{Q}_{\hat{C}_{\perp}} \left\lbrace\bm{Q}_{\hat{C}_{\perp}}^T \bm{W}\left( \bm{\theta}\right) \bm{Q}_{\hat{C}_{\perp}}\right\rbrace^{-1}\bm{Q}_{\hat{C}_{\perp}}^T \bm{y}_{g_2}.
\end{align}
We use the same technique used in the proof of Lemma \ref{lemma:Vhat} to show that
\begin{align}
\label{equation:lnglng.hat}
\sup_{\bm{\theta} \in \Theta_*}\abs{l_{n,g}\left( \bm{\theta}\right) - \hat{l}_{n,g}\left( \bm{\theta}\right)} = O_P\left\lbrace n^{1/2}\left( \gamma_K p\right)^{-1/2} + n^{-1/2}\right\rbrace,
\end{align}
which can also be used to show that $\abs{\nabla l_{n,g} - \nabla \hat{l}_{n,g}}, \abs{\nabla^2 l_{n,g} - \nabla^2 \hat{l}_{n,g}} = O_P\left\lbrace n^{1/2}\left( \gamma_K p\right)^{-1/2} + n^{-1/2}\right\rbrace$. First, from \eqref{equation:corrPCA:EigsResults} we get that
\begin{align*}
n^{-1/2}\norm{\bm{Q}_{\hat{C}_{\perp}}^T \bm{C}_{\perp}}_2 = O\left\lbrace \norm{\bm{Q}_{\hat{\bar{C}}_{\perp}}^T \left(\bar{\bm{C}}_{\perp} - \hat{\bar{\bm{C}}}_{\perp} \right)}_2\right\rbrace = O\left( \norm{\bar{\bm{C}}_{\perp} - \hat{\bar{\bm{C}}}_{\perp}}_2\right) = O_P\left\lbrace n^{1/2}\left( \gamma_K p\right)^{-1/2} + n^{-1} \right\rbrace
\end{align*}
where $\bar{\bm{C}}_{\perp} = n^{-1/2}\hat{\bm{W}}^{-1/2}\bm{C}_{\perp}$ and $\hat{\bar{\bm{C}}}_{\perp} = n^{-1/2}\hat{\bm{W}}^{-1/2}\hat{\bm{C}}_{\perp}$. Define $\eta = n^{1/2}\left( \gamma_K p\right)^{-1/2} + n^{-1}$. Then
\begin{align*}
n^{-1}\bm{y}_{g_2}^T \bm{Q}_{\hat{C}_{\perp}} \left\lbrace\bm{Q}_{\hat{C}_{\perp}}^T \bm{W}\left( \bm{\theta}\right) \bm{Q}_{\hat{C}_{\perp}}\right\rbrace^{-1}\bm{Q}_{\hat{C}_{\perp}}^T \bm{y}_{g_2} =& \underbrace{n^{-1}\bm{\ell}_g^T \bm{C}_{\perp}^T \bm{Q}_{\hat{C}_{\perp}} \left\lbrace\bm{Q}_{\hat{C}_{\perp}}^T \bm{W}\left( \bm{\theta}\right) \bm{Q}_{\hat{C}_{\perp}}\right\rbrace^{-1}\bm{Q}_{\hat{C}_{\perp}}^T \bm{C}_{\perp} \bm{\ell}_g}_{o_P\left( \eta\right)}\\
&+ \underbrace{n^{-1}\bm{\ell}_g^T \bm{C}_{\perp}^T \bm{Q}_{\hat{C}_{\perp}} \left\lbrace\bm{Q}_{\hat{C}_{\perp}}^T \bm{W}\left( \bm{\theta}\right) \bm{Q}_{\hat{C}_{\perp}}\right\rbrace^{-1}\bm{Q}_{\hat{C}_{\perp}}^T\bm{e}_{g_2}}_{O_P\left(\eta \right)}\\
& + \left[ n^{-1}\bm{\ell}_g^T \bm{C}_{\perp}^T \bm{Q}_{\hat{C}_{\perp}} \left\lbrace\bm{Q}_{\hat{C}_{\perp}}^T \bm{W}\left( \bm{\theta}\right) \bm{Q}_{\hat{C}_{\perp}}\right\rbrace^{-1}\bm{Q}_{\hat{C}_{\perp}}^T\bm{e}_{g_2} \right]^T\\
&+ \underbrace{n^{-1} \bm{e}_{g_2}^T \bm{Q}_{\hat{C}_{\perp}} \left\lbrace\bm{Q}_{\hat{C}_{\perp}}^T \bm{W}\left( \bm{\theta}\right) \bm{Q}_{\hat{C}_{\perp}}\right\rbrace^{-1}\bm{Q}_{\hat{C}_{\perp}}^T\bm{e}_{g_2}}_{n^{-1}e_g^T \bm{V}\left( \bm{\theta}\right)^{-1}e_g + O_P\left( n^{-1}\right)}\\
=& n^{-1}e_g^T \bm{V}\left( \bm{\theta}\right)^{-1}e_g + O_P\left( \eta\right).
\end{align*}
And since
\begin{align*}
\sup_{\bm{\theta} \in \Theta_*} \abs{n^{-1}\log\abs{\bm{Q}_{\hat{C}_{\perp}}^T \bm{W}\left( \bm{\theta}\right)\bm{Q}_{\hat{C}_{\perp}}} - n^{-1}\log\abs{\bm{V}\left( \bm{\theta}\right)}} = O_P\left( n^{-1}\right),
\end{align*}
\eqref{equation:lnglng.hat} follows. Therefore,
\begin{align*}
\sup_{\bm{\theta} \in \Theta_*} \abs{\hat{l}_{n,g}\left( \bm{\theta}\right) - f_{g}\left( \bm{\theta}\right)} = o_P(1),
\end{align*}
meaning $\hat{\bm{v}}_g$, the solution to \eqref{equation:Argmax.lnghat}, satisfies $\norm{\hat{\bm{v}}_g - \bm{v}_g}_2 = o_P(1)$. For $\hat{\bm{\Gamma}}_g$ defined as it was in the proof of Lemma \ref{lemma:Vhat} (except applied to $\hat{\bm{v}}_g$ instead of $\hat{\bm{\tau}}$),
\begin{align*}
\bm{0} &= \hat{\bm{\Gamma}}_g^T\nabla \hat{l}_{n,g}\left( \hat{\bm{v}}_g\right) = \hat{\bm{\Gamma}}_g^T\nabla \hat{l}_{n,g}\left(\bm{v}_g\right) + \hat{\bm{\Gamma}}_g^T\nabla^2 \hat{l}_{n,g}\left(\bm{v}_g\right) \left( \hat{\bm{v}}_g - \bm{v}_g\right) + o_P\left( \norm{\hat{\bm{v}}_g - \bm{v}_g}_2\right)\\
&= \hat{\bm{\Gamma}}_g^T\nabla l_{n,g}\left( \bm{v}_g\right) + \hat{\bm{\Gamma}}_g^T\nabla^2 l_{n,g}\left(\bm{v}_g\right)\hat{\bm{\Gamma}}_g \left( \hat{\bm{\theta}}_g - \bm{\theta}_g\right) + O_P\left( \eta\right) + o_P\left( \norm{\hat{\bm{v}}_g - \bm{v}_g}_2\right)
\end{align*}
where $\hat{\bm{\Gamma}}_g \left( \hat{\bm{\theta}}_g - \bm{\theta}_g\right) = \hat{\bm{v}}_g - \bm{v}_g$ with probability tending to 1, since $\bm{v}_g$ lies in the column space of $\hat{\bm{\Gamma}}_g$ with probability tending to 1. This proves that
\begin{align*}
\norm{\bm{V}_g - \hat{\bm{V}}_g}_2 = O_P\left\lbrace n^{1/2}\left( \gamma_K p\right)^{-1/2} + n^{-1/2}\right\rbrace
\end{align*}
(and therefore \eqref{equation:VgAsym}), since $\norm{\nabla l_{n,g}\left(\bm{v}_g\right)}_2 = O_P\left( n^{-1/2}\right)$ and the eigenvalues of $\nabla^2 l_{n,g}\left(\bm{v}_g\right)$ are bounded above 0 (and below $\infty$) with probability tending to 1. We can then express $\hat{\bm{V}}_g^{-1}$ and $\hat{\bm{W}}_g$ as
\begin{subequations}
\label{equation:VghatInv}
\begin{align}
&\norm{\hat{\bm{V}}_g^{-1} - \left( \bm{V}_g^{-1} + \sum\limits_{j=1}^b \epsilon_{g,j} \bm{V}_g^{-1} \bm{B}_j \bm{V}_g^{-1}\right)}_2 = o_P\left( n^{-1/2}\right)\\
&\norm{\hat{\bm{W}}_g^{-1} - \left( \bm{W}_g^{-1} + \sum\limits_{j=1}^b \epsilon_{g,j} \bm{W}_g^{-1} \tilde{\bm{B}}_j \bm{W}_g^{-1}\right)}_2 = o_P\left( n^{-1/2}\right)\\
&\epsilon_{g,j} = v_{g,j} - \hat{v}_{g,j},
\end{align}
\end{subequations}
which we will use to prove \eqref{equation:BetaAsym}.\par
\indent Define $\hat{\bm{D}} = \begin{bmatrix}
\bm{X} & \hat{\bm{C}}
\end{bmatrix}$. The generalized least squares estimate for $\bm{\beta}_g$ is
\begin{align}
\hat{\bm{\beta}}_g &= \left[ \left(\hat{\bm{D}}^T \hat{\bm{V}}_g^{-1} \hat{\bm{D}}\right)^{-1} \hat{\bm{D}}^T \hat{\bm{V}}_g^{-1} \bm{y}_g\right]_{1:d} = \left( \bm{X}^T \hat{\bm{V}}_g^{-1}\bm{X}\right)^{-1}\bm{X}^T \hat{\bm{V}}_g^{-1}\bm{y}_g - \hat{\bm{\Omega}}_g \hat{\bm{\ell}}_g\nonumber\\
\label{equation:BetaGLSExact}
&= \bm{\beta}_g + \bm{\Omega}_g \bm{\ell}_g - \hat{\bm{\Omega}}_g \hat{\bm{\ell}}_g + \left( \bm{X}^T \hat{\bm{V}}_g^{-1}\bm{X}\right)^{-1}\bm{X}^T \hat{\bm{V}}_g^{-1}\bm{e}_g
\end{align}
where here
\begin{align*}
\hat{\bm{\ell}}_g &= \left( \hat{\bm{C}}_{\perp}^T \hat{\bm{W}}_g^{-1}\hat{\bm{C}}_{\perp}\right)^{-1}\hat{\bm{C}}_{\perp}^T\hat{\bm{W}}_g^{-1} \bm{Q}_X^T \bm{y}_g\\
\hat{\bm{\Omega}}_g &= \left( \bm{X}^T \hat{\bm{V}}_g^{-1}\bm{X}\right)^{-1}\bm{X}^T \hat{\bm{V}}_g^{-1} \hat{\bm{C}}\\
\bm{\Omega}_g &= \left( \bm{X}^T \hat{\bm{V}}_g^{-1}\bm{X}\right)^{-1}\bm{X}^T \hat{\bm{V}}_g^{-1} \bm{C}.
\end{align*}
To prove the theorem, we will prove two relations in lemmas \ref{lemma:theorem:CGLS_L} and \ref{lemma:theorem:CGLS_Omega}:
\begin{align}
\label{equation:GLS:lDist}
&n^{1/2}\left\lbrace \bm{\Omega}_g\left( \bm{\ell}_g - \hat{\bm{\ell}}_g\right) + \left( \bm{X}^T \hat{\bm{V}}_g^{-1}\bm{X}\right)^{-1}\bm{X}^T \hat{\bm{V}}_g^{-1}\bm{e}_g\right\rbrace \edist \bm{F} + o_P(1)\\
\label{equation:GLS:OmegaAsym}
&n^{1/2}\norm{\hat{\bm{\Omega}}_g - \bm{\Omega}_g}_2 = o_P(1)
\end{align}
where $\bm{F} \sim N\left\lbrace \bm{0}, \left( n^{-1}\bm{X}^T \bm{V}_g^{-1} \bm{X}\right)^{-1} + \bm{\Omega}_g \left(n^{-1}\bm{C}_{\perp}^T \bm{W}_g^{-1}\bm{C}_{\perp} \right)^{-1}\bm{\Omega}_g^T\right\rbrace$.

\begin{lemma}
\label{lemma:theorem:CGLS_L}
\it{Under the assumptions of Theorem \ref{theorem:CGLS}, relation \eqref{equation:GLS:lDist} holds.}
\end{lemma}
\begin{proof}
To prove this we need to understand how $\hat{\bm{\ell}}_g$ behaves. First,
\begin{align*}
\hat{\bm{\ell}}_g =& \left( \hat{\bm{C}}_{\perp}^T \hat{\bm{W}}_g^{-1}\hat{\bm{C}}_{\perp}\right)^{-1}\hat{\bm{C}}_{\perp}^T\hat{\bm{W}}_g^{-1} \bm{Q}_X^T \bm{y}_g\\
 =& \left( \hat{\bm{C}}_{\perp}^T \hat{\bm{W}}_g^{-1}\hat{\bm{C}}_{\perp}\right)^{-1} \hat{\bm{C}}_{\perp}^T\hat{\bm{W}}_g^{-1} \bm{C}_{\perp}\bm{\ell}_g + \left( \hat{\bm{C}}_{\perp}^T \hat{\bm{W}}_g^{-1}\hat{\bm{C}}_{\perp}\right)^{-1} \hat{\bm{C}}_{\perp}^T\hat{\bm{W}}_g^{-1} \bm{Q}_X^T \bm{e}_g\\
=& \bm{\ell}_g + \left( \bm{C}_{\perp}^T \hat{\bm{W}}_g^{-1} \bm{C}_{\perp}\right)^{-1} \hat{\bm{C}}_{\perp}^T\hat{\bm{W}}_g^{-1} \bm{Q}_X^T \bm{e}_g + o_P\left( n^{-1/2}\right).
\end{align*}
The last equality follows from Corollary \ref{corollary:ConfSub:CtC} and because
\begin{align*}
\left( n^{-1}\hat{\bm{C}}_{\perp}^T \hat{\bm{W}}_g^{-1}\hat{\bm{C}}_{\perp}\right)^{-1} =& \left( n^{-1}\hat{\bm{C}}_{\perp}^T \bm{W}_g^{-1} \hat{\bm{C}}_{\perp} + \sum\limits_{j=1}^b \epsilon_{g,j} n^{-1}\hat{\bm{C}}_{\perp}^T \bm{W}_g^{-1} \tilde{\bm{B}}_j \bm{W}_g^{-1}\hat{\bm{C}}_{\perp} \right)^{-1} + o_P\left( n^{-1/2}\right)\\
=& \left( n^{-1}\hat{\bm{C}}_{\perp}^T \bm{W}_g^{-1} \hat{\bm{C}}_{\perp} \right)^{-1} + \sum\limits_{j=1}^b \epsilon_{g,j}\left\lbrace\left( n^{-1}\hat{\bm{C}}_{\perp}^T \bm{W}_g^{-1} \hat{\bm{C}}_{\perp} \right)^{-1}\left( n^{-1}\hat{\bm{C}}_{\perp}^T \bm{W}_g^{-1} \tilde{\bm{B}}_j \bm{W}_g^{-1}\hat{\bm{C}}_{\perp} \right)\times \cdots\right.\\
&\times \left.\left( n^{-1}\hat{\bm{C}}_{\perp}^T \bm{W}_g^{-1} \hat{\bm{C}}_{\perp} \right)^{-1}\right\rbrace + o_P\left( n^{-1/2}\right)\\
=& \left( n^{-1}\bm{C}_{\perp}^T \bm{W}_g^{-1} \bm{C}_{\perp} \right)^{-1} + \sum\limits_{j=1}^b \epsilon_{g,j}\left\lbrace\left( n^{-1}\bm{C}_{\perp}^T \bm{W}_g^{-1} \bm{C}_{\perp} \right)^{-1}\left( n^{-1}\bm{C}_{\perp}^T \bm{W}_g^{-1} \tilde{\bm{B}}_j \bm{W}_g^{-1}\bm{C}_{\perp} \right)\times \cdots\right.\\
&\times \left.\left( n^{-1}\bm{C}_{\perp}^T \bm{W}_g^{-1} \bm{C}_{\perp} \right)^{-1}\right\rbrace + o_P\left( n^{-1/2}\right) = \left( n^{-1}\bm{C}_{\perp}^T \hat{\bm{W}}_g^{-1}\bm{C}_{\perp}\right)^{-1} + o_P\left( n^{-1/2}\right).
\end{align*}
Next, $\bm{Q}_X^T \bm{e}_g$ is independent of $\bm{X}^T \bm{V}_g^{-1}\bm{e}_g$ by Craig's Theorem and
\begin{align*}
n^{-1}\bm{X}^T \hat{\bm{V}}_g^{-1}\bm{e}_g = n^{-1}\bm{X}^T\bm{V}_g^{-1}\bm{e}_g + \sum\limits_{j=0}^b \epsilon_{g,j} n^{-1}\bm{X}^T\bm{V}_g^{-1} \bm{B}_j \bm{V}_g^{-1}\bm{e}_g + o_P\left( n^{-1/2}\right) = n^{-1}\bm{X}^T\bm{V}_g^{-1}\bm{e}_g + o_P\left( n^{-1/2}\right).
\end{align*}
Therefore, the proof will be complete if we can show
\begin{align*}
\norm{n^{-1}\hat{\bm{C}}_{\perp}^T\hat{\bm{W}}_g^{-1} \bm{Q}_X^T \bm{e}_g - n^{-1}\bm{C}_{\perp}^T \bm{W}_g^{-1} \bm{Q}_X^T \bm{e}_g}_2 = o_P\left( n^{-1/2}\right),
\end{align*}
since a simple application of Slutsky's Theorem would give us the result. To show this, we first note that
\begin{align*}
\norm{n^{-1}\hat{\bm{C}}_{\perp}^T\hat{\bm{W}}_g^{-1} \bm{Q}_X^T \bm{e}_g - n^{-1}\bm{C}_{\perp}^T \bm{W}_g^{-1} \bm{Q}_X^T \bm{e}_g}_2 \leq &\norm{n^{-1}\hat{\bm{C}}_{\perp}^T\bm{W}_g^{-1} \bm{Q}_X^T \bm{e}_g - n^{-1}\bm{C}_{\perp}^T \bm{W}_g^{-1} \bm{Q}_X^T \bm{e}_g}_2\\
&+ \sum\limits_{j=0}^b \epsilon_{g,j}\norm{n^{-1}\hat{\bm{C}}_{\perp}^T\bm{W}_g^{-1} \tilde{\bm{B}}_j \bm{W}_g^{-1}\bm{Q}_X^T \bm{e}_g - n^{-1}\bm{C}_{\perp}^T\bm{W}_g^{-1} \tilde{\bm{B}}_j \bm{W}_g^{-1}\bm{Q}_X^T \bm{e}_g}_2\\
& + o_P\left( n^{-1/2}\right).
\end{align*}
We will first prove
\begin{align}
\label{equation:lemma:ToShow}
\norm{n^{-1}\hat{\bm{C}}_{\perp}^T\bm{W}_g^{-1} \bm{Q}_X^T \bm{e}_g - n^{-1}\bm{C}_{\perp}^T \bm{W}_g^{-1} \bm{Q}_X^T \bm{e}_g}_2 = o_P\left( n^{-1/2}\right)
\end{align}
and then an identical analysis can be used to show
\begin{align*}
\norm{n^{-1}\hat{\bm{C}}_{\perp}^T\bm{W}_g^{-1} \tilde{\bm{B}}_j \bm{W}_g^{-1}\bm{Q}_X^T \bm{e}_g - n^{-1}\bm{C}_{\perp}^T\bm{W}_g^{-1} \tilde{\bm{B}}_j \bm{W}_g^{-1}\bm{Q}_X^T \bm{e}_g}_2 = o_P\left( n^{-1/2}\right).
\end{align*}
To prove \eqref{equation:lemma:ToShow}, we first define
\begin{align*}
\bar{\bm{C}}_{\perp} &= (n-d)^{-1/2}\hat{\bm{W}}^{-1/2}\bm{C}_{\perp}\\
\bar{\bm{Q}}_{\perp} &= \bm{Q}_{\bar{C}_{\perp}}.
\end{align*}
Then using \eqref{equation:Eigsv}, \eqref{equation:Eigsz} and 
\eqref{equation:corrPCA:EigsResults} in the proof of Theorem \ref{theorem:CorrPCA}, we have that
\begin{align}
\label{equation:CperpHat}
\hat{\bm{C}}_{\perp} &= (n-d)^{1/2}\hat{\bm{W}}^{1/2}\hat{\bar{\bm{C}}}_{\perp} = (n-d)^{1/2}\hat{\bm{W}}^{1/2} \bar{\bm{C}}_{\perp}\hat{\bm{v}} + (n-d)^{1/2}\hat{\bm{W}}^{1/2} \bar{\bm{Q}}_{\perp} \hat{\bm{z}}.
\end{align}
We first see that
\begin{align*}
n^{-1}(n-d)^{1/2} \hat{\bm{v}}^T \bar{\bm{C}}_{\perp}^T \hat{\bm{W}}^{1/2}\bm{W}_g^{-1} \bm{Q}_X^T \bm{e}_g &= n^{-1}(n-d)^{1/2} \bar{\bm{C}}_{\perp}^T \hat{\bm{W}}^{1/2}\bm{W}_g^{-1} \bm{Q}_X^T \bm{e}_g + o_P\left( n^{-1/2}\right)\\
&= n^{-1}\bm{C}_{\perp}^T \bm{W}_g^{-1} \bm{Q}_X^T \bm{e}_g + o_P\left( n^{-1/2}\right).
\end{align*}
To then show \eqref{equation:lemma:ToShow} and complete the proof, we need only show that
\begin{align*}
\norm{\hat{\bm{z}}^T \bar{\bm{Q}}_{\perp}^T \hat{\bm{W}}^{1/2}\bm{W}_g^{-1} \bm{Q}_X^T \bm{e}_g}_2 = o_P\left( 1\right).
\end{align*}
$\hat{\bm{z}}$ is such that
\begin{align*}
\norm{\hat{\bm{z}} - \left(p \gamma_K \right)^{-1/2}\bar{\bm{Q}}_{\perp}^T\hat{\bm{W}}^{-1/2}\tilde{\bm{E}}^T \left\lbrace \left(p \gamma_K \right)^{-1/2} \tilde{\bm{E}}\hat{\bm{W}}^{-1/2}\bar{\bm{C}}_{\perp} + \bar{\bm{L}}\right\rbrace}_2 = o_P\left( n^{-1/2}\right)
\end{align*}
where $\tilde{\bm{E}} = \bm{E}\bm{Q}_X$. Therefore,
\begin{align*}
\hat{\bm{z}}^T \bar{\bm{Q}}_{\perp}^T \hat{\bm{W}}^{1/2}\bm{W}_g^{-1} \bm{Q}_X^T \bm{e}_g =& \underbrace{\left(p \gamma_K \right)^{-1/2} \bar{\bm{L}}^T \tilde{\bm{E}} \bm{Q}_{C_{\perp}} \left(\bm{Q}_{C_{\perp}}^T \hat{\bm{W}}\bm{Q}_{C_{\perp}} \right)^{-1}\bm{Q}_{C_{\perp}}^T \hat{\bm{W}} \bm{W}_g^{-1} \tilde{\bm{e}}_g}_{(1)}\\
&+ \underbrace{\gamma_K^{-1} \left( n-d\right)^{-1/2} \bm{C}_{\perp}^T \hat{\bm{W}}^{-1} \left( p^{-1}\tilde{\bm{E}}^T \tilde{\bm{E}}\right)\bm{Q}_{C_{\perp}} \left(\bm{Q}_{C_{\perp}}^T \hat{\bm{W}}\bm{Q}_{C_{\perp}} \right)^{-1}\bm{Q}_{C_{\perp}}^T \hat{\bm{W}} \bm{W}_g^{-1} \tilde{\bm{e}}_g}_{(2)}\\
&+ o_P(1).
\end{align*}
We now go through each one of these terms.
\begin{enumerate}[label=(\arabic*)]
\item Since this is a continuously differentiable function of $\hat{\bm{W}}$ and $\norm{\hat{\bm{W}} - \bm{W}^*}_2 = O_P\left( n^{-1}\right)$,
\begin{align*}
&\norm{\left(p \gamma_K \right)^{-1/2} \bar{\bm{L}}^T \tilde{\bm{E}} \bm{Q}_{C_{\perp}} \left(\bm{Q}_{C_{\perp}}^T \hat{\bm{W}}\bm{Q}_{C_{\perp}} \right)^{-1}\bm{Q}_{C_{\perp}}^T \hat{\bm{W}} \bm{W}_g^{-1} \tilde{\bm{e}}_g - \left(p \gamma_K \right)^{-1/2} \bar{\bm{L}}^T \tilde{\bm{E}} \bm{Q}_{C_{\perp}} \left(\bm{Q}_{C_{\perp}}^T \bm{W}^*\bm{Q}_{C_{\perp}} \right)^{-1}\bm{Q}_{C_{\perp}}^T \bm{W}^* \bm{W}_g^{-1} \tilde{\bm{e}}_g}_2\\
 = &O_P\left\lbrace n^{1/2}\left(p \gamma_K \right)^{-1/2}\right\rbrace = o_P(1).
\end{align*}
Define the non-random matrix
\begin{align*}
\bm{M}_g = \bm{Q}_{C_{\perp}} \left(\bm{Q}_{C_{\perp}}^T \bm{W}^*\bm{Q}_{C_{\perp}} \right)^{-1}\bm{Q}_{C_{\perp}}^T \bm{W}^* \bm{W}_g^{-1} \in \mathbb{R}^{(n-d) \times (n-d)}.
\end{align*}
and the random vector $\bm{w}_g = \bm{M}_g \tilde{\bm{e}}_g \sim N\left( 0,\bm{M}_g \bm{W}_g \bm{M}_g^T\right)$. Then
\begin{align*}
\left(p \gamma_K \right)^{-1/2} \bar{\bm{L}}^T \tilde{\bm{E}} \bm{w}_g = \left(p \gamma_K \right)^{-1/2} \sum\limits_{h \neq g}\bar{\bm{\ell}}_h \tilde{\bm{e}}_h^T \bm{w}_g + \underbrace{\left(p \gamma_K \right)^{-1/2} \bar{\bm{\ell}}_g \tilde{\bm{e}}_g^T \bm{w}_g}_{O_P\left\lbrace n^{3/2} \left( \gamma_K p\right)^{-1} \right\rbrace = o_P(1)}
\end{align*}
with
\begin{align*}
\E \left\lbrace \left(p \gamma_K \right)^{-1/2} \sum\limits_{h \neq g}\bar{\bm{\ell}}_h \tilde{\bm{e}}_h^T \bm{w}_g\right\rbrace &= \E \left\lbrace \left(p \gamma_K \right)^{-1/2} \sum\limits_{h \neq g}\bar{\bm{\ell}}_h \tilde{\bm{e}}_h^T \bm{w}_g \mid \bm{w}_g\right\rbrace = \bm{0}\\
\V\left\lbrace \left(p \gamma_K \right)^{-1/2} \sum\limits_{h \neq g}\bar{\bm{\ell}}_h \tilde{\bm{e}}_h^T \bm{w}_g \mid \bm{w}_g\right\rbrace &= \left(p \gamma_K \right)^{-1}\sum\limits_{h \neq g} \V\left( \bar{\bm{\ell}}_h \tilde{\bm{e}}_h^T \bm{w}_g \mid \bm{w}_g\right) = \frac{1}{\lambda p} \sum\limits_{h \neq g} \left( \bm{w}_g^T \bm{W}_h \bm{w}_g\right)\bar{\bm{\ell}}_h \bar{\bm{\ell}}_h^T \\
& \preceq c \left(p \gamma_K \right)^{-1} \norm{\bm{w}_g}_2^2\bar{\bm{L}}^T \bar{\bm{L}}
\end{align*}
where $c$ bounds the eigenvalues of $\bm{W}_h$ ($h=1,\ldots,p$) from above. Therefore,
\begin{align*}
\V\left\lbrace \left(p \gamma_K \right)^{-1/2} \sum\limits_{h \neq g}\bar{\bm{\ell}}_h \tilde{\bm{e}}_h^T \bm{w}_g\right\rbrace = O\left\lbrace n \left(p \gamma_K \right)^{-1}\right\rbrace = o(1).
\end{align*}
This shows that $\norm{1.)}_2 = o_P(1)$.

\item For this, we use the same technique to replace $\hat{\bm{W}}$ with $\bm{W}^*$, which only differs from 2.) by $O_P\left(n^{-1/2} \right) = o_P(1)$ in 2-norm. Next, we define non-random the matrix $\bm{A}\in \mathbb{R}^{(n-d) \times K}$ to be a matrix whose columns form an orthonormal basis for the column space of $\bm{C}_{\perp}$. For $\bm{w}_g$ defined above, note that $\bm{A}^T \bm{w}_g = \bm{0}$. Therefore,
\begin{align*}
2.) = \gamma_K^{-1}\bm{A}^T \bm{W}^{*-1}\left( p^{-1}\tilde{\bm{E}}^T \tilde{\bm{E}}\right)\bm{w}_g = \gamma_K^{-1}\bm{A}^T\bm{W}^{*-1}\left( p^{-1}\tilde{\bm{E}}_{-g}^T \tilde{\bm{E}}_{-g}\right)\bm{w}_g + \underbrace{\left( \gamma_K p\right)^{-1}\bm{A}^T \tilde{\bm{e}}_g \tilde{\bm{e}}_g^T \bm{w}_g}_{O_P\left\lbrace n \left( \gamma_K p\right)^{-1}\right\rbrace = o_P(1)}
\end{align*}
where $\tilde{\bm{E}}_{-g}$ is the $\tilde{\bm{E}}$ with the $g^{\text{th}}$ row removed. Note that
\begin{align*}
\E \left\lbrace \bm{A}^T\bm{W}^{*-1}\left( p^{-1}\tilde{\bm{E}}_{-g}^T \tilde{\bm{E}}_{-g}\right)\bm{w}_g \mid \bm{w}_g\right\rbrace &= p^{-1}\bm{A}^T \bm{W}^{*-1} \bm{W}_g \bm{w}_g\\
&\sim p^{-1} N_{K}\left( \bm{0}, \bm{A}^T \bm{W}^{*-1} \bm{W}_g \bm{M}_g \bm{W}_g \bm{M}_g \bm{W}_g \bm{W}^{*-1} \bm{A}\right)
\end{align*}
Since $K$ is fixed and finite, it suffices to assume $K = 1$. If we let $\bm{u} = \bm{W}^{*-1}\bm{A} \in \mathbb{R}^{(n-d) \times K}$ (which has bounded 2-norm), then
\begin{align*}
\V\left\lbrace \bm{A}^T\bm{W}^{*-1}\left( p^{-1}\tilde{\bm{E}}_{-g}^T \tilde{\bm{E}}_{-g}\right)\bm{w}_g \mid \bm{w}_g \right\rbrace &= p^{-2} \sum\limits_{h \neq g} \V\left\lbrace \bm{u}^T \tilde{\bm{e}}_h \tilde{\bm{e}}_h^T \bm{w}_g \mid \bm{w}_g\right\rbrace\\
&= p^{-2} O\left( p \norm{\bm{w}_g}_2^2\right) = O\left( p^{-1}\norm{\bm{w}_g}_2^2 \right).
\end{align*}
This shows that $\norm{2.)}_2 = o_P(1)$, and completes the proof.
\end{enumerate}
\end{proof}

\begin{lemma}
\label{lemma:theorem:CGLS_Omega}
\it{Under the assumptions Theorem \ref{theorem:CGLS}, relation \eqref{equation:GLS:OmegaAsym} holds.}
\end{lemma}
\begin{proof}
From the expression for $\bm{C}$ and $\hat{\bm{C}}$ in item \ref{item:EstimationProcedure:Omega} of Algorithm \ref{algorithm:EsimationProcedure} and equation \eqref{equation:Cestimate}, we can write $\bm{\Omega}_g$ and $\hat{\bm{\Omega}}_g$ as
\begin{align*}
\bm{\Omega}_g &= \left( \bm{X}^T \hat{\bm{V}}_g^{-1}\bm{X}\right)^{-1}\bm{X}^T \hat{\bm{V}}_g^{-1}\bm{C} = \bm{\Omega} + \left( \bm{X}^T \hat{\bm{V}}_g^{-1}\bm{X}\right)^{-1}\bm{X}^T \hat{\bm{V}}_g^{-1} \hat{\bm{V}}\bm{Q}_{\bm{X}}\hat{\bm{W}}^{-1} \bm{C}_{\perp}\\
\hat{\bm{\Omega}}_g &= \left( \bm{X}^T \hat{\bm{V}}_g^{-1}\bm{X}\right)^{-1}\bm{X}^T \hat{\bm{V}}_g^{-1}\hat{\bm{C}} = \hat{\bm{\Omega}} + \left( \bm{X}^T \hat{\bm{V}}_g^{-1}\bm{X}\right)^{-1}\bm{X}^T \hat{\bm{V}}_g^{-1} \hat{\bm{V}}\bm{Q}_{\bm{X}}\hat{\bm{W}}^{-1} \hat{\bm{C}}_{\perp}.
\end{align*}
By Lemma \ref{lemma:OmegaWLS}, $\norm{\bm{\Omega} - \hat{\bm{\Omega}}}_2 = o_P\left( n^{-1/2}\right)$. Therefore, to prove the current lemma, we need only show that
\begin{align*}
\norm{n^{-1}\bm{X}^T \hat{\bm{V}}_g^{-1} \hat{\bm{V}}\bm{Q}_{\bm{X}}\hat{\bm{W}}^{-1} \bm{C}_{\perp} - n^{-1}\bm{X}^T \hat{\bm{V}}_g^{-1} \hat{\bm{V}}\bm{Q}_{\bm{X}}\hat{\bm{W}}^{-1} \hat{\bm{C}}_{\perp}}_2 = o_P\left( n^{-1/2}\right).
\end{align*}
Just as we did in the proof of Lemma \ref{lemma:theorem:CGLS_L}, we use \eqref{equation:VghatInv} to expand the above equation:
\begin{align*}
&\norm{n^{-1}\bm{X}^T \hat{\bm{V}}_g^{-1} \hat{\bm{V}}\bm{Q}_{\bm{X}}\hat{\bm{W}}^{-1} \bm{C}_{\perp} - n^{-1}\bm{X}^T \hat{\bm{V}}_g^{-1} \hat{\bm{V}}\bm{Q}_{\bm{X}}\hat{\bm{W}}^{-1} \hat{\bm{C}}_{\perp}}_2 \leq \\
& \norm{n^{-1}\bm{X}^T \bm{V}_g^{-1} \hat{\bm{V}}\bm{Q}_{\bm{X}}\hat{\bm{W}}^{-1} \bm{C}_{\perp} - n^{-1}\bm{X}^T \bm{V}_g^{-1} \hat{\bm{V}}\bm{Q}_{\bm{X}}\hat{\bm{W}}^{-1} \hat{\bm{C}}_{\perp}}_2 \\
+ &\sum\limits_{j=1}^b \epsilon_{g,j} \norm{n^{-1}\bm{X}^T \bm{V}_g^{-1}\bm{B}_j \bm{V}_g^{-1}\hat{\bm{V}}\bm{Q}_{\bm{X}}\hat{\bm{W}}^{-1} \bm{C}_{\perp} - n^{-1}\bm{X}^T \bm{V}_g^{-1}\bm{B}_j \bm{V}_g^{-1} \hat{\bm{V}}\bm{Q}_{\bm{X}}\hat{\bm{W}}^{-1} \hat{\bm{C}}_{\perp}}_2 + o_P\left( n^{-1/2}\right)
\end{align*}
where again we will first show
\begin{align}
\label{equation:lemma:Omega}
\norm{n^{-1}\bm{X}^T \bm{V}_g^{-1} \hat{\bm{V}}\bm{Q}_{\bm{X}}\hat{\bm{W}}^{-1} \bm{C}_{\perp} - n^{-1}\bm{X}^T \bm{V}_g^{-1} \hat{\bm{V}}\bm{Q}_{\bm{X}}\hat{\bm{W}}^{-1} \hat{\bm{C}}_{\perp}}_2 = o_P\left( n^{-1/2}\right).
\end{align}
An identical argument can then be applied to show that
\begin{align*}
\norm{n^{-1}\bm{X}^T \bm{V}_g^{-1}\bm{B}_j \bm{V}_g^{-1}\hat{\bm{V}}\bm{Q}_{\bm{X}}\hat{\bm{W}}^{-1} \bm{C}_{\perp} - n^{-1}\bm{X}^T \bm{V}_g^{-1}\bm{B}_j \bm{V}_g^{-1} \hat{\bm{V}}\bm{Q}_{\bm{X}}\hat{\bm{W}}^{-1} \hat{\bm{C}}_{\perp}}_2 = o_P\left( n^{-1/2}\right).
\end{align*}
\indent First, by the expression for $\hat{\bm{C}}_{\perp}$ in \eqref{equation:CperpHat},
\begin{align*}
n^{-1}\bm{X}^T \bm{V}_g^{-1} \hat{\bm{V}}\bm{Q}_{\bm{X}}\hat{\bm{W}}^{-1} \hat{\bm{C}}_{\perp} =& n^{-1}\left(n-d\right)^{1/2}\bm{X}^T \bm{V}_g^{-1} \hat{\bm{V}}\bm{Q}_{\bm{X}}\hat{\bm{W}}^{-1/2} \bar{\bm{C}}_{\perp} \hat{\bm{v}}\\
&+ n^{-1}\left(n-d\right)^{1/2}\bm{X}^T \bm{V}_g^{-1} \hat{\bm{V}}\bm{Q}_{\bm{X}}\hat{\bm{W}}^{-1/2} \bar{\bm{Q}}_{\perp} \hat{\bm{z}}
\end{align*}
To show \eqref{equation:lemma:Omega},
\begin{align*}
\norm{n^{-1}\bm{X}^T \bm{V}_g^{-1} \hat{\bm{V}}\bm{Q}_{\bm{X}}\hat{\bm{W}}^{-1} \bm{C}_{\perp} - n^{-1}\left(n-d\right)^{1/2}\bm{X}^T \bm{V}_g^{-1} \hat{\bm{V}}\bm{Q}_{\bm{X}}\hat{\bm{W}}^{-1/2} \bar{\bm{C}}_{\perp} \hat{\bm{v}}}_2 = o_P\left( n^{-1/2}\right)
\end{align*}
by \eqref{equation:CperpHat} and because $\norm{\hat{\bm{v}} - I_K}_2 = o_P\left( n^{-1/2} \right)$. Therefore, we need only show that
\begin{align*}
\norm{\bm{u}^T \hat{\bm{V}}\bm{Q}_{\bm{X}}\hat{\bm{W}}^{-1/2} \bar{\bm{Q}}_{\perp} \hat{\bm{z}}}_2 = o_P\left( n^{-1/2}\right)
\end{align*}
where $\bm{u} = n^{-1/2}\bm{V}_g^{-1} \bm{X} \in \mathbb{R}^{n \times d}$ has bounded 2-norm. By the expression for $\hat{\bm{z}}$ in \eqref{equation:Eigsz},
\begin{align*}
\bm{u}^T \hat{\bm{V}}\bm{Q}_{\bm{X}}\hat{\bm{W}}^{-1/2} \bar{\bm{Q}}_{\perp} \hat{\bm{z}} =& \left( \gamma_K p\right)^{-1/2}\bm{u}^T \hat{\bm{V}}\bm{Q}_{\bm{X}}\bm{Q}_{\bm{C}_{\perp}}\left( \bm{Q}_{\bm{C}_{\perp}}^T \hat{\bm{W}} \bm{Q}_{\bm{C}_{\perp}}\right)^{-1} \bm{Q}_{\bm{C}_{\perp}}^T \tilde{\bm{E}}^T \bar{\bm{L}}\\
&+ \gamma_K^{-1}(n-d)^{-1/2}\bm{u}^T \hat{\bm{V}}\bm{Q}_{\bm{X}}\bm{Q}_{\bm{C}_{\perp}}\left( \bm{Q}_{\bm{C}_{\perp}}^T \hat{\bm{W}} \bm{Q}_{\bm{C}_{\perp}}\right)^{-1} \bm{Q}_{\bm{C}_{\perp}}^T \left(p^{-1}\tilde{\bm{E}}^T \tilde{\bm{E}} \right)\hat{\bm{W}}^{-1}\bm{C}_{\perp}\\
&+ o_P\left( n^{-1/2}\right)
\end{align*}
where $\tilde{\bm{E}} = \bm{E}\bm{Q}_{\bm{X}}$. Since $\norm{\hat{\bm{V}} - \bm{V}^*}_2, \norm{\hat{\bm{W}} - \bm{W}^*}_2 = O_P\left( n^{-1}\right)$, we can use identical techniques used in the proof of Lemma \ref{lemma:theorem:CGLS_L} to show that
\begin{align*}
\norm{\bm{u}^T \hat{\bm{V}}\bm{Q}_{\bm{X}}\hat{\bm{W}}^{-1/2} \bar{\bm{Q}}_{\perp} \hat{\bm{z}}}_2 = O_P\left\lbrace n^{-1} + \left( \gamma_K p\right)^{-1/2}\right\rbrace = o_P\left( n^{-1/2}\right)
\end{align*}
which completes the proof.
\end{proof}

\indent Going back to the expression for $\hat{\bm{\beta}}_g$ in \eqref{equation:BetaGLSExact}, Lemmas \ref{lemma:theorem:CGLS_L} and \ref{lemma:theorem:CGLS_Omega} and the fact that $\norm{\hat{\bm{V}}_g - \bm{V}_g}_2 = o_P(1)$ shows that
\begin{align*}
n^{1/2}\left( \hat{\bm{\beta}}_g-\bm{\beta}_g \right) &= n^{1/2}\bm{\Omega}_g \left( \bm{\ell}_g - \hat{\bm{\ell}}_g \right) + n^{1/2}\left( \bm{\Omega}_g - \hat{\bm{\Omega}}_g \right)\hat{\bm{\ell}}_g + n^{1/2}\left( \bm{X}^T \bm{V}_g^{-1}\bm{X}\right)^{-1}\bm{X}^T \bm{V}_g^{-1} \bm{e}_g\\
&= \bm{F} + o_P(1)
\end{align*}
where
\begin{align*}
\bm{F} \sim N\left\lbrace \bm{0}, \left( n^{-1}\bm{X}^T \bm{V}_g^{-1}\bm{X}\right)^{-1} + \left( \bm{X}^T \bm{V}_g^{-1}\bm{X}\right)^{-1}\bm{X}^T \bm{V}_g^{-1} \bm{C} \left( n^{-1}\bm{C}_{\perp}^T \bm{W}_g^{-1}\bm{C}_{\perp} \right)^{-1}\bm{C}^T \bm{V}_g^{-1}\bm{X}\left( \bm{X}^T \bm{V}_g^{-1}\bm{X}\right)^{-1} \right\rbrace.
\end{align*}
This shows \eqref{equation:BetaAsym} and completes the proof.

\end{proof}

\end{document}